\documentclass{article}

\usepackage{arxiv}
\usepackage{graphicx}
\usepackage{subcaption}
\usepackage{caption}
\usepackage{amsthm}
\usepackage{mathtools}

\newtheorem{dfn}{Definition}
\newtheorem{thm}{Theorem}

\newtheorem{cond}{Condition}
\newtheorem{prop}{Proposition}
\newtheorem{lem}[thm]{Lemma}
\newtheorem{remark}{Remark}
\newcommand{\argmax}{\mathop{\rm arg~max}\limits}
\newcommand{\argmin}{\mathop{\rm arg~min}\limits}
\usepackage[utf8]{inputenc} 
\usepackage[T1]{fontenc}    
\usepackage{hyperref}       
\usepackage{url}            
\usepackage{booktabs}       
\usepackage{amsfonts}       
\usepackage{nicefrac}       
\usepackage{microtype}      
\usepackage{lipsum}
\usepackage{appendix}
\usepackage{amsmath}
\usepackage{amssymb}
\usepackage[square,numbers]{natbib}

\title{An adaptive strategy for sequential designs of multilevel computer experiments}

\author{
  Ayao Ehara\thanks{Corresponding author: ayao.ehara.18@ucl.ac.uk} \\
 Department of Statistical Science\\
  University College London\\
  London, UK \\
   \And
 Serge Guillas \\
 Department of Statistical Science\\
  University College London\\
  London, UK \\
}

\begin{document}
\maketitle

\begin{abstract}
Investigating uncertainties in computer simulations can be prohibitive in terms of computational costs, since the simulator needs to be run over a large number of input values. Building an emulator, i.e. a statistical surrogate model of the simulator constructed using a design of experiments made of a comparatively small number of evaluations of the forward solver, greatly alleviates the computational burden to carry out such investigations. Nevertheless, this can still be above the computational budget for many studies. Two major approaches have been used to reduce the budget needed to build the emulator: efficient design of experiments, such as sequential designs, and combining training data of different degrees of sophistication in a so-called multi-fidelity method, or multilevel in case these fidelities are ordered typically for increasing resolutions. We present here a novel method that combines both approaches, the multilevel adaptive sequential design of computer experiments (MLASCE) in the framework of Gaussian process (GP) emulators. We make use of reproducing kernel Hilbert spaces as a tool for our GP approximations of the differences between two consecutive levels. This dual strategy allows us to allocate efficiently limited computational resources over simulations of different levels of fidelity and build the GP emulator. The allocation of computational resources is shown to be the solution of a simple optimization problem in a special case where we theoretically prove the validity of our approach. Our proposed method is compared with other existing models of multi-fidelity Gaussian process emulation.  Gains in orders of magnitudes in accuracy or computing budgets are demonstrated in some of numerical examples for some settings.
\end{abstract}


\section{Introduction}
Complex mathematical models, implemented in a large computer simulation (e.g., numerical solvers of partial differential equations), have been used to study real systems in many areas of scientific research. For example, if the wave height at a fixed location in a tsunami simulation is a quantity of interest, we can simulate it by running the numerical solver of the governing equation and grasp how the target quantity changes with different values of the parameters in the system. These parameters, which determine the system, are typically treated as input variables of the underlying model. This type of analysis, implementing a number of runs of computer simulations with various inputs, is called a computer experiment \cite{sacks1989design}. Since running a complex model is computationally expensive, it is hard to have a complete understanding of this input-output structure; thereby, the need for the framework of surrogate emerges. We follow the framework of \cite{sacks1989design}. Their approach is to model the deterministic output from computer simulation as the realization of a stochastic process, thereby providing a statistical basis for designing experiments (choosing the inputs) for efficient predictions. Statistical emulators are constructed subsequently for the predictions of the input-output structure, and statistical emulators allow us a quick exhaustive sweep of the outputs over the entire domain of the input, avoiding the computational burden expended for evaluating a simulation at many different values of the input. The design problem or design of an experiment is to determine the locations of input variables and the corresponding output data. A stochastic process is fitted to the collected data, and Gaussian processes are employed to build a statistical emulator. Gaussian processes are suitable for interpolation in computer experiments and relatively easy to implement because of the form of conditional distributions. Besides, Gaussian processes can represent a wide variety of functions if the correlation structure is appropriately chosen. For these reasons, Gaussian processes or Gaussian process regression \cite{rasmussen2006gaussian,stein2012interpolation} (also called kriging) are a commonly used class of surrogate models, which assume that prior beliefs about the input-output structure can be modeled by a Gaussian process \cite{craig2001bayesian, kennedy2001bayesian}.

As mentioned, it is prohibitive to implement computer simulations numerous times because of the computational expense. But a computer simulation often can be run at different levels of accuracy, with versions ranging from the most sophisticated simulator to the most basic. For example, in the case of finite difference method used for solving a differential equation, we can set different levels of fidelity by changing the degree of discretization in the solver. The sophisticated versions are accurate but so expensive that only a limited numbers of runs can be obtained. In contrast, the low-level simulations return relatively inaccurate results but are cheaper to evaluate, and we can use a large number of samples. Then, it is a natural question of how the data from cheap computer simulations can effectively complement the small size of samples from expensive simulations. In other words, how can we use the data from several simulators with different levels of accuracy in the most effective way to construct a statistical emulator? In order to address this question, we develop a novel multi-fidelity Gaussian process regression and optimize the numbers of runs from simulators of different accuracy under a limited computational budget.

The first multi-fidelity statistical emulator proposed in \cite{craig1998constructing} was based on a linear regression formulation. Then, this model was improved in \citep{cumming2009small} by employing a Bayes linear formulation. The drawback is a lack of accuracy since they are based on a linear regression formulation. Following the fundamental framework of \cite{sacks1989design},  multi-fidelity Gaussian process regression, or co-kriging, was considered in a seminal paper in \cite{kennedy2000predicting} as a more general methodology of interpolation. \cite{higdon2004combining, reese2004integrated, qian2008bayesian} followed this formulation. \cite{ le2013bayesian,le2014recursive} articulated the formulation, and \cite{le2015cokriging} developed another strategy for obtaining training data as a strategy of sequential experimental design. \cite{perdikaris2017nonlinear} extended the scheme of co-kriging by assuming the deep GP formulation, the nonlinear autoregressive multi-fidelity GP
regression (NARGP), put forth by \cite{damianou2013deep,damianou2015deep}. However, \cite{perdikaris2017nonlinear} does not offer a strategy for experimental design. Recent papers such as \cite{perdikaris2015multi, yang2019physics, yang2020bifidelity} took advantage of this type of method. The benefit of co-kriging is that the correlation between the data of high and low fidelity is introduced so that it allows us to construct efficiently an accurate emulator by performing relatively few evaluations of an expensive high-fidelity model and more evaluations of a cheaper one. However, these approaches do not address the relationships between the precision of an emulator and the number of runs from each simulator, which is crucial to obtain the most accurate surrogate model under a limited computational budget, hence can lead to a large prediction errors. In addition, choosing the location of data, or a design of experiments, directly affects the quality of an emulator. Most of the related works use space-filling or Latin Hypercube Design (LHD), which can be inefficient, particularly in the case of high dimensions. Moreover, they assume nested design across every level of fidelity, which means that the simulated points for the more sophisticated computer simulations are included in the design points for less accurate ones. We exploit the correlation across levels but we also need to explore where higher fidelity simulations behave differently from the ones of low fidelity. Nested LHD may not allow this exploration and thus provide insufficient samplings to fit the differences across levels. This is one of the reasons why will employ sequential design in this paper. 

The contributions of this research are two-fold. First, our approach is the first multi-fidelity emulation method that is both free from the nested property of the design and takes advantage of sequential designs, which have been shown to be superior to one-shot designs for computer experiments \cite{beck2016sequential} and in multi-fidelity work \cite{le2015cokriging}. Secondly, compared to previous multi-fidelity methods, we allow our choice of design points across fidelities to be fully flexible. Our criterion selects both the input point and level of fidelity of the next run that results in the largest improvement of the overall prediction quality using all levels combined into one emulator. This dual strategy provides large gains compared to state-of-the-art multi-fidelity experimental designs. Defining the accuracy of an emulator and evaluating the effect of including additional evaluations of expensive computer simulations as the training data are necessary steps in optimising the usage of computational resources. This is made possible by decomposing the high resolution computer simulations into a telescoping sum of incremental differences across consecutive levels and assuming a reproducing kernel Hilbert space (RKHS) as a function space to which each incremental function belongs. Evaluating the accuracy of an approximation in the RKHS guarantees that our emulator gets closer to the underlining function as the number of runs increases; hence our methodology solves the problem of inaccuracy which is inherently included in the above-mentioned emulators. Regarding this type of asymptotic theory, profound literature exists in scattered data approximation and related fields, which is often used in Gaussian process approximation. The exhaustive study was given in \cite{wendland2004scattered}. In the context of practical implementation of Gaussian process regressions in recent papers, the effects of misspecified parameters in covariance functions were investigated in \cite{tuo2020kriging, wang2020prediction} while convergence rates with statistically estimated parameters were studied in \cite{teckentrup2020convergence, wynne2021convergence}. The convergence rates by these methods are based on geometric features of the domain and design. Especially in the case of high dimension and a tensor grid design, these convergence rates are quite slow. Therefore, we rely on an empirical but usually more efficient design strategy, Mutual Information Computer Experiment (MICE) \cite{beck2016sequential}, to obtain faster convergence in a practical way.

 This paper is organized as follows. In section \ref{sec:GP} we present a review of Gaussian process regression for emulation and existing multi-fidelity methods. We also provide an overview of our strategy of experimental design and explain the choice of covariance structure in relation with reproducing kernel Hilbert spaces. It also provides an overview of multilevel Monte Carlo and an error bound of the kriging predictor for our analysis as the background theory. Section 3 is the core part of this research. We show the novel multilevel adaptive sequential design of computer experiments (MLASCE) in the framework of multi-fidelity GPR and an algorithm for constructing it. Section \ref{MLMC_kriging} details an independent analysis of optimal numbers of runs in the special case of tensor grid design. Numerical examples are given in section \ref{nume_example} showing that MLASCE performs better than other existing multi-fidelity models. We also demonstrate that MLASCE is robust in a realistic and multi-dimensional example of a tsunami simulation.

\section{Gaussian process (GP) emulation: background and multi-fidelity context} 
\label{sec:GP}
We assume that $\mathbb{D} \subseteq \mathbb{R}^{d}$ is the domain of input for a positive integer $d$ and $y:\mathbb{D} \rightarrow \mathbb{R}$ denotes the deterministic quantity of interest dependent on the input. Further, let $\boldsymbol{y}(X_{N})=\left(y(x_{1}), y(x_{2}) \ldots, y(x_{N})\right)^{\top}$ denote $N$ observations of $y$ collected at the distinct inputs $X_{N}=\left(x_{1}, x_{2}, \ldots, x_{N}\right)^{\top}$ where $x_{i} \in \mathbb{D} \subseteq \mathbb{R}^{d}$, $y(x_{i}) \in \mathbb{R}$. We aim to predict $\boldsymbol{y}$ at any new input $\boldsymbol{x} \in \mathbb{D}$. The GP regression (GPR) assumes that the observation vector $\boldsymbol{y}(X_{N})$ is a realization of the $N$-dimensional random vector with multivariate Gaussian distribution
\begin{equation}
\boldsymbol{Y}=\left(Y\left(\boldsymbol{x}_{1}\right), Y\left(\boldsymbol{x}_{2}\right), \ldots, Y\left(\boldsymbol{x}_{N} \right)\right)^{\top}.
\end{equation}
where $Y(\cdot): \mathbb{D}  \rightarrow \mathbb{R}$ is a GP. From now on, we write $Y(x)$ instead of $Y_{x}$ in order to emphasize the corresponding index $x$. The GP $Y$ is usually represented using GP notation as $Y(\boldsymbol{x}) \sim GP\left(m(\boldsymbol{x}), K\left(\boldsymbol{x}, \boldsymbol{x}^{\prime}\right)\right)$, where $m(\cdot): \mathbb{D} \rightarrow \mathbb{R}$ and $K(\cdot, \cdot): \mathbb{D} \times \mathbb{D} \rightarrow \mathbb{R}$ are the mean and covariance functions
\begin{eqnarray}
\begin{aligned}
m(\boldsymbol{x}) &=\mathbb{E}[Y(\boldsymbol{x})] \\
K\left(\boldsymbol{x}, \boldsymbol{x}^{\prime}\right) &=\operatorname{Cov}\bigl(Y(\boldsymbol{x}), Y\left(\boldsymbol{x}^{\prime}\right) \bigr)=\mathbb{E}\bigl[\bigl(Y(\boldsymbol{x})-m(\boldsymbol{x})\bigr) \bigl(Y\left(\boldsymbol{x}^{\prime}\right)-m\left(\boldsymbol{x}^{\prime}\right) \bigr) \bigr]
\end{aligned}
\label{def_mean_cov}
\end{eqnarray}
The variance of $Y(\boldsymbol{x})$ is $K(\boldsymbol{x}, \boldsymbol{x})$ and its standard deviation is $\sigma(\boldsymbol{x})=\sqrt{K(\boldsymbol{x}, \boldsymbol{x})}$. The mean $m(\boldsymbol{x})$ and covariance $K\left(\boldsymbol{x}, \boldsymbol{x}^{\prime}\right)$ often include the constants which determines the structures, such as $\gamma^{2}$ and $\sigma^{2}$ in a Gaussian convariance function $\sigma^{2}\exp(-\|x-x^{\prime}\|^{2}_{2} / \gamma^{2})$ where $\|\cdot\|_{2}$ is the Euclidean distance. We call these constants hyperparameters or parameters. The covariance matrix of the random vector $\boldsymbol{Y}$ is then given by
\[
K(X_{N},X_{N})=\left(\begin{array}{ccc}
k\left(\boldsymbol{x}_{1}, \boldsymbol{x}_{1}\right) & \cdots & k\left(\boldsymbol{x}_{1}, \boldsymbol{x}_{N}\right) \\
\vdots & \ddots & \vdots \\
k\left(\boldsymbol{x}_{N}, \boldsymbol{x}_{1}\right) & \cdots & k\left(\boldsymbol{x}_{N}, \boldsymbol{x}_{N}\right)
\end{array}\right)
\]
 The hyperparameters in $m(\boldsymbol{x})$ and $K\left(\boldsymbol{x}, \boldsymbol{x}^{\prime}\right)$ are typically identified by maximizing the $\log$ marginal likelihood of the observations (see examples in Section 2.2 in \citep{rasmussen2006gaussian}). Typically, the log likelihood is given in the following form
 \[\mathcal{L}_{N}(\theta) = -\frac{N}{2}\log(2 \pi) - \frac{1}{2}\log(\operatorname{det} ( K(X_{N},X_{N})))-\frac{\boldsymbol{Y}^{\top{}} K^{-1}(X_{N},X_{N}) \boldsymbol{Y}}{2} \]
 where $\operatorname{det} ( K(X_{N},X_{N}))$ is the determinant of the matrix $K(X_{N},X_{N})$.  The GPR prediction of $y\left(\boldsymbol{x}\right)$ at $x$ consists of the posterior normal distribution $\mathcal{N}\left(m_{N}^{y}\left(x\right), \hat{s}^{2}_{N}\left(\boldsymbol{x}\right)\right),$ with the posterior or predictive mean and variance given by
\begin{eqnarray}
\label{predictives}
\begin{aligned}
m_{N}^{y}\left(\boldsymbol{x}\right) &=m\left(\boldsymbol{x}\right)+K\left(\boldsymbol{x}, X_{N}\right)^{\top} K(X_{N},X_{N})^{-1}(\boldsymbol{y}(X_{N})-\boldsymbol{m}(X_{N})) \\
\hat{\mathrm{s}}^{2}_{N}\left(\boldsymbol{x}\right) &=K(\boldsymbol{x}, \boldsymbol{x})-K\left(\boldsymbol{x}, X_{N}\right)^{\top} K(X_{N},X_{N})^{-1} K\left(\boldsymbol{x}, X_{N}\right)
\end{aligned}
\end{eqnarray}
where $\boldsymbol{m}(X_{N})=\left(m\left(x_{1}\right), \ldots, m\left(x_{N}\right)\right)^{\top}$ and $K\left(\boldsymbol{x}, X_{N}\right)$ is the vector of covariances
\[
K\left(\boldsymbol{x}, X_{N}\right)=\left(K(\boldsymbol{x}, \boldsymbol{x}_{1}), K(\boldsymbol{x}, \boldsymbol{x}_{2}), \cdots, K(\boldsymbol{x}, \boldsymbol{x}_{N}) \right)^{\top}
.\]

We now present a brief overview of the initial approach of multi-fidelity GP emulators or co-kriging by \citep{kennedy2000predicting}. A co-kriging model is introduced based on a first-order auto-regressive relation between model outputs of different levels of fidelity. Suppose that we have $L$-levels of variable-fidelity computer simulations $\left(y_{l}(x)\right)_{l=1}^{L}$, which are deterministic. We denote  $\boldsymbol{y}_{l,N_{l}}=\left(y_{l}(x_{1,l}), y_{l}(x_{2,l}) \ldots, y_{l}(x_{N_{l},l})\right)^{\top}$ observations at the inputs $X_{l,N_{l}}=\left(x_{1,l}, \dots, x_{N_{l},l}\right)^{\top{}}$ sorted by increasing fidelity and modelled as observations of a Gaussian process $\left(Y_{l}(\mathrm{x})\right)_{l=1}^{L}$ where $x_{i,l} \in \mathbb{D} \subseteq \mathbb{R}^{d}, y_{l}(x_{i,l}) \in \mathbb{R}$ ($1 \leq i \leq N_{l}$). Then, $y_{L}(x)$ denotes the output of the most accurate and expensive model, while $y_{1}(x)$ is the output of the cheapest and least accurate surrogate at our disposal. It is assumed that the corresponding experimental design sets $\left \{X_{l,N_{l}}\right\}_{l=1}^{L}$ have a nested structure, i.e., $X_{L,N_{L}} \subseteq X_{L-1,N_{L-1}} \subseteq \ldots \subseteq X_{1,N_{1}} .$ The auto-regressive co-kriging scheme of \citep{kennedy2000predicting} reads as
\begin{eqnarray}
\label{ohagan}
Y_{l}(x)=\rho_{l-1} Y_{l-1}(x)+\delta_{l}(x), \quad l=2, \ldots, L
\end{eqnarray}
where $\delta_{l}(x)$ is a Gaussian process independent of $\left\{Y_{l-1}(x), \ldots, Y_{1}(x)\right\}$ and distributed with expectation $\mu_{\delta_{t}}$ and covariance $\Sigma_{l}$, i.e. $\delta_{l} \sim \mathcal{N}\left(\mu_{\delta_{l}}, \Sigma_{l}\right)$. We denote by convention, $Y_{1}(x) = \delta_{1}(x)$ and also utilize a Gaussian process to represent this base low fidelity run, as we do for the increments. Also, $\rho$ is a scaling factor and a deterministic scalar independent of $x$. $\rho$ also quantifies the correlation between the model outputs $\left(y_{l}(x), y_{l-1}(x)\right)$. The derivation of this model is based on the Markov property
\[
\operatorname{Cov} \bigl( Y_{l}(\mathbf{x}), Y_{l-1}\left(\mathbf{x}^{\prime}\right) | Y_{l-1} (\mathbf{x}) \bigr)=0, \quad \text{ for all } \mathbf{x} \neq \mathbf{x}^{\prime}
\]
which translates into assuming that given $Y_{l-1}(x),$ we can learn nothing more about $Y_{l}(\mathrm{x})$ from any other model output $Y_{l-1}\left(\mathrm{x}^{\prime}\right),$ for $\mathrm{x} \neq \mathrm{x}^{\prime}$.

The resulting posterior distribution at the $l$th co-kriging level has a mean and covariance given by
\begin{eqnarray}
\begin{aligned}
\mu_{Y_{l} | \boldsymbol{y}_{l,N_{l}}, \boldsymbol{y}_{l-1,N_{l-1}}, \ldots, \boldsymbol{y}_{1,N_{1}}}(\mathbf{x})=\mu_{l}+\left(\Sigma_{l} A_{l}^{\top }\right)\left(A_{l} \Sigma_{l} A_{l}^{\mathrm{T}}+\sigma_{l} I\right)^{-1}\left(\boldsymbol{y}_{l,N_{l}}-A_{l} \mu_{l}\right) \\
\Sigma_{Y_{l} | \boldsymbol{y}_{l,N_{l}}, \boldsymbol{y}_{l-1,N_{l-1}}, \ldots, \boldsymbol{y}_{1,N_{1}}}(\mathbf{x})=\Sigma_{Y_{l}}+\rho_{l-1}^{2} \Sigma_{Y_{l-1}}+\rho_{l-1}^{2} \rho_{l-2}^{2} \Sigma_{Y_{l-2}}+\cdots+\left(\Pi_{l=1}^{l-1} \rho_{l}\right)^{2} \Sigma_{Y_{1}}
\end{aligned}
\label{ohagan_mean}
\end{eqnarray}
where $\mu_{l}$ is the mean value of $Y_{l}(x)$,  $ \Sigma_{l}$ is a covariance matrix comprising $l$ blocks, representing all cross-correlations between levels $\{l, l-1, \ldots, 1\}$ and $\Sigma_{Y_{l}}=\Sigma_{l}-\left(\Sigma_{l} A_{l}^{\top}\right)\left(A_{l} \Sigma_{l} A_{l}^{\top}\right)^{-1}\left(A_{l}^{\top } \Sigma_{l}\right)$
is the covariance of the co-kriging predictor at level $l$. Also, $A_{l}$ is a simple matrix that restricts a realization of the Gaussian process $Y_{l}(x)$ to the locations of observed data $\boldsymbol{y}_{l,N_{l}}$ at level $l$. Similar to ordinary GPR, the set of unknown hyperparameters in these distributions are determined by maximum likelihood method or standard Markov Chain Monte Carlo. 

\cite{le2014recursive} presents the extended formulation of the multi-fidelity model and \cite{le2015cokriging} follows its special form and considers a strategy of sequential design in this special case. The unique difference with the previous model is that  $Y_{l}(x)$ (the Gaussian process modeling the response at level $l$ ) is a function of the Gaussian process $Y_{l-1}(x)$ conditioned on the values $\{ \boldsymbol{y}_{i,N_{i}} \}_{i=1}^{l-1}$ at points in the experimental design sets $\{X_{i,N_{i}}\}_{i=1, \ldots, l-1}$. For brevity, let the realizations $Y_{l}(X_{l,N_{l}})$ denote $\boldsymbol{y}_{l,N_{l}}$. They also introduced the dependency of $\rho_{l}(x)$ on input $x$. As in the previous model, the nested property for the experimental design sets is assumed.

Let us consider the following model for $l=2, \ldots, L:$
\begin{eqnarray}
\label{leg}
\left\{\begin{array}{l}
Y_{l}(x)=\rho_{l-1}(x) \tilde{Y}_{l-1}(x)+\delta_{l}(x) \\
\tilde{Y}_{l-1}(x) \perp \delta_{l}(x) \\
\rho_{l-1}(x)=g_{l-1}^{\top{}}(x) \beta_{\rho_{l-1}}
\end{array}\right.
\end{eqnarray}
where $ \delta_{l}(x)$ is a
Gaussian process $GP\left(f_{l}^{\top{}}(x) \beta_{l}, \sigma_{l}^{2} r_{l}\left(x, x^{\prime}\right)\right)$, $f_{l}(x)$ and $g_{l}(x)$ are vectors of regression functions and $\beta_{l}$ and $\beta_{\rho_{l}}$ are the coefficient vectors (hyperparameters) for $f_{l}(x)$ and $g_{l}(x)$. $\perp$ denotes the independence relationship and $X_{L,N_{L}} \subseteq X_{L-1,N_{L-1}} \subseteq \cdots \subseteq X_{1,N_{1}}$. 
$\tilde{Y}_{l-1}(x)$ is a GP conditioned on $\left[ \boldsymbol{y}_{l-1,N_{l-1}}, \beta_{l-1}, \beta_{\rho_{l-2}}, \sigma_{l-1}^{2}\right]$ and the hyperparameters $\beta_{l-1}, \beta_{\rho_{l-2}}, \sigma_{l-1}^{2}$ are assumed to have the prior distributions. Then, for $l=2, \ldots, L$ and $x \in \mathbb{R}^{d}$, we have: 
\[
\left[Y_{l}(x) \mid \boldsymbol{y}_{l,N_{l}}, \beta_{l}, \beta_{\rho_{l-1}}, \sigma_{l}^{2}\right] \sim \mathcal{N}\left(\mu_{Y_{l}}(x), s_{Y_{l}}^{2}(x)\right)
\]
where
\begin{eqnarray}
\mu_{Y_{l}}(x) &=&\rho_{l-1}(x) \mu_{Y_{l-1}}(x)+f_{l}^{\top{}}(x) \beta_{l} \nonumber \\
&+&r_{l}^{\top{}}(x) R_{l}^{-1}\left(Y_{l}\left(X_{l,N_{l}}\right)-\rho_{l-1}\left(X_{l,N_{l}}\right)  \odot  Y_{l-1}\left(X_{l,N_{l}}\right)-F_{l} \beta_{l}\right)
\label{legratiet_mean}
\end{eqnarray}
and
\[
\sigma_{Y_{l}}^{2}(x)=\rho_{l-1}^{2}(x) \sigma_{Y_{l-1}}^{2}(x)+\sigma_{l}^{2}\left(1-r_{l}^{\top{}}(x) R_{l}^{-1} r_{l}(x)\right)
\]
$R_{l}$ is the correlation matrix $R_{l}=\left(r_{l}\left(x, x^{\prime}\right)\right)_{x, x^{\prime} \in X_{l,N_{l}}}, r_{l}^{\top{}}(x)$ is the correlation vector $r_{l}^{\top{}}(x)=\left(r_{l}\left(x, x^{\prime}\right)\right)_{x^{\prime} \in X_{l,N_{l}}}$. $F_{l}$ and $\rho_{l-1}\left(X_{l,N_{l}}\right)$ is the matrix containing the values of $f_{l}(x)^{\top{}}$ and $\rho_{l-1}(x)^{\top{}}$ on $X_{l,N_{l}}$. The notation $\odot$ represents the element by element matrix product.

The mean $\mu_{Y_{l}}(x)$ is the surrogate model of the response at level $l, 1 \leq l \leq L,$ taking into account the known values of the $l$ first levels of responses $\left(Y_{i}\right)_{i=1, \ldots, L}$ and the variance $\sigma_{Y_{l}}^{2}(x)$ represents the mean squared error of this model. The mean and the variance of the Gaussian process regression at level $l$ are expressed in function of the ones of level $l-1$ and we have a recursive multi-fidelity GPR. The estimation of hyperparameters is implemented by maximum likelihood method or minimising a loss function of a Leave-One-Out Cross-Validation procedure. \cite{le2014recursive} considers a nested space-filling design for choosing $X_{l,N_{l}}$ and \cite{le2015cokriging}, which assumes constant $\rho$ thus a special case of \cite{le2014recursive}, proposes a strategy of sequential design, which combines both the error evaluation provided by the GPR model and the observed errors of a Leave-One-Out Cross-Validation procedure.

One of the drawbacks of these formulations is that the discrepancy of the mean of the statistical emulator from the true response is not considered. Therefore, the above methods would not guarantee accurate prediction in theory and they do not inherently incorporate the prediction error. Indeed, there is little theoretical background for the predictive mean functions in (\ref{ohagan_mean}) and (\ref{legratiet_mean}) to get closer to the true response even as the number of the samples grows. Thus, we expect that more carefully specified mean functions would lead to a reduction in the prediction error.

\
\cite{perdikaris2017nonlinear} generalizes the autoregressive multi-fidelity scheme of equation (\ref{ohagan}) while it does not cover the framework of experimental design. This formulation is the so-called deep GP as first put forth in \cite{damianou2013deep,damianou2015deep} and goes well beyond the linear structure of the multi-fidelity GPR. Their formulation is 
\begin{eqnarray}
Y_{l}(x)=f_{l-1}\left(Y_{l-1}(x)\right)+\delta_{l}(x)
\label{deep_GP}
\end{eqnarray}
where $f_{l-1}(\cdot)$ is an unknown function that maps the lower fidelity model output to the higher fidelity one. They propose a Bayesian non-parametric treatment of $f_{l-1}$ by assigning it a GP prior. Because $Y_{l-1}$ in equation (\ref{deep_GP}) is also assigned a GP prior, the functional composition of two GP priors, i.e. $f_{l-1}\left(Y_{l-1}(x)\right)$, gives rise to the deep GP framework \cite{damianou2013deep,damianou2015deep}, and, therefore, the posterior distribution of $Y_{l}$ is no longer Gaussian. 

To avoid a significant increase in computational cost and far more complex implementations than standard GPR, they replace the GP prior $Y_{l-1}$ with the GP posterior from the previous inference level $Y_{* l-1}(x)$, where $Y_{* l-1}(x)$ is the posterior of $Y_{l-1}$. Then, using the additive structure of equation (\ref{deep_GP}), along with the independence assumption between the GPs $Y_{l-1}$ and $\delta_{l}$, they summarise the autoregressive scheme of equation (\ref{deep_GP}) as
$$
Y_{l}(x)=g_{l}\left(x, Y_{*{l-1}}(x)\right)
$$
where $g_{l} \sim \text{GP}\left( \mathbf{0}, k_{l}\left(\left(x, Y_{* _l-1}(x)\right),\left(x^{\prime}, Y_{*_l-1}\left(x^{\prime}\right)\right) \right)\right)$.
They consider a covariance kernel that decomposes as
$$
k_{l}=k_{l_{1}}\left(x, x^{\prime} \right) \cdot k_{l_{2}}\left(Y_{*{l-1}}(x), Y_{*{l-1}}\left(x^{\prime}\right)\right)+k_{l_{3}}\left(x, x^{\prime}\right)
$$
where $k_{l_{1}}, k_{l_{2}}$ and $k_{l_{3}}$ are square exponential covariance functions with different hyperparameters. The estimation of the hyperprameters is based on the maximum-likelihood estimation procedure and they utilise the gradient descend optimizer L-BFGS using randomized restarts to ensure convergence to a global optimum.
For the purpose of performing predictions given uncertain inputs $x_{*}$ for $l \geq 2$, the posterior distribution of $Y_{* l}\left(x_{*}\right)$ is given by
\begin{eqnarray}
\begin{aligned}
p\left(Y_{* l}\left(x_{*}\right)\right) &:=p\left(Y_{l}\left(x_{*}, Y_{*_{l-1}}\left(x_{*}\right)\right) \mid Y_{*_{l-1}}, x_{*}, y_{l}(X_{l,N_{l}}), X_{l,N_{l}}\right) \\
&=\int p\left(Y_{l}\left(x_{*}, Y_{*_{l-1}}\left(x_{*}\right)\right) \mid  x_{*}, y_{l}(X_{l,N_{l}}), X_{l,N_{l}}\right) p\left(Y_{*_{l-1}}\left(x_{*}\right)\right) \mathrm{d} \bigl( Y_{*_{l-1}}(x_{*}) \bigr)
\end{aligned}
\label{deep_MC}
\end{eqnarray}
where $p\left(Y_{*_{l-1}}\left(x_{*}\right)\right)$ denotes the posterior distribution of $Y_{*_{l-1}}$ at the previous level $(l-1)$. For $l=1$, the posterior distribution is simply the multidimensional normal distribution with the mean and covariance given in (\ref{predictives}). They compute the predictive mean and variance of all posteriors $p\left(Y_{*_{t}}\left(x_{*}\right)\right), l \geq 2$, using Monte Carlo integration of equation (\ref{deep_MC}). 

In their numerical example, they claim that their algorithm allows for capturing complex nonlinear, non-functional and space-dependent cross-correlations based on their deep GP formulation. 
Moreover, in the scarce data regime typically encountered in multi-fidelity applications, they argue that their algorithm avoids the problem of estimating a large number of hyperparameters in $\rho_{l}(x)$ in (\ref{leg}) and achieves a fundamental extension of the schemes by assuming the deep GP formulation.  
However, their approach assumes that all kernels account for directional anisotropy in each input dimension using automatic relevance determination (ARD) weights \cite{rasmussen2006gaussian} hence requires the estimation of $(2d + 3)$ model hyperparameters for every $l \geq 2$, resulting in the total $(L-1)\times (2d+3) + d+1$ hyperparameters. In a high dimensional case and under a severe restriction of the computational budget, the number of the hyperparameters might still be overwhelming. Furthermore, their Monte Carlo integration scheme is highly prohibitive in computational costs when the number of the levels is large. 


\subsection{Sequential design }
An experimental design is a specification of inputs in the domain of our interest at which we wish to compute the outputs \cite{santner2003design}. In our case, experimental design is about obtaining the training data and specifying the appropriate location. In the case of GPR, the posterior mean and covariance functions are used for its prediction. Therefore, the locations of inputs and corresponding outputs used as the training data as well as its number of data points do affect the quality of the prediction. In a design of experiments, it is customary to use space-filling designs, which spread observations evenly throughout the input region \cite{simpson2001sampling,santner2003design}, such as uniform designs, multilayer designs, maximin (Mm)- and minimax (mM)-distance designs, and Latin hypercube designs (LHD). Space-filling designs treat all parts of the design space as equally important hence it can be inefficient especially if the domain of input is high dimensional. A variety of adaptive designs (an algorithm is called adaptive if it updates its behavior to new data) have been proposed, which can take advantage of information collected during the experimental design process \citep{santner2003design}, typically in the form of input-output data from simulation runs.  In the case of GPR, the hyperparameters in kernel functions are updated recursively whenever new data becomes available. Some classical adaptive design criteria are the maximum mean squared prediction error (MMSPE), the integrated MSPE (IMSPE), and the entropy criterion \citep{sacks1989design}.
\

For a computer experiment, the domain $\mathbb{D}$ is often discretised into grid $\mathbb{D}_{G}$ with $N_{G}$ points in practice since the optimization search of training data is a formidable task if $\mathbb{D}$ is not.  Consequently, we replace the search over
$\mathbb{D}$ by a search over a set of candidate points $X_{\text{cand}} \subseteq \mathbb{D}_{G}$. We consider sequential designs as practical, computationally cheaper alternatives to “one shot” designs, which determines all input points in one step, albeit often suboptimal. The sequential design is defined as follows: Suppose that we have an initial design $ \{(x_{j},y(x_{j}))\}_{j=1}^{n} = \{X_{n},\boldsymbol{y}(X_{n})\}$, then for each $k = n, n+1, n+2, \ldots$ one collects an input-output pair $(x_{k+1},y(x_{k+1}))$ by choosing the input values
\begin{eqnarray*}
x_{k+1} = \argmax_{x \in X_{\text{cand}} \setminus X_{k} } F_{k}(x).
\end{eqnarray*}
where $F_{k}(\cdot)$ is a design criterion to be maximized. The algorithm iterates until a stopping criterion is met or the computational budget allocated is exhausted. Often the design criterion makes use of the predictive variance of a Gaussian process, defined in (\ref{predictives}) and denoted by $\hat{s}^{2}_{k}(x)$   when conditional on $X_{k}$ with $k$ points with estimated parameters in $K$. Indeed it is natural to look for points where the predictive variance is large, or where the variance can be reduced the most when adding a point in that region. In our context, the parameters in the kernel of a GP need to be estimated. Sequential designs allow the estimates to be improved sequentially with the additional new design points. This is especially advantageous when some input variables are considerably more influential on the output than others. 

We employ a useful method called MICE (mutual information computer experiment) design as a design strategy in our proposed framework, MLASCE. In general, MICE outperforms other approaches for sequential design, such as the active learning MacKay (ALM) \cite{mackay1992information} and Active learning Cohn (ALC) \cite{cohn1995active} algorithm, as shown in the numerical examples in \citep{beck2016sequential}. For our purpose of selecting a very small number of runs in the design for high fidelity simulators, it is essential to pick these runs very carefully to optimize the design. MICE aims to construct an emulator for an expensive simulation with small to medium-sized experimental designs, therefore, an appropriate tool in our framework. Since MICE is a significant improvement from mutual information criteria, we start from the review of it.

\subsubsection{Mutual information criteria}
Mutual information (MI) is a classical information-theoretic measure. The goal is to place a set of design $\mathcal{X}$ that will give us good predictions at all unexplored locations $\mathbb{D} \setminus \mathcal{X}$. Specifically, we want to find 
\begin{eqnarray}
\label{MI}
\mathcal{X}^{*} = \argmax_{\mathcal{X} \subseteq X_{\text{cand}}} H(\mathbb{D}_{G} \setminus \mathcal{X}) - H(\mathbb{D}_{G} \setminus \mathcal{X} | \mathcal{X} )
\end{eqnarray}
where $H(X)$ denotes the entropy of Gaussian process $Y(X)=\left( Y(x_{1}),\ldots, Y(x_{n}) \right)$ with the location $X= \left(x_{1}, \ldots, x_{n}\right)$ and $H(X|Z)$ is the entropy of a conditional Gaussian process $Y(X|Z)$ conditioned on the realized values $Y(Z)$ with the location $Z= \left(z_{1}, \ldots, z_{m}\right)$. 
The set $\mathcal{X}^{*}$ maximally reduces the entropy over the rest of space $\mathbb{D}_{G} \setminus \mathcal{X}$. The entropy of a Gaussian process is  
\[ H(X)= {n \over 2} \log(2 \pi) \operatorname{det} \bigl( \operatorname{Cov} (Y(X)) \bigr)\] where $\operatorname{det} \left( \cdot \right)$ is the determinant of a matrix and $\operatorname{Cov} (Y(X))$ means the covariance matrix of $Y(X)$. Though this optimization problem  is NP-hard, \cite{krause2008near} presented an alternative algorithm, known as MI algorithm, which avoids directly maximizing the difference of entropies. The algorithm is as follows and we denote the collection of $k$ inputs $X_{k} = \left( x_{1}, \ldots, x_{k}\right)$ and outputs $\boldsymbol{y}_{k} = \left(y_{1}, \ldots, y_{k} \right)$.

\begin{description}
  \item[Require]: Function $y(x)$, GP emulator $(m(\cdot), K(\cdot, \cdot))$, nugget parameters $\tau^{2}$, grid $\mathbb{D}_{G}$, candidate set $X_{\text{cand}} \subseteq \mathbb{D}_{G}$, an initial design $X_{k} \subseteq X_{\text{cand}}$ of size $k$, desired design size $N$.
  \item[Step 1] Let $X_{\text{cand}} \leftarrow X_{\text{cand}}  \setminus X_{k} $ 
  \item[Step 2] Solve $x_{k+1} \leftarrow  \argmax_{x \in X_{\text{cand}}}  \hat{s}^{2}_{k}(x,\tau^{2})  \slash \hat{s}^{2}_{G \setminus (k\cup x)}(x,\tau^{2})$
    \item[Step 3] Let $X_{k+1} \leftarrow X_{k} \cup x_{k+1}$, and $X_{\text{cand}} \leftarrow X_{\text{cand}} \setminus x_{k+1}$
  \item[Step 4] If $k+1= N $, then stop; otherwise let $k=k+1$, and go to Step 1.
   \item[Output] : $X_{N}$. 
\end{description}

In the above algorithm, a parameter $\tau^{2}>0$ is introduced in the diagonal elements of the correlation matrix $K$ of the GP. Such a parameter is often called a nugget parameter. In practice, when ${K}$ should be inverted, we actually invert $K+ \tau^{2} I$ instead of $K$, where $I$ is the $k \times k$ identity matrix. A nugget parameter $\tau^{2}$ is commonly used to stabilize the inversion, using the Cholesky decomposition, and usually takes a very small value. $\hat{s}^{2}_{k}(x,\tau^{2})$ is the prediction variance where the inversion of $K$ is replaced with that of  $K+\tau^{2}I$ in (\ref{predictives}) and $G \setminus (k\cup x)$ in $\hat{s}^{2}_{G \setminus (k\cup x)}(x,\tau^{2})$ means the number of points where $X_{k}$ and $x$ are excluded. In Step $2$, a GP emulator is assigned to the set of points $\{ (x_{i}, y(x_{i})) \}_{i=1}^{k}$, and, for each $x \in X_{\text{cand}}$, a GP emulator is assigned to $\mathbb{D}_{G} \setminus\left(X_{k} \cup x\right)$. The GP over $\mathbb{D}_{G} \setminus \left(X_{k} \cup x\right)$ is required in order to estimate the difference between the total information and the information we have obtained by $X_{k} \cup x$.

\subsubsection{Mutual information for computer experiments}
 MI reduces uncertainties over the domain with relatively cheap computational costs. However, MI lacks adaptivity; that is, the parameters in the covariance are assumed to be known. Moreover, computational instability is inherent due to the inversion of the covariance matrix in the denominator $\hat{s}^2_{G \cup X}(x, \tau^{2})$ in Step 2 of MI algorithm. \cite{beck2016sequential} developed a modified approach known as mutual information for computer experiments (MICE). They introduced a different nugget parameter in diagonal elements of the covariance matrix in the denominator $\hat{s}^2_{G \cup X}(x, \tau^{2})$ to stabilize the inversion of a matrix. In addition, they added adaptivity in the sense that $\mathbb{D}_{G}$ is sampled in the design space $\mathbb{D}$ instead of keeping it fixed, and the parameters in the covariance matrix are estimated at each step of sampling. The outlined algorithm is as follows \cite{beck2016sequential}. The covariance $K(\cdot, \cdot)$ is replaced with $K_{\theta}(\cdot, \cdot)$ to emphasize the hyperparameters $\theta$ in $K$. 

\begin{description}
  \item[Require]: Function $y(x)$, GP emulator $(m(\cdot), K_{\theta}(\cdot, \cdot ))$, nugget parameters $\tau^{2}$ and $\tau^{2}_{s}$, domain of input $\mathbb{D}$, initial data $\{X_{k},\boldsymbol{y}(X_{k}) \}$, discrete set size $N_{G}$, candidate set size $N_{\text{cand}}$, desired size $N$. 
  \item[Step 1] Implement maximum likelihood estimation to obtain estimates $\hat{\theta}_{k}$ of $\theta$ based on $\{ X_{k},\boldsymbol{y}(X_{k}) \}$ in $K_{\theta}(\cdot, \cdot)$.
  \item[Step 2] Fit GP emulator to data  $\{ X_{k},\boldsymbol{y}(X_{k}) \}$.
    \item[Step 3] Generate a discrete set $\mathbb{D}_{G}$ of size $N_{G}$, and choose a candidate set $X_{\text{cand}} \subseteq \mathbb{D}_{G}$.
  \item[Step 4] Solve $x_{k+1} = \argmax_{x \in X_{\text{cand}}}  \hat{s}^{2}_{k}(x,\hat{\theta}_{k},\tau^{2}) \slash \hat{s}^{2}_{G \setminus (k\cup x)}(x,\hat{\theta}_{k},\max({\tau^{2}, \tau^{2}_{s}}))$
    \item[Step 5] Evaluate $y_{k+1}=y(x_{k+1})$, and let $X_{k+1} = X_{k} \cup x_{K+1}$ and $\bold{y}(X_{k+1}) = \bold{y}(X_{k}) \cup y_{k+1}$.
  \item[Step 6] If $k+1= N $, then stop; otherwise let $k=k+1$, and go to Step 1.
   \item[Output]: The collection of inputs and outputs  $\{X_{N},\boldsymbol{y}(X_{N}) \}$ of size $N$. 
\end{description}
In step $2$, the maximum likelihood method is implemented to get the estimation of the hyperparameters in the covariance. For example, if the covariance kernel is Gaussian, the shape and scale parameters are estimated while all of the hyperparameters except the smoothness parameter are estimated in the case of Mat\'ern kernel. In Step $3$, \cite{beck2016sequential} suggest that $\mathbb{D}_{G}$ be sampled in the design space $\mathbb{D},$ instead of keeping $\mathbb{D}_{G}$ fixed throughout. In Step 4 the MICE criterion is evaluated for all $x \in X_{\text {cand }}$ and $\tau^{2}_{s}$ is another nugget parameter that may be used in place of $\tau^{2}$  to avoid numerical stability in the denominator $\hat{s}^{2}_{G \setminus (k\cup x)}(x,\hat{\theta}_{k},\max({\tau^{2}, \tau^{2}_{s}}))$. The estdimate $\hat{\theta}_{k}$ in $\hat{s}^{2}_{k}(x,\hat{\theta}_{k},\tau^{2})$ is explicitly written so that the hyperparameter is updated adaptively. The nugget parameter $\tau^{2}$ usually takes small values and would cause numerical instability if $\tau_{s}^{2}$ were not in use. \cite{beck2016sequential} suggests around $1.0$ for $\tau_{s}^{2}$. 

A few comments on the choice of a nugget parameter are given here. \cite{ranjan2011computationally} and \cite{peng2014choice} suggest a framework for obtaining a lower bound of a nugget parameter to address ill-conditioning, whereas \cite{dancik2008mlegp} uses MLE of the small nugget parameter for achieving numerical stability. In contrast, \cite{gramacy2012cases} strongly advocates for the use of a (non-small) nugget parameter compared to consequences of the zero-nugget model concentrating on deterministic computer experiments. They argue that the nugget is crucial for maintaining smoothness of the emulator when the data are too sparse to get a good fit of the function. They also observe that while the nugget model smooths and produces reasonable confidence bands, the no-nugget model ends up making predictions well outside the range of the actual data in that region, and its confidence bands are all over the place. Further, \cite{beck2016sequential} derives in their Theorem 3.1 if the nugget $\tau^{2}$ is large enough, the predictive variance $\hat{s}_{\tau^{2}}^{2}\left(\boldsymbol{x}_{i}\right) \approx \sigma^{2} \tau^{2} / k$ for $\boldsymbol{x}_{i} \in \boldsymbol{X}_{k}$ ($k$ is the number of points in the design), which is used in the denominator of Step 4 in the above MICE algorithm. Thus, the initial choice $\tau^{2}_{s}=1.0$, assuming this value is large enough for this approximation, could be partially verified as both the prediction and the denominator of Step 4 would become stable.

\subsection{Reproducing kernel Hilbert spaces and choice of kernels}
To implement the error analysis of an emulator and evaluate the significance of expensive simulations compared to the cheaper ones based on the function norm, a reproducing kernel Hilbert space (RKHS) is a convenient tool. Apparently, the first deep investigation of RKHS goes back to \citep{aronszajn1950theory}, and an overview of statistical applications for RKHS theory is provided in \citep{berlinet2011reproducing}. The following definition of RKHS is given in \citep{hsing2015theoretical}. 

\begin{dfn}[{\cite{hsing2015theoretical}}, Definition 2.7.1.] 
Let $\mathbb{H}$ be a Hilbert space composed of the real valued functions defined on $\mathbb{D}$. A bivariate function $K$ on $\mathbb{D} \times \mathbb{D}$ is said to be a reproducing kernel (RK) for $\mathbb{H}$ if

\begin{enumerate}
\item for every $t \in \mathbb{D}, K(\cdot, t) \in \mathbb{H}$ and
\item $K$ satisfies the reproducing property that $f(t)=\langle f, K(\cdot, t)\rangle_{\mathbb{H}}$ for every $f \in  \mathbb{H}$ and $t \in \mathbb{D}$ where $\langle \cdot ,  \cdot \rangle_{\mathbb{H}}$ is the inner product of $\mathbb{H}$.
\end{enumerate}

When $\mathbb{H}$ possesses an RK it is said to be an RKHS.

\end{dfn}

A covariance function in (\ref{def_mean_cov}) precisely corresponds to the reproducing kernel of an RKHS. Indeed, if we define the $\mathbb{R}$-linear space using the stationary covariance function $K$, $[\mathbb{H}_{0} \coloneqq \text{span}\{K(\cdot,y) \hspace{4pt} : \hspace{4pt} y \in \mathbb{D}\}$, and equip this space with the bilinear form
\[\langle \sum_{j=1}^{N} \alpha_{j}K(\cdot,x_{j}), \sum_{i=1}^{M} \beta_{i}K(\cdot,y_{i}) \rangle_{\mathbb{H}_{0}} \coloneqq \sum_{j=1}^{N} \sum_{i=1}^{M} \alpha_{j} \beta_{i} K(x_{j},y_{i})\]
then $\mathbb{H}_{0}$ is a pre-Hilbert space with a reproducing kernel $K$ and $\langle \cdot, \cdot\rangle_{\mathbb{H}_{0}}$ is an inner product in $\mathbb{H}_{0}$ (Theorem 10.7 in \citep{wendland2004scattered}). Let $\tilde{\mathbb{H}}_{0}$ be the completion of $\mathbb{H}_{0}$ with respect to $\langle \cdot, \cdot\rangle_{\mathbb{H}_{0}}$ and the linear mapping for every $f \in \tilde{\mathbb{H}}_{0}$ be defined as

\[R: \tilde{\mathbb{H}}_{0} \rightarrow C(\mathbb{D}), \hspace{10pt} R(f)(x) = \langle f, K(\cdot, x) \rangle_{\mathbb{H}_{0}}.\]

As stated in Definition 10.9 and Theorem 10.10 in \citep{wendland2004scattered}, we can construct an RKHS $\mathbb{H}$ by defining $\mathbb{H} \coloneqq R(\tilde{\mathbb{H}}_{0})$ with the RK $K$ where $C\left(\mathbb{D}\right)$ is a set of continuous functions in the domain $\mathbb{D}$. Further, $\mathbb{H}_{0}$ is dense in $\mathbb{H}$ by its construction.

The structure of RKHS, or a variety of functions with which a Gaussian process emulator can deal, is totally determined by the covariance structure. In other words, we determine the function space where the true response lies by choosing the covariance structure. In this paper, Mat\'ern covariance functions are used due to a clear interpretation of the corresponding function spaces. We start by reviewing the following form of Mat\'ern covariance functions
\begin{eqnarray}
\label{mat}
K_{\nu, \lambda, \sigma^{2}}(x,x^{\prime}) =\sigma^{2} \frac{2^{1-\nu }}{\Gamma(\nu)}\left(\sqrt{2 \nu} \frac{ \|x-x^{\prime} \|_{2}}{ \lambda}\right)^{\nu} J_{\nu}\left(\sqrt{2 \nu} \frac{ \|x-x^{\prime} \|_{2}}{ \lambda}\right)
\end{eqnarray}
where $\|x-x^{\prime} \|_{2}$ denotes the Euclidean distance between $x$ and $x^{\prime}$ for $x, x^{\prime} \in \mathbb{R}^{d}$, $\Gamma(\nu)$ is the Gamma function for $\nu>0$ and $J_{\nu}$ is a modified Bessel function of order $\nu$. The hyperparameter $\sigma^{2}$ is usually referred to as the (marginal) variance, $\lambda$ as the correlation length and $\nu$ as the smoothness parameter. The expression for the Mat\'ern covariance kernel simplifies for particular choices of $\nu$. Notable examples include the exponential covariance kernel $\sigma^{2} \exp \left(-\left\|x-x^{\prime}\right\|_{2} / \lambda\right)$ with $\nu=1 / 2,$ and the Gaussian covariance kernel $\sigma^{2} \exp \left(-\left\|x-x^{\prime}\right\|_{2}^{2} / \lambda \right)$ in the $\operatorname{limit} \nu \rightarrow \infty$. Useful examples include the cases where $\nu = m +1/2$ for a non-negative integer $m$ \citep{rasmussen2006gaussian}:
\begin{eqnarray*}
K_{m +1/2, \lambda, \sigma^{2}}(x,x^{\prime}) = \sigma^{2}  \exp \Bigl( - \frac{\sqrt{2\nu}\|x-x^{\prime} \|_{2} }{\lambda} \Bigr) \frac{\Gamma(m+1)}{\Gamma(2m+1)} \sum_{i=0}^{m} \frac{(m+i)!}{i!(m-i)!}\Bigl(\frac{\sqrt{8\nu}\|x-x^{\prime}\|_{2}}{\lambda} \Bigr)^{m-i} 
\end{eqnarray*}
Choosing the hyperparemeters $\{ \nu, \lambda, \sigma^{2} \}$ in Mat\'ern kernel functions is quite important. We shall see later that $\nu$ determines the smoothness of the function we want to recover. Indeed, an overly large $\nu$, which implies that the corresponding RKHS is composed of very smooth functions, can miss the target function which may be potentially less smooth. Sobolev spaces are the common way of measuring the smoothness of a function \citep{tuo2020kriging,wang2020prediction,wynne2021convergence,teckentrup2020convergence} since we can interpret the function spaces via weak differentiability and we grasp how the parameter affects our approximation in the same vein. Let us denote $\mathbb{H}(\mathbb{R}^{d})$ be a RKHS defined in $\mathbb{R}^{d}$ with a Mat\'ern kernel ${K_{\nu, \lambda, \sigma^{2}}}$ specified in $(\ref{mat})$. Moreover, we define the Sobolev spaces $H^{s}(\mathbb{R}^{d})$ with the norm $ \|\cdot\|_{H^{s}}$ as 
\begin{eqnarray}
\label{fracsobolev}
\begin{aligned}
 H^{s}(\mathbb{R}^{d}) \coloneqq \left\{f \in L^{2}\left( \mathbb{R}^{d}\right) \mid (1+\|\xi\|^{2})^{s /2 }\widehat{f}(\xi) \in L^{2}(\mathbb{R}^{d})  \right\} & \\
 \|\cdot\|_{H^{s}} \coloneqq \Bigl( \int_{\mathbb{R}^{d}} (1+\|\xi\|^{2})^{s } |\widehat{f}(\xi)|^{2} d \xi \Bigr)^{1/2}
 \end{aligned}
 \end{eqnarray}

where $\hat{f}(\xi)$ is the Fourier transform of $f(x)$ and $L^{p}(\mathbb{R}^{d})$ means the $L^{p}$ space in the measure space $(\mathbb{R}^{d},\mathcal{R},\mu)$ for $1 \leq p \leq \infty$ ($\mathcal{R}$ is the $\sigma -$algebra on $\mathbb{R}^{d}$ and $\mu$ is the Lebesgue measure on $(\mathbb{R}^{d}, \mathcal{R})$). Then, we can understand the smoothness of the functions in $\mathbb{H}(\mathbb{R}^{d})$ by the following proposition.

\begin{prop}[{\cite{stuart2018posterior}}, Proposition 3.2]
\label{stuart}
 Let $k_{\nu, \lambda, \sigma^{2}}$ be a Mat\'ern covariance kernel as defined in (\ref{mat}).
Then the RKHS $\mathbb{H}(\mathbb{R}^{d})$ with the reproducing kernel $K_{\nu, \lambda, \sigma^{2}}$ is equal to the Sobolev space $H^{\nu+d / 2}(\mathbb{R}^{d})$ as a vector space, and the norm of $\mathbb{H}(\mathbb{R}^{d})$  and $H^{\nu+d / 2}(\mathbb{R}^{d})$ are equivalent.
\end{prop}

It is seen that as a vector space, the RKHS $\mathbb{H}(\mathbb{R}^{d})$ of the Mat\'ern kernel $K_{\nu, \lambda, \sigma^{2}}$ is fully determined by the smoothness parameter $\nu$. The parameters $\lambda$ and $\sigma^{2}$ do, however, influence the constants in the norm equivalence of the RKHS norm and the standard Sobolev norm. In summary, if we assume the target function $y$ which we aim to recover satisfies $y \in H^{\nu+d / 2}(\mathbb{R}^{d})$, the smoothness of $y$ corresponds to the order $\nu+d/2$ in $H^{\nu+d / 2}(\mathbb{R}^{d})$, thereby the smoothness parameter $\nu$ in the Mat\'ern kernel \citep{kanagawa2018gaussian}. 
While the domain is assumed to be $\mathbb{R}^{d}$ in the above discussion, we may use an RKHS $\mathbb{H}(\mathbb{D})$ as a restriction of $\mathbb{H}(\mathbb{R}^{d})$ and this RKHS is equipped with the kernel being the restriction of $K$ to $ \mathbb{D} \times \mathbb{D}$. Theorem 6 in \cite{berlinet2011reproducing} guarantees this statement. Before finishing this subsection, the interpretation of the smoothness in $H^{\nu+d / 2}(\mathbb{R}^{d})$ is investigated via weak differentiability. 
The definition of the norm of the Sobolev space $W^{k, p}(\mathcal{D})$ is given by
\begin{eqnarray}
\label{sobolev_norm_integer}
\|u\|_{W^{k, p}(\mathcal{D})}:=\left\{\begin{array}{ll}
\left(\sum_{|\alpha| \leq k}\left\|D^{\alpha} u\right\|_{L^{p}(\mathcal{D})}^{p}\right)^{\frac{1}{p}} & 1 \leq p<\infty \\
\max _{|\alpha| \leq k}\left\|D^{\alpha} u\right\|_{L^{\infty}(\mathcal{D})} & p=\infty
\end{array}\right.
\end{eqnarray}

If we set $k=s$ as a positive integer and $p =2$ in $W^{k, p}(\mathbb{R}^{d})$, the norm in (\ref{sobolev_norm_integer}) is equivalent of the norm in (\ref{fracsobolev}). Hence, $W^{k, 2}(\mathbb{R}^{d})$ with the norm in (\ref{sobolev_norm_integer}) is equivalent of the Sobolev space defined in $(\ref{fracsobolev})$. Then, in this case, the smoothness of the Sobolev space $(\ref{fracsobolev})$ is interpreted in a way that a function in that space is weak differentiable up to $k$th order. The definition of $(\ref{fracsobolev})$ is more general since the index $s$ may take non-integer values in $(\ref{fracsobolev})$. Therefore, we can see that the RKHS with Mat\'ern kernel $k_{\nu, \lambda, \sigma^{2}}$ consists of functions that are weak differentiable up to order $\nu+d / 2$, assuming $\nu+d / 2$ is a positive integer \cite{kanagawa2018gaussian}.

\subsection{Uniform error bounds in GPR prediction}
\label{MLMC_basis_1}
In section \ref{MLMC_kriging}, we implement an independent analysis of optimal numbers of runs mixed with the concept of multilevel Monte Carlo (MLMC). The objective of this analysis is to get a sense of an optimal number of experimental design points before running computer simulations. A clear relationship between the accuracy of an emulator and the number of training data points was provided in \citep{wang2020prediction}, based on the framework of \citep{sacks1989design}. We introduce the main result of \citep{wang2020prediction} as a preparation for the later analysis in section \ref{MLMC_kriging}. The main theorem of \citep{wang2020prediction} is then combined with the idea of MLMC, which is introduced in the next section, for our novel result of optimal numbers of evaluations of different levels of simulations.
 In a non-Bayesian framework (frequentist approach), \cite{sacks1989design} treats the target function as if it were drawn randomly from some population of functions and postulate a Gaussian process model for the distribution of functions in that population. This randomness is assumed to be characterized by the Gaussian process with the true hyperparameters in its covariance structure, which cannot be actually known. This interpretation of the non-Bayesian work in \cite{sacks1989design} is also explained in \cite{kennedy2001bayesian}. In addition, \cite{sacks1989design} explains that taking the frequentist stance leads to replace the deterministic $y(x)$ in (\ref{predictives}) by the corresponding stochastic process or GP. \cite{wang2020prediction} adopts the same philosophy of \cite{sacks1989design}. Accordingly, in our framework of this section, Mat\'ern kernels (\ref{mat}) are assumed for the covariance structure of GPR and we study the case of misspecified kernel functions. One cannot know the underlining Gaussian process whose realisations are the values of the target function. As we reviewed, the GP is determined by the covariance function. Then, this addresses a common question in GPR when the true kernel function, hence the true underlining Gaussian process, is unknown. Let $\{\nu_{0},\lambda_{0},\sigma_{0}^{2}\}$ denote the true parameters, which are usually unknown, in (\ref{mat}) and we denote the Mat\'ern kernel in (\ref{mat}) as $K_{0}$ if the parameters are equal to $\{\nu_{0},\lambda_{0},\sigma_{0}^{2}\}$. Further, let $Z(x)$ be the Gaussian process $\text{GP}(0,K_{0}(x,x^{\prime}))$ with a Mat\'ern kernel $K_{0}(x,x^{\prime})$ in $\mathbb{R}^{d}$ and we assume that the evaluations of the target function are the realisations of $Z(x)$ (a non-Bayesian framework by \cite{sacks1989design,wang2020prediction}). For $x \in \mathbb{D}$ ($\mathbb{D}$ is the input domain), we also define a Power function $P_{K,X_{N}}(x)$ as
\begin{eqnarray}
\label{power_function}
P_{K,X_{N}}^{2}(x) = 1-K\left(\boldsymbol{x}, X_{N}\right)^{\top} K(X_{N},X_{N})^{-1}  K\left(\boldsymbol{x}, X_{N}\right) / \sigma^{2}
\end{eqnarray}
The Power function $P_{K,X_{N}}(x)$ is the special case of predictive variance $\hat{s}^{2}_{N}\left(\boldsymbol{x}\right)$ where $\sigma^{2} = 1$. We further define the supremum of the Power function as 
\begin{eqnarray}
\label{sup_power}
P_{K,X_{N}} = \sup_{x \in \mathbb{D}}P_{K,X_{N}}(x)
\end{eqnarray}
The following condition for the kernel is necessary to present the upper bound on the maximum prediction error of GPR in Theorem 1 in \cite{wang2020prediction}. Condition \ref{cond_scattered} and Theorem \ref{scattered_conv} hold for a general stationary kernel but we assume the Mat\'ern kernel in our framework. In the following, since $K$ and $k_{0}$ are stationary then $K(x,x^{\prime}) = K(r)$ and $K_{0}(x,x^{\prime})=K_{0}(r)$ with $r=x-x^{\prime}$ hold.

\begin{cond}[{\cite{wang2020prediction}}]
\label{cond_scattered}
 
 The stationary kernels $K(x,x^{\prime}) = K(r)$ and $K_{0}(x,x^{\prime}) = K_{0}(r)$ ($r= x-x^{\prime}$) are continuous and integrable as a function of $r$ on $\mathbb{R}^{d}$, satisfying
\[
\|\hat{K}_{0} / \hat{K}\|_{L^{\infty}\left(\mathbb{R}^{d}\right)}=: A_{1}^{2}<+\infty
\]

where $\hat{K}_{0}$ and $\hat{K}$ are the Fourier transforms of $K_{0}$ and $K$. In addition, there exists $\tilde{\alpha} \in(0,1]$, such that
\[
\int_{\mathrm{R}^{d}}\|\omega\|^{\tilde{\alpha}} \hat{K}(\omega) d \omega=: A_{0}<+\infty.
\]
\end{cond}

Then, we can present Theorem \ref{scattered_conv} to provide an upper bound on the maximum prediction
error of GPR, which is Theorem 1 in \citep{wang2020prediction}. We treat $Z(x) - m_{N}^{Z}\left(\boldsymbol{x}\right) $ as the prediction error and $m_{N}^{Z}\left(\boldsymbol{x}\right) =K\left(\boldsymbol{x}, X_{N}\right)^{\top} K(X_{N},X_{N})^{-1}\boldsymbol{Z}(X_{N})$ as defined in (\ref{predictives}). Since we are assuming a non-Bayesian approach in this section, remember that $y(x)$ in (\ref{predictives}) is replaced with the GP $Z(x)$. It is important to see that $\sigma^{2}$ (the variance of the GP $Z(x)$) in the following upper bound represents the magnitude of the target function. 

\begin{thm}[{Theorem 1 in \cite{wang2020prediction}}]
\label{scattered_conv}
 Suppose Condition \ref{cond_scattered} holds, and the design set $X_{N}$ is dense enough in the sense that $P_{K, X_{N}}$ defined in (\ref{sup_power}) is no more than some given constant $C$. 
 For any $u>0$ with probability at least $1-2 \exp \left\{-u^{2} /\left(2 A_{1}^{2} \sigma^{2} P_{K, X_{N}}^{2}\right)\right\}$, the
prediction error has the upper bound
\begin{eqnarray}
\label{conv_wang}
\sup_{x \in \mathbb{D}}\left| Z(x) - m_{N}^{Z}\left(x\right)\right| \leq \mathcal{K} \sigma P_{K, X_{N}} \log ^{1 / 2}\left(\mathrm{e} / P_{K, X_{N}}\right)+u
\end{eqnarray}
Here the constants $C, \mathcal{K}>0$ depend only on $\mathbb{D}, \tilde{\alpha}, A_{0}$, and $A_{1}$.
\end{thm}
Here $K_{0}$ is a Mat\'ern kernel with smoothness parameter $\nu_{0}$. As discussed in \citep{wang2020prediction}, it can be verified that Condition \ref{cond_scattered} holds if and only if $0 < \nu \leq \nu_{0}$. $0 < \nu \leq \nu_{0}$ means that we set the smoothness parameter $\nu$ as the less smooth value than unknown $\nu_{0}$. If we want this convergence result to hold with probability at least $p$, $u$ should be equal to $\sigma P_{K, X_{N}} A_{1} \left(2\log{2 \over 1-p}\right)^{1/2}$. Note that the design set cannot reach density for costly simulations, so this result is rather theoretical and thus only guides us towards the use of a multilevel design that can achieve a higher density in the low fidelity instances.

\section{Multilevel adaptive sequential design of computer experiments (MLASCE)}
\label{main_body}
This section presents our novel multilevel adaptive sequential design of computer experiments (MLASCE) in the framework of multi-fidelity GPR. The novelty is an error analysis for specifying the form of the mean prediction and associated optimization of the usage in computational resources. The assumptions and formulation of this model are dealt with in Section \ref{rkhs} to \ref{smooth_para}. Section \ref{MLMC_basis_2} reviews the Multilevel Monte Carlo method.  Section \ref{er_analysis} presents an error analysis and Section \ref{algo_imp} illustrates the algorithm for implementing our method. We offer an independent analysis of optimal numbers of runs (prior to actual implementation of computer simulations) in Section \ref{MLMC_kriging}.
We present the results of our implementation of MLASCE on several numerical examples in Section \ref{nume_example}.


\subsection{Multilevel Monte Carlo (MLMC)}
\label{MLMC_basis_2}
This section presents the concept of MLMC as an ingredient of our analysis of optimal numbers of runs. Multilevel Monte Carlo (MLMC) \citep{giles2015multilevel} decomposes an expensive function into increments, whose concept is used for our analysis of optimal numbers of runs combined with the content in Section \ref{MLMC_basis_1}. In other words, the essence of MLMC, i.e., minimising the overall estimation error (variance) with respect to the numbers of samples under the limit of computational time, is employed in our framework although we do not rely on MLMC itself. As discussed in \cite{peherstorfer2018survey}, in the broad context of multi-fidelity framework, less accurate but cheap simulators are used to reduce the runtime where possible and more sophisticated and expensive simulators are employed to preserve the accuracy of the emulator. More recently, strategies \citep{schaden2020multilevel, seelinger2021high} have been introduced to overcome the limitations of approaches such as MLMC and Multi-Fidelity Monte Carlo (MFMC).  Then, a natural question is how to know the optimal numbers of runs for each simulator, preferably before implementing the simulations. In MLMC, the goal is to estimate the expectation $\mathbb{E}[P]$ for some random variable $P$ using Monte Carlo simulation while our objective is to build an emulator of the functional form of the computer model hence goes much further than the estimation of the expectation of a random variable. 

We briefly review the framework of MLMC. 
The random variables $P_{0}, \ldots, P_{L}$ are prepared and the goal is to estimate $\mathbb{E}[P_{L}]$. $P_{0}, \ldots, P_{L-1}$ approximate $P_{L}$ with increasing accuracy, but also increasing cost. Then, we have the simple identity by decomposing $P_{L}$
\[
\mathbb{E}\left[P_{L}\right]=\mathbb{E}\left[P_{0}\right]+\sum_{\ell=1}^{L} \mathbb{E}\left[P_{\ell}-P_{\ell-1}\right]
\]
and therefore we can use the following unbiased estimator for $\mathbb{E}\left[P_{L}\right]:$
\[
N_{0}^{-1} \sum_{n=1}^{N_{0}} P_{0}^{(0, n)}+\sum_{\ell=1}^{L}\left\{N_{\ell}^{-1} \sum_{n=1}^{N_{\ell}}\left(P_{\ell}^{(\ell, n)}-P_{\ell-1}^{(\ell, n)}\right)\right\}
\]
where the inclusion of the level $\ell$ in the superscript $(\ell, n)$ indicates that independent samples are used at each level of correction. Also, $N_{l}$ is the number of samples to estimate $\mathbb{E}\left[P_{\ell}-P_{\ell-1}\right]$.

If we define $C_{0}, V_{0}$ to be the cost and variance of one sample of $P_{0},$ and $C_{\ell}, V_{\ell}$ to be the cost and variance of one sample of $P_{\ell}-P_{\ell-1},$ then the overall cost and variance of the multilevel estimator are $\sum_{\ell=0}^{L} N_{\ell} C_{\ell}$ and $\sum_{\ell=0}^{L} N_{\ell}^{-1} V_{\ell},$ respectively.
For a fixed cost, the variance is minimized by choosing $N_{\ell}$ to minimize
\begin{eqnarray}
\label{MLMLprinciple}
\sum_{\ell=0}^{L}\left(N_{\ell}^{-1} V_{\ell}+\mu^{2} N_{\ell} C_{\ell}\right)
\end{eqnarray}
for some value of the Lagrange multiplier $\mu^{2}$, which is determined by the desired overall variance $\varepsilon^{2}$. This gives $N_{\ell}=\mu \sqrt{V_{\ell} / C_{\ell}}$. To achieve an overall variance of $\varepsilon^{2}$ then requires that
\[
\mu=\varepsilon^{-2} \sum_{\ell=0}^{L} \sqrt{V_{\ell} C_{\ell}}
\]

and the total computational cost is therefore
\[
C=\varepsilon^{-2}\left(\sum_{\ell=0}^{L} \sqrt{V_{\ell} C_{\ell}}\right)^{2}
\]
The formal statement of the above discussion is stated in Theorem 2.1 in \cite{giles2015multilevel}.

\subsection{Incremental decomposition and reproducing kernel Hilbert spaces (RKHS)}
\label{rkhs}
Although an elegant framework of multi-fidelity GPR or co-kriging was established in \cite{kennedy2000predicting}, this type of method fails to evaluate the improvement from simulations of higher fidelity in the sense of optimizing the accuracy of emulators and usage of computational resources. Moreover, these emulators would inherently incorporate significant prediction errors due to the specification of the mean predictive functions. Here, we construct the necessary settings in order to capture the degree of benefit of more sophisticated simulations and construct a theoretically accurate statistical emulator. One way is to decompose computer simulations into a basic approximation and incremental functions. This may seem an analogy of the correlation structure (\ref{ohagan}) of \cite{kennedy2000predicting} with $\rho_{l} = 1$ but it is essential for us to implement analysis of the error. Decomposing an expensive function into increments is not only necessary to evaluate the additional effects of sophisticated computer simulations but also a common technique in computational statistics.

We assume $y$ is a deterministic and unknown function and try to approximate it by other deterministic functions of different fidelity $y_{l}: \mathbb{D} \rightarrow \mathbb{R}$ for level $l=1, \ldots, L$, which are called computer simulations. Each $y_{l}$ has its own degree of accuracy denoted as $h_{l}$ and cost $t_{l}$. $h_{l}$ satisfies $1 \geq h_{1} > h_{2} > \cdots > h_{L}$. $h_{l}$ is an ordinal number and should be decided by the user in proportion to different accuracy of computer simulations. For example, $h_{l}$ can be determined by the proportion of different discretization levels in a numerical solver. The actual value of $h_{l}$ is not used in MLASCE but is present only in the next section. The cost $t_{l}$ satisfies $t_{L} > t_{L-1} > \cdots > t_{1}$ hence $y_{l_{1}}$ is more accurate but expensive than $y_{l_{2}}$ for different levels $l_{1} > l_{2}$. Our goal is to recover the most reliable simulation $y_{L}$. For this purpose, we introduce the discrepancy terms $\delta_{l}:\mathbb{D} \rightarrow \mathbb{R}$ for $l = 2, \ldots, L$ so $y_{l}(\boldsymbol{x}) = y_{l-1}(\boldsymbol{x}) + \delta_{l}(\boldsymbol{x})$ where $\boldsymbol{x}$ is an arbitrary element in $\mathbb{D}$.
Using this expression, $y_{L}$ is decomposed into the sum of the basic approximation and increments: 
\begin{eqnarray}
\label{decomp}
y_{L}(\boldsymbol{x}) = \sum_{l=1}^{L} \delta_{l}(\boldsymbol{x})
\end{eqnarray}
where $\delta_{1}(\boldsymbol{x}) = y_{1}(\boldsymbol{x})$ and $\delta_{l}(\boldsymbol{x})= y_{l}(\boldsymbol{x}) - y_{l-1}(\boldsymbol{x})$ for $ l=2, \ldots, L$. We can see that $\delta_{1}$ is the baseline approximation and $\{ \delta_{l} \}_{l=2}^{L}$ are the incremental refinements. It would be usually expected that the value of increment $\delta_{l}(\boldsymbol{x})$ should decrease as the level $l$ increases thus less contribution comes from higher fidelity simulations. In order to evaluate the degree of improvements accurately, we first assume each $\delta_{l}$ is an element of a reproducing kernel Hilbert space (RKHS) $\mathbb{H}_{l}$ with a reproducing kernel $K_{l}$ for $l=1, \ldots, L$. Let $\boldsymbol{\delta_{l}}(X_{l,N_{l}}) = \left( \delta_{l}(x_{l,1}), \delta_{l}(x_{l,2}) \ldots, \delta_{l}(x_{l,N_{l}}) \right)^{\top}$ denote $N_{l}$ outputs of $\delta_{l}$ collected at the inputs $X_{l,N_{l}}=\left(x_{l,1}, x_{l,2}, \ldots, x_{l,N_{l}}\right)^{\top}$. Next, we build a new multi-fidelity GPR by assuming that $\boldsymbol{\delta_{l}}(X_{l,N_{l}})$ is a realization of the GP $\eta_{l}(\boldsymbol{x}) \sim GP \left( 0, K_{l}(\boldsymbol{x}, \boldsymbol{x}^{\prime}) \right)$, where every $\eta_{l}$ is independent of each other for different $l$. Based on the experimental design or training data $X_{l,N_{l}}$ and $\boldsymbol{\delta_{l}}(X_{l,N_{l}})$, we obtain the GPR prediction at $x$ consisting of the posterior distribution $\eta_{l}\left(\boldsymbol{x}\right)|X_{l,N_{l}},\boldsymbol{\delta_{l}}(X_{l,N_{l}})  \sim \mathcal{N}\left(m_{l,N_{l}}^{\delta_{l}}\left(\boldsymbol{x}\right), \hat{s}^{2}_{l,N_{l}}\left(\boldsymbol{x}\right)\right)$ with posterior mean and variance given by
\[
\begin{aligned}
m_{l,N_{l}}^{\delta_{l}}\left(\boldsymbol{x}\right) &=K_{l}\left(\boldsymbol{x}, X_{l,N_{l}}\right)^{\top} K_{l}(X_{l,N_{l}},X_{l,N_{l}})^{-1}\boldsymbol{\delta_{l}}(X_{l,N_{l}}) \\ 
\hat{s}^{2}_{l,N_{l}}\left(\boldsymbol{x}\right) &=K_{l}(\boldsymbol{x}, \boldsymbol{x})-K_{l}\left(\boldsymbol{x}, X_{l,N_{l}}\right)^{\top} K_{l}(X_{l,N_{l}},X_{l,N_{l}})^{-1} K_{l}\left(\boldsymbol{x}, X_{l,N_{l}}\right)
\end{aligned}.
\]

Therefore, we can construct the prediction for $y_{L}(\boldsymbol{x})$ by taking note that every $\eta_{l}$ is independent of each other and summing up $\eta_{l}(\boldsymbol{x})| X_{N_{l}},\boldsymbol{\delta_{l,N_{l}}}$. The predictive distribution of $y_{L}(x)$ based on our multi-fidelity GPR is given by
\begin{eqnarray}
\label{mymodel}
\mathcal{N}\left(m^{y_{L}}_{L}(\boldsymbol{x}), s^{2}_{L}(\boldsymbol{x}) \right)
\end{eqnarray}
where $m^{y_{L}}_{L}(\boldsymbol{x}) = \sum_{l=1}^{L} m_{l,N_{l}}^{\delta_{l}}\left(\boldsymbol{x}\right)$ and $s^{2}_{L}(\boldsymbol{x}) = \sum_{l=1}^{L} \hat{s}^{2}_{l,N_{l}}\left(\boldsymbol{x}\right)$. 

The posterior mean function $m_{l,N_{l}}^{\delta_{l}}\left(\boldsymbol{x}\right)$ is used for predicting $\delta_{l}(\boldsymbol{x})$ and $\delta_{l}$ is assumed to be an element of the corresponding RKHS with the RK $K_{l}$. When we use $m^{y_{L_{1}}}_{L_{1}}(\boldsymbol{x}) = \sum_{l=1}^{L_{1}} m_{l,N_{l}}^{\delta_{l}}\left(\boldsymbol{x}\right)$ as the prediction value at $x$ for $ 1 \leq L_{1} \leq L$, we call it the level $L_{1}$ emulator. We further expect that it would converge to the true response $y_{L}$ pointwise and in the RKHS norm if the number of data points tends to infinity. Theorem \ref{conv} guarantees this desired property and this is going to be discussed later.

If we use a non-zero mean $m_{l}(\cdot)$ in GP $\eta_{l}$, the formula for the predictive mean $m_{l,N_{l}}^{\delta_{l}}$ changes to
$$
m_{l,N_{l}}^{\delta_{l}}\left(\boldsymbol{x}\right) = m_{l}(\boldsymbol{x}) + K_{l}\left(\boldsymbol{x}, X_{l,N_{l}}\right)^{\top} K_{l}(X_{l,N_{l}},X_{l,N_{l}})^{-1} \left( \boldsymbol{\delta_{l}}(X_{l,N_{l}}) - m_{l}\left(X_{l,N_{l}}\right) \right)
$$
where $m_{l}\left(X_{l,N_{l}}\right):=\left(m_{l}\left(x_{l,1}\right), \ldots, m_{l}\left(x_{l,N_{l}}\right)\right) \in \mathbb{R}^{N_{l}}$ and the predictive covariance is unchanged. The benefit of specifying a non-zero mean function is that one could reduce ``the amount of learning'' by considering the discrepancy $m_{l}-\delta_{l}$ instead of $\delta_{l}$ itself if a specific form of $m_{l}$ is promising. On the other hand, parametrising $m_{l}$ would require a large number of hyperparameters especially in the case of high dimensional inputs. 

\subsection{Choice of smoothness parameter}
\label{smooth_para}
As we can see in the Proposition \ref{stuart}, the smoothness parameter decides the (equivalent) underlining function space to which the target function is assumed to belong. As we reviewed, if we assume $y \in H^{\nu+d / 2}(\mathbb{R}^{d})$ for the target function $y$, the smoothness of $y$ corresponds to the order $\nu+d/2$ in $H^{\nu+d / 2}(\mathbb{R}^{d})$, hence the smoothness parameter $\nu$ in the Mat\'ern kernel. However, choosing an appropriate value for the smoothness parameter $\nu$, which specifies the smoothness of the target function via weak differentiability (reviewed in earlier), may be a difficult task. The hyperparameters in a covariance function is usually unknown. We might be able to estimate $\nu$ by maximum likelihood method or other statistical methods but a sufficient amount of data may not be available to obtain a reliable result, which is a typical case where the computer simulation is expensive to run. Moreover, evaluating the Mat\'ern kernel with $\nu$ other than $0.5, 1.5, 2.5, \infty$ numerous times using the maximum likelihood method and experimental design is expensive, and it is common to fix the smoothness parameter $\nu$ before implementing the computation in the packages of statistical software. Therefore, we usually have to set the value for $\nu$ in advance. For the different smooth parameters $\nu$ and $\nu^{\prime}$ $(\nu < \nu^{\prime})$, the corresponding Sobolev spaces satisfy $H^{\nu^{\prime}}(\mathbb{R}^{d})  \subseteq H^{\nu}(\mathbb{R}^{d})$ hence one possible strategy of choosing $\nu$ is to set a small value, for $y \in H^{\nu^{\prime}}(\mathbb{R}^{d})$ implies $y \in H^{\nu}(\mathbb{R}^{d})$.

Some papers investigate the effect of misspecified smoothness parameters on the GPR approximation in the context of space filling design. The convergence rate of the posterior mean of GPR becomes slower if the (fixed) smoothness parameters are different from the true value \citep{wang2020prediction,tuo2020kriging}. Similar results \citep{teckentrup2020convergence,wynne2021convergence} are available in the case that the estimated smoothness parameter is used. Although we do not use space filling design for locating the training data, these results are insightful for choosing the smoothness parameter.     

\subsection{Error analysis}
\label{er_analysis}
 In this subsection, we show that our multi-fidelity GPR is based on a reasonable specification of the mean predictive function for approximating the target function. We have to remember that the co-kriging modelling, which is commonly used multi-fidelity GPR, lacks a theoretical background for its specification of the mean predictive function. Without an appropriate theoretical guarantee, its predictive value might not get close to the true response even as the number of the training data points grows. Therefore, the mean prediction of the co-kriging may tend to deviate from the target, as we reviewed in the last chapter. In contrast, our multi-fidelity emulator is based on a theoretically verified specification of the mean predictive function.
 
 In our framework, the training data is obtained by implementing MICE as the design of experiment. The estimated parameters of $\lambda, \sigma^{2}$ and a priori fixed $\nu$ are used in the Mat\'ern function. The estimation of $\lambda, \sigma^{2}$ is based on the data obtained by design of experiment. It is assumed that the observation vector $\boldsymbol{y}$ is a realization of the Gaussian process $Y(\boldsymbol{x}) \sim GP\left(0, K_{\nu,\theta}\left(\boldsymbol{x}, \boldsymbol{x}^{\prime}\right)\right)$ where $K_{\nu,\theta} = K_{\nu, \lambda, \sigma^{2}}$ is a Mat\'ern kernel defined in $(\ref{mat})$ and $\theta$ denotes $\{\lambda, \sigma^{2}\}$ for convenience. The GPR prediction of $y(x)$ at $x$ is obtained in the form of the normal distribution $\mathcal{N}\left(m_{N,\theta}^{y}\left(x\right), \hat{s}^{2}_{N,\theta}\left(x\right)\right),$ with posterior mean and variance given by
\[
\begin{aligned}
m_{N,\theta}^{y}\left(\boldsymbol{x}\right) &=K_{\nu, \theta}\left(\boldsymbol{x}, X_{N}\right)^{\top} K_{\nu, \theta}(X_{N},X_{N})^{-1}\boldsymbol{y}(X_{N}) \\
\hat{\mathrm{s}}^{2}_{N,\theta}\left(\boldsymbol{x}\right) &=K_{\nu, \theta}(\boldsymbol{x}, \boldsymbol{x})-K_{\nu, \theta}\left(\boldsymbol{x}, X_{N}\right)^{\top} K_{\nu, \theta}(X_{N},X_{N})^{-1} K_{\nu, \theta}\left(\boldsymbol{x}, X_{N}\right)
\end{aligned}
.\]
We are interested in a pointwise prediction of deterministic $y\left(\boldsymbol{x}\right)$ and it is a natural question to investigate the error between the $y\left(\boldsymbol{x}\right)$ and its mean prediction $m_{N,\theta}^{y}\left(\boldsymbol{x}\right)$. To formulate this analysis, we let $\Theta$ be a compact set in $\mathbb{R}^{2}_{+}$ ($\mathbb{R}^{2}_{+}$ denotes the subset of $\mathbb{R}^{2}$ composed of positive values) and $y \in \mathbb{H}$ for an arbitrary $\theta \in \Theta$ where $\mathbb{H}$ is an RKHS with the RK $K_{\nu, \theta}$. For any $\boldsymbol{x} \in \mathbb{D}$ and $\theta \in \Theta$, we can easily check that the error of the mean prediction of the deterministic $y \in \mathbb{H}$ is given by (see e.g., Theorem 11.4 in \citep{wendland2004scattered} for a proof)
\begin{eqnarray}
\label{error1}
|y\left(\boldsymbol{x}\right)- m_{N,\theta}^{y}\left(\boldsymbol{x}\right) | \leq  \| y\|_{\mathbb{H}} \bigl(  \hat{\mathrm{s}}^{2}_{N}\left(\boldsymbol{x}\right) \bigr)^{\frac{1}{2}} .
\end{eqnarray}
This inequality neglects the modeling error of the GP. e.g. when too little training data are available or when the assumption of a smooth
functional mapping which be violated (for which we suggest the use of deep GPs \citep{sauer2021active,ming2021deep}. Moreover, the error converges to $0$ pointwise and in the RKHS norm if an enough number of data points is available. Although it may be argued that the number of data points for an expensive computer simulation is limited in practice, the following theorem still ensures that our specification of the mean predictive function is the right choice at least. The proof is given in the Appendix. 

\begin{thm}
\label{conv}
Let $\mathbb{H}_{\theta}$ denote an RKHS defined in $\mathbb{D} \subseteq \mathbb{R}^{d}$ with the reproducing kernel $K_{\nu, \theta} = K_{\nu, \lambda, \sigma^{2}}$ defined in (\ref{mat}) and $\mathbb{R}^{2}_{+}$ denote the subset of $\mathbb{R}^{2}$ composed of positive values. Further, we assume the following conditions.

\begin{enumerate}
  \item $\Theta$ is a compact set in $\mathbb{R}^{2}_{+}$.
  \item $y \in \mathbb{H}_{\theta}$ for an arbitrary $\theta \in \Theta$ and a fixed $\nu$.
  \item $X = \{x_{i} : 1 \leq i \leq \infty \}$ is a set of points in $\mathbb{D}$ such that $\text{span} \{ K_{\nu, \theta}(\cdot, x_{i}) : 1 \leq i \leq \infty \}$ is dense in $\mathbb{H}_{\theta}$ and $y(X) = \{y(x_{i}) : 1 \leq i \leq \infty \}$ is the values of $y$ located at $X$. 
  \item $X_{N}=\left(x_{1}, x_{2}, \ldots, x_{N}\right)^{\top}$ and $\boldsymbol{y} (X_{N}) = \bigl( y(x_{1}), y(x_{2}) \ldots, y(x_{N}) \bigr)$ are the subsets of $X$ and $y(X)$.
 \item $K_{\nu, \theta}(X_{N},X_{N})^{-1}$ exists for any $\theta \in \Theta$ and $X_{N}$.
\end{enumerate}

Then, the posterior covariance $K_{\nu, \theta}(x, u)- K_{\nu, \theta}(x, X_{N})K_{\nu, \theta}(X_{N}, X_{N})^{-1} K_{\nu, \theta}(X_{N}, u) \to 0$ and posterior mean $m_{N,\theta}^{y}\left(\boldsymbol{x}\right) = K_{\nu, \theta}(x, X_{N})K_{\nu, \theta}(X_{N}, X_{N})^{-1} y(X_{N}) \to y(x)$ for any $x, u \in \mathbb{D}$, a fixed $\nu$ and $\theta \in \Theta$ as $N \to \infty$. Furthermore, $\|m_{N,\theta}^{y} -y\|_{\mathbb{H}_{\theta}} \to 0$ holds under the same conditions.
\end{thm}

In practice, the hyperparameters $\theta$ or $\{\lambda,\sigma^{2}\}$ in $K_{\nu,\theta}$, which are usually unknown, are identified by the maximum likelihood estimator $\hat{\theta}_{N}$, whereas $\nu$ is set to a specific fixed value in advance. One might argue that the changing value of $\widehat{\theta}_{N}$ in  $m_{N,\hat{\theta}_{N}}^{y}\left(\boldsymbol{x}\right)$ as $N$ grows might not be trivial but it can be verified that $m_{N,\hat{\theta}_{N}}^{y}\left(\boldsymbol{x}\right)$ also converges to $y(x)$ if we further assume a mild condition on $m_{N,\theta}^{y}\left(\boldsymbol{x}\right)$ and that $\{ \hat{\theta}_{N} \}_{N=1}^{\infty}$ is in $\Theta$ with a fixed limit $\theta_{0} \in \Theta$ (see Remark 2 in Appendix).

We now apply the above results to our multi-fidelity GPR. First, we prepare Mat\'ern kernels $K_{\nu_{l},\lambda_{l},\sigma^{2}_{l}}$ with fixed $\nu_{l}$ for $\eta_{l}(x)$ ($l=1, \ldots, L$). We estimate $\{\lambda_{l}, \sigma_{l}^{2}\}$ based on the design $\{ X_{l,N_{l}}, \boldsymbol{\delta}(X_{l,N_{l}})\}$ accordingly and set these values to the hyperparameters in $K_{\nu_{l},\lambda_{l},\sigma^{2}_{l}}$. We plug these Mat\'ern kernels in the covariance structure of $(\ref{mymodel})$. We also assume that $\delta_{l} \in \mathbb{H}_{l}$ with the RK $K_{\nu_{l},\lambda_{l},\sigma^{2}_{l}}$ for a chosen $\nu_{l}$ and any $\lambda_{l},\sigma^{2}_{l} $ in a compact set $\Theta_{l} \subseteq \mathbb{R}^{2}_{+}$. Further, the estimated hyperparameters of $\lambda_{l},\sigma^{2}_{l}$ are in $\Theta_{l}$. Then, we derive an error bound for our mean prediction $m^{y_{L}}_{L}(x)$ in $(\ref{mymodel})$.
\begin{eqnarray}
\label{myerror}
|m^{y_{L}}_{L}(\boldsymbol{x})  - y_{L}(\boldsymbol{x})| \leq \sum_{l=1}^{L}|m_{l,N_{l}}^{\delta_{l}}(x) - \delta_{l}(x)|
\leq \sum_{l=1}^{L} \bigl( \hat{s}^{2}_{l,N_{l}}\left(\boldsymbol{x}\right) \bigr)^{\frac{1}{2}} \| \delta_{l}\|_{\mathbb{H}_{l}}
\end{eqnarray}
We write $m_{N,\theta}^{y}$ as $m_{N}^{y}$ and $K_{\theta}(\cdot, \cdot)$ as $K(\cdot, \cdot)$ when there is no need for expressing the dependency on specific $\theta$. 

Further, Theorem \ref{conv} applies to each $|m_{l,N_{l}}^{\delta_{l}}\left(\boldsymbol{x}\right) - \delta_{l}(\boldsymbol{x})|$ hence we are able to make the error as small as possible if a sufficient amount of data can be obtained and the functional relationship is amenable to a GP fitting, as we mentioned before. In practice, we usually cannot afford to run expensive simulations to obtain enough data. Therefore, the next task is optimizing the number of training data points $\{X_{l,N_{l}}, \delta_{l}(X_{l,N_{l}})\}$ across every $l$ under a constraint of computational resources to achieve a small upper bound in $(\ref{myerror})$. 

\subsection{Algorithm for MLASCE and optimal usage of computational resources}
\label{algo_imp}

Let us go back to the error bound in $(\ref{myerror})$. Our goal is to implement the method of experimental design so that the above upper bound is made as small as possible under a restriction of the computational resources. That is, we prepare the training data $\{X_{l,N_{l}}, \boldsymbol{\delta}_{l,N_{l}}(X_{l,N_{l}})\}$ for all levels $l=1, \ldots, L$ in a way that the error bound becomes as small as possible. The error bound depends on not only the number of data points $N_{l}$ but also the locations of the data $X_{l,N_{l}}$. However, evaluating it in an explicit form without specific knowledge is a rather formidable task. Thus, an optimisation strategy should be developed. The essence of our strategy is to choose the most effective level $l$ for reducing the overall error in $(\ref{myerror})$ and run $\delta_{l}$ to obtain the new training data point $x_{l,N_{l}+1}$ and $\delta_{l}(x_{l,N_{l}+1})$ using MICE for deciding the location of $x_{l,N_{l}+1}$.

 Although $\hat{s}^{2}_{l,N_{l}}\left(x\right) ^{\frac{1}{2}} \| \delta_{l}\|_{\mathbb{H}_{l}}$ tends to 0 for every $l$ as $N_{l}$ goes to infinity (assuming the density), its weight in the overall upper bound (\ref{myerror}) is different because of the magnitude of $\| \delta_{l}\|_{\mathbb{H}_{l}}$. Thus, it is sensible to get more samples at the level where the reduction in $\hat{s}^{2}_{l,N_{l}}\left(x\right) ^{\frac{1}{2}} \| \delta_{l}\|_{\mathbb{H}_{l}}$ is large as $N_{l}$ increases. This is to ensure that the upper bound of (\ref{myerror}) becomes smaller quickly. We propose the optimisation problem for our purpose instead of considering the upper bound of $(\ref{myerror})$ directly: maximizing $\sum_{l=1}^{L} \|m_{l,N_{l}}^{\delta_{l}}\|_{\mathbb{H}_{l}}^{2}$ under the restriction of $\sum_{l=1}^{L}N_{l}t_{l} = T$. Indeed, $\|m_{l,N_{l}}^{\delta_{l}}\|_{\mathbb{H}_{l}}^{2}$ contains the information of both $\bigl( \hat{s}^{2}_{l,N_{l}}\left(\boldsymbol{x}\right) \bigr)^{\frac{1}{2}}$ and $ \| \delta_{l}\|_{\mathbb{H}_{l}}$. The validation of this formulation is as follows. 

First, the pointwise convergence of $m_{l,N_{l}}^{\delta_{l}}(x)$ is equivalent to that of  $\hat{s}^{2}_{l,N_{l}}\left(x\right) ^{\frac{1}{2}} $ as discussed in the proof of Theorem \ref{conv} and the convergence of $\|m_{l,N_{l}}^{\delta_{l}}\|_{\mathbb{H}_{l}}$ in the RKHS norm guarantees the pointwise convergence of $m_{l,N_{l}}^{\delta_{l}}(x)$ (e.g. Theorem  2.7.6 \citep{hsing2015theoretical}). The convergence of $m_{l,N_{l}}^{\delta_{l}}$ in the RKHS norm is also guaranteed in Theorem \ref{conv}. Moreover, since $\|m_{l,N_{l}}^{\delta_{l}}\|_{\mathbb{H}_{l}}$ is bounded by $\| \delta_{l}\|_{\mathbb{H}_{l}}$ and a non-decreasing function of $N_{l}$ (see Proposition 3.1 in \cite{teckentrup2020convergence}), increasing the value of $\|m_{l,N_{l}}^{\delta_{l}}\|_{\mathbb{H}_{l}}$ with growing $N_{l}$ leads to the smaller $\hat{s}^{2}_{l,N_{l}}\left(x\right) ^{\frac{1}{2}}$. 

Second, $\| \delta_{l}\|_{\mathbb{H}_{l}}$ may be generally unknown but $\|m_{l,N_{l}}^{\delta_{l}}\|_{\mathbb{H}_{l}}$ can be used as a substitute of $\|\delta_{l} \|_{\mathbb{H}_{l}}$ albeit with an approximation error. Indeed, the proof of Corollary 10.25 in \citep{wendland2004scattered} provides the equation: $\| \delta_{l}\|_{\mathbb{H}_{l}}^{2} =\|  m_{l,N_{l}}^{\delta_{l}}\|_{\mathbb{H}_{l}}^{2}+ \| \delta_{l} - m_{l,N_{l}}^{\delta_{l}}\|_{\mathbb{H}_{l}}^{2}$. Hence, $\|m_{l,N_{l}}^{\delta_{l}}\|_{\mathbb{H}_{l}}$ also includes the information of the weight $\| \delta_{l}\|_{\mathbb{H}_{l}}$ then we can consider the magnitude of $\| \delta_{l}\|_{\mathbb{H}_{l}}$ as well as the convergence of $\hat{s}^{2}_{l,N_{l}}\left(x\right) ^{\frac{1}{2}}$ at the same time by focusing on $\|m_{l,N_{l}}^{\delta_{l}}\|_{\mathbb{H}_{l}}$ only.

Finally, we can compute the actual value of $\|m_{l,N_{l}}^{\delta_{l}}\|_{\mathbb{H}_{l}}$ with ease by the following relationship:
\[\|m_{l,N_{l}}^{\delta_{l}}\|_{\mathbb{H}_{l}}^{2} =  \boldsymbol{\delta}_{l,N_{l}}(X_{l,N_{l}})^{\top} K_{l}(X_{l,N_{l}},X_{l,N_{l}})^{-1}\boldsymbol{\delta}_{l,N_{l}}(X_{l,N_{l}}).\] 
Measuring the uncertainty of the mean prediction over the entire input domain would be difficult if we use the value of the prediction variance $\hat{s}^{2}_{l,N_{l}}\left(x\right)$ only; the value of $\hat{s}^{2}_{l,N_{l}}\left(x\right)$ is varied depending on the location of $x$. If $x$ is close to the design point, $\hat{s}^{2}_{l,N_{l}}\left(x\right)$ should be small and vice versa. Therefore, $\|m_{l,N_{l}}^{\delta_{l}}\|_{\mathbb{H}_{l}}$ is a candidate for a concise measure of uncertainty.

Our optimisation problem is similar to that in the context of nonlinear knapsack problem or resource allocation problems \cite{morin1976algorithm,bretthauer2002nonlinear} and our optimisation strategy can be based on the greedy methods, which is one of the common approaches in this discipline (for example, see \cite{marsten1978hybrid,bretthauer1995nonlinear, d2011heuristic}). In our case, the change in the value $\|m_{l,N_{l}}^{\delta_{l}}\|_{\mathbb{H}_{l}}^{2}$ as $N_{l}-1$ is increased to $N_{l}$ is seen as a criterion or ``score''. When increasing the number of points $N_{l}$, we decide which $l$ should be chosen based on the score of every level and continue this process recursively.

Let us define the ``score'' $\gamma_{l,N_{l}}$ by
\begin{eqnarray}
\label{score}
\gamma_{l,N_{l}} \coloneqq  | \|m_{l,N_{l}}^{\delta_{l}}\|_{\mathbb{H}_{l}}^{2} -\|m_{l,N_{l}-1}^{\delta_{l}}\|_{\mathbb{H}_{l}}^{2} | \times \frac{a_{l}}{t_{l}}
\end{eqnarray}
where $a_{l}$ is a tuning parameter that decides the degree of the penalty of the cost $t_{l}$. It takes a positive value up to $t_{l}$ and a smaller value means less emphasis on that level. The idea is to see the score as a benefit of running $\delta_{l}$ weighted by its cost $t_{l}$. This kind of heuristic strategy is widely proposed as a feasible solution to the nonlinear knapsack problem such as in \cite{marsten1978hybrid,bretthauer1995nonlinear, d2011heuristic}. 

We would recommend $a_{l} = t_{l}$ (the maximum value of $a_{l}$) as a first attempt. If the computational costs $t_{l}$ of the simulations of the higher levels are much more expensive than those of the lower levels and tiny values of $a_{l}$ are set for those of the higher levels, this would lead to very small scores for these higher levels due to the small $\frac{a_{l}}{t_{l}}$. As a result, too much emphasis on the cheaper simulations would occur. As we shall see later, another existing strategy of sequential design \cite{le2015cokriging} yields inefficient emulators for the same reason. Our algorithm with $a_{l} = t_{l}$ for every $l$ (the score $\gamma_{l,N_{l}}$ becomes $|\|m_{l,N_{l}}^{\delta_{l}}\|_{\mathbb{H}_{l}}^{2} -\|m_{l,N_{l}-1}^{\delta_{l}}\|_{\mathbb{H}_{l}}^{2}|$) automatically tends not to pick up the higher levels of simulations too frequently since we would expect that $y_{l}(x) \simeq y_{l-1}(x)$, hence $\delta_{l}(x) \simeq 0$, for a bigger $l$, which implies the overall difference between the two high-fidelity simulations is small, and $\gamma_{l,N_{l}}$ takes a smaller value accordingly. In the case that $y_{l}(x)$ behaves very differently than $y_{l-1}(x)$ does, which can happen in tsunami dynamics simulations with kinky movements, a wisely chosen weight $a_{l}$ may improve the algorithm. 

The available computational budget is denoted as $T$, often the available computational time. We have to pay the computational cost $t_{l}$ for running $\delta_{l}$ and $T$ is updated to $T-t_{l}$ accordingly if it is run. We repeat this process until the computational budget depletes. The score $\gamma_{l,N_{l}}$ is used for computing the potential reduction in every level. We choose the level where $\gamma_{l,N_{l}}$ is the largest. It should be noticed that $\delta_{l}$ cannot be run if $t_{l}$ becomes larger than $T$. Thus, we have to grasp which computer simulations are available at each step. The summary of our strategy, MLASCE, is as follows.

\begin{description}

  \item[Preparation (1)] For $l=1, \ldots, L$, we prepare the computer simulation $\delta_{l}$ with the computational cost $t_{l}$ for each run ($t_{1} < \cdots <t_{L}$), discretised input domain  $\mathrm{D}_{G,l} \subseteq \mathbb{R}^{d}$ and candidate set $X_{\text{cand}, l} \subseteq \mathrm{D}_{G,l}$. We denote the total computational budget $T$ and $ X_{l, k_{l}}\subseteq X_{\text{cand}, l} $ the experimental design at level $l$ of size $k_{l}$.

  \item[Preparation (2)] We prepare GP$(0, K_{l}(\cdot, \cdot))$ for $l=1, \ldots, L$. $ K_{l}$ is a Mat\'ern kernel with some fixed $\nu_{l}$ and the other hyperparameters $\sigma_{l}$ and $\lambda_{l}$ updated during the procedure of MICE. We first set the score $\gamma_{l,X_{l,0}} = 0$ before obtaining any data.

\item[Step 1] By implementing random sampling for every level, we obtain the initial design (training data) with the size $1$ for every $l$. That is, let $N_{l} =1 $, sample $\bigl( X_{l, 1}, \delta_{l}(X_{l, 1}) \bigr)$ and the hyperparameters $\sigma^{2}_{l}$ and $\lambda_{l}$ of $K_{l}$ are updated based on the obtained data.

  \item[Step 2] Calculate $\gamma_{l,X_{l,k_{l}}}$ based on the current design $\bigl( X_{l, k_{l}}, \delta_{l}(X_{l, k_{l}}) \bigr)$ for every $l$. The candidate sets are updated accordingly so that $X_{\text{cand}, l}  \leftarrow X_{\text{cand}, l}   \setminus  X_{l, k_{l}} $ for every $l$.

   \item[Step 3-1]  We choose the level $l$ where the score $\gamma_{l,X_{l,k_{l}}}$ is the largest under the condition $t_{l} \leq T$. The corresponding computational cost $t_{l}$ should be smaller than the current budget.

    \item[Step 3-2] At the chosen level $l$, sample a next training point $\{ x_{l}, \delta_{l}(x_{l}) \}$ using MICE. Set $X_{l, k_{l}+1} \leftarrow X_{l, k_{l}} \cup x_{l}$. The hyperparameters $\sigma^{2}_{l}$ and $\lambda_{l}$ are updated by fitting the new data $\{X_{l, k_{l}+1}, \delta_{l}(X_{l, k_{l}+1}) \}$ to GP$(0, K_{l}(\cdot, \cdot))$.

  \item[Step 4] If $T<t_{1} $, then stop; otherwise go to Step 2.
   \item[Output] : The experimental design $\{ X_{l,N^{*}_{l}}, \delta_{l}(X_{l,N^{*}_{l}}) \}$ for every $l$ and 
   the predictive distribution (\ref{mymodel}). 
   
\end{description}

By implementing this algorithm, the Gaussian process emulator is automatically built since the hyperparameters $\sigma^{2}_{l}$ and $\lambda_{l}$ are updated whenever new data is sampled. Obviously the model error will be large at the beginning with few points available to fit the GP and acts as an agnostic random design, like an LHD. Howver, the design strategy progressively improves its performance to deliver gains relatively quickly in our examples. 
\begin{remark}
\label{mlasce_algo}
In the above explanation, only 1 point is assumed for the initial design of every level and the algorithm increases just 1 point for the chosen level at Step 3-2. But, these settings can be modified flexibly. The number of the initial data point for the emulator of each level may be determined by the user; instead of fixing it to 1 point, one would choose an arbitrary number of points if the budget permits. Moreover, one may increase an arbitrary number of points at Step 3-1 and 3-2. For instance, if MLASCE chooses level $l_{1}$ and $l_{2}$ subsequently, then one can increase, say, 3 design points for both level $l_{1}$ and $l_{2}$ at each step. 
\end{remark}

\begin{figure}[ht]
\begin{center}
\includegraphics[width=0.4\linewidth]{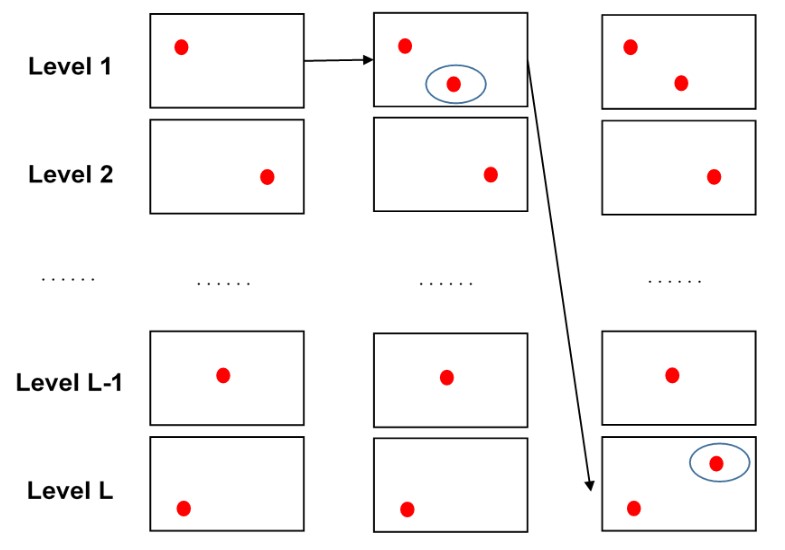}
\caption{Workflow of MLASCE. New design point is generated at level 1 first and next at level $L$.}
\label{plan}
\end{center}
\end{figure}

\section{The special case of a uniform tensor grid: optimal numbers of runs}
\label{MLMC_kriging_1}
Although MLASCE produces an efficient statistical emulator, the number of data points for each computer simulation is actually unknown until the algorithm finishes its computation. It would be beneficial to get a sense of an optimal number of experimental design points before running computer simulations, albeit in the case of space filling design. This can be the case if the locations of the design are specified by prior knowledge of a system. In this section, an analysis of optimal numbers of runs (prior to actual implementation of computer simulations) is provided and this is independent of the algorithm of MLASCE presented in the previous sections. We again assume that the evaluations of the each increment $\delta_{l}(x)$ are the realizations of the mutually independent GP $\eta_{l}(x) \sim \text{GP}(0,K_{l}(x,x^{\prime}))$. The scheme presented in this part is similar to that of MLMC (\ref{MLMLprinciple}): the error of the emulator is expressed as a function of the number of data points $N_{l}$ and it is minimized under a restriction. The assumptions and notations are same as in the previous section. The main principle of this analysis is composed of (i) deriving a convergence rate of $\sum_{l=1}^{L}|\eta_{l}(x)-m_{l,N_{l}}^{\eta_{l}}\left(x\right) |$ in (\ref{myerror_2}) with respect to the number of runs $N_{l}$ and (ii) minimising this upper bound as a function of $N_{l}$. Since the realizations $\eta(X_{l,N_{l}})$ are $\delta_{l}(X_{l,N_{l}})$, we take note that $m_{l,N_{l}}^{\delta_{l}}\left(x\right) = m_{l,N_{l}}^{\eta_{l}}\left(x\right)$. There is rich literature regarding obtaining the convergence rate in the field of scattered data approximation (e.g. \cite{wendland2004scattered,wendland2005approximate}). An approximation is constructed using radial basis functions originating from reproducing kernels and the approximation coincides with the mean of a conditioned GP. The results from this literature are hence directly transferable to approximation with GPs. The work of \cite{wang2020prediction} and \cite{teckentrup2020convergence} contains the most up-to-date application of scattered data approximation results for GPR. In \cite{wang2020prediction}, an upper bound of worst-case prediction error is given, which is to be applied to our analysis. 

We start from the posterior prediction error of the GPs $\eta_{l}(x)$ in (\ref{myerror_2}). We use Mat\'ern kernels (\ref{mat}) with hyperparameters $\{\nu_{l}, \lambda_{l},\sigma_{l}^{2}\}$ again for $\eta_{l}(x)$ and $m_{l,N_{l}}^{\eta_{l}}\left(x\right)$ for every $l$ and assume the case of misspecified kernel functions. This analysis of optimal numbers of runs is implemented before obtaining training data and there is no clue to know the actual values of parameters in the kernels hence this assumption is crucial. The outline of this study is (i) the error bound in (\ref{myerror_2}) is converted into the function of Power functions based on Theorem \ref{scattered_conv}, (ii) Power functions are further converted into the fill distance and finally (iii) the fill distance is translated into the $N_{l}$. 
\begin{eqnarray}
\label{myerror_2}
\sum_{l=1}^{L} | \eta_{l}(\boldsymbol{x}) - m_{l,N_{l}}^{\eta_{l}}\left(\boldsymbol{x}\right) | 
\end{eqnarray}
Then, we present the theorem as the result of the above direction. To simplify the notation, we write $a \lesssim b$ for two positive quantities $a$ and $b,$ if $a / b$ is uniformly bounded independent of the number of points $N_{l}$ and the degree of accuracy $h_{l}$ of $\delta_{l}$. We also write $a \simeq b,$ if $a \lesssim b$ and $b \lesssim a$. The definition of the fill distance is given after this theorem.

\begin{thm}
\label{mlmc_bound}
We assume the following conditions.

\begin{enumerate}
\item $\mathbb{D} \subseteq \mathbb{R}^{d}$ is compact and convex with a positive Lebesgue measure.
\item $X_{l,N_{l}}$ is a uniform tensor grid in $\mathbb{D}$ for $l=1, \ldots, L$.     \footnote{As long as $X_{l,N_{l}}$ is quasi-uniform (see Definition 4.6 and Proposition 14.1 in \cite{wendland2004scattered}), this condition is satisfied. } 
\item Condition \ref{cond_scattered} in \citep{wang2020prediction} is satisfied for Mat\'ern kernels $K_{l}$ and $K_{l,0}$ in (\ref{mat}) for $l=1, \ldots, L$.
\item $\sigma_{l} \lesssim (h_{l}-h_{l-1})^{2\alpha}$ holds for some positive $\alpha$ and $l=1,\ldots,L$ where $h_{0}$ is $0$.
\item $\tilde{c}_{l} h_{\mathrm{X}_{l,N_{l}}}^{\nu_{l}}$ in Lemma \ref{wu} takes value in $(0,\sqrt{\mathrm{e}}]$ and $h_{\mathrm{X}_{l,N_{l}}}^{\nu_{l}}$ is smaller than $\frac{1}{e}$ for $l=1,\ldots, L$, where $h_{\mathrm{X}_{l,N_{l}}}$ is the fill distance of $X_{l,N_{l}}$.
\item $N_{l} \geq 2$ for $l=1, \ldots, L$.
\end{enumerate}
Then, the following inequality holds for probability at least $p$ where $p$ is an arbitrary number in $[0,1)$.
\begin{eqnarray}
\label{error_overall}
\sum_{l=1}^{L} |\eta_{l}(\boldsymbol{x}) - m_{l,N_{l}}^{\eta_{l}}\left(\boldsymbol{x}\right) |  
&\lesssim& \sum_{l=1}^{L} \Bigl( (h_{l}-h_{l-1})^{2\alpha} N_{l}^{-{\nu_{l} \over d}} \log ^{1 / 2} N_{l} \Bigr)
\end{eqnarray}
\end{thm}

Since the evaluations of the each increment $\delta_{l}(x)$ are assumed to be the realizations of the mutually independent $\eta_{l}(x) \sim \text{GP}(0,K_{l}(x,x^{\prime}))$, we can apply the Theorem \ref{scattered_conv}, which deals with a GP sample path itself, to our case.
Before showing the detail of proof, the auxiliary lemma is presented so that the Power function is converted into the function of the fill distance, which is a quantity depending only on the design $X_{l,N_{l}}$. Given the input domain $\mathbb{D}$, the fill distance of a design $\mathbf{X}$ is defined as $
h_{\mathrm{X}} \coloneqq \sup _{x \in \mathbb{D}} \min _{x_{j} \in \mathrm{X}}\left\|x-x_{j}\right\|
.$ Clearly, the fill distance quantifies the space-filling property \citep{santner2003design} of a design. A design having the minimum fill distance among all possible designs with the same number of points is known as a minimax distance design \citep{johnson1990minimax}. Thanks to the following Lemma, which is introduced in \citep{wang2020prediction}, the unknown hyperparameter $\lambda_{l}$ in the kernel is incorporated into the constant $\tilde{c}_{l}, \tilde{h}_{l,0}$.

\begin{lem}[{\cite{wu1993local}}, Theorem 5.14]
\label{wu}
 Let $\mathbb{D}$ be compact and convex with a positive Lebesgue measure; $K_{l}(x,x^{\prime})$ be a Mat\'ern kernel given by (\ref{mat}) with the smoothness parameter $\nu_{l}$. Then there exist constants $\tilde{c}_{l}, \tilde{h}_{l,0}$ depending only on $\mathbb{D}, \nu_{l}$ and the scale parameter $\lambda_{l}$ in (\ref{mat}), such that $P_{K_{l}, X_{l,N_{l}}} \leq \tilde{c}_{l} h_{\mathrm{X}_{l,N_{l}}}^{\nu_{l}}$ provided that $h_{\mathrm{X}_{l,N_{l}}} \leq \tilde{h}_{l,0}$
\end{lem}
\begin{proof}[Proof of Theorem \ref{mlmc_bound}]
By applying Theorem \ref{scattered_conv} and Lemma \ref{wu} to $| \eta_{l}(x)-m_{l,N_{l}}^{\eta_{l}}\left(x\right)|$, we obtain the following inequality with probability at least $p$ :
\begin{eqnarray}
\label{conv_wang_again_now}
\sup_{x \in \mathbb{D}}\left|\eta_{l}(x)-m_{l,N_{l}}^{\eta_{l}}\left(x\right) \right| &\leq& \mathcal{K}_{l} \sigma_{l} P_{K_{l}, X_{l,N_{l}}} \log ^{1 / 2}\left(\mathrm{e} / P_{K_{l}, X_{l,N_{l}}}\right)+u \nonumber \\ 
&\leq & \mathcal{K}_{l} \sigma_{l} \tilde{c}_{l} h_{\mathrm{X}_{l,N_{l}}}^{\nu_{l}} \log ^{1 / 2}\left(\mathrm{e} / \tilde{c}_{l} h_{\mathrm{X}_{l,N_{l}}}^{\nu_{l}}\right)+u \nonumber \\
& \leq&  \mathcal{K}_{l} \sigma_{l} \tilde{c}_{l} h_{\mathrm{X}_{l,N_{l}}}^{\nu_{l}} \log ^{1 / 2}\left(1 /  h_{\mathrm{X}_{l,N_{l}}}^{\nu_{l}}\right)^{\kappa_{l}} +u
\end{eqnarray}
where $u= \sigma_{l}P_{K_{l}, X_{l,N_{l}}}A_{1,l} \left(2\log{2 \over 1-p} \right)^{1/2}$ so that the inequality holds for probability at least $p$. 
In the second inequality, we used the fact the function $x\log^{1/2}(\mathrm{e}/x) = x(1-\log x)^{1/2}$ is monotone increasing in $(0,\sqrt{\mathrm{e}} ]$ and assumption that $\tilde{c}_{l} h_{\mathrm{X}_{l,N_{l}}}^{\nu_{l}}$ falls in this interval. For the third inequality, $(\mathrm{e} / \tilde{c}_{l} h_{\mathrm{X}_{l,N_{l}}}^{\nu_{l}}) \leq (1 /  h_{\mathrm{X}_{l,N_{l}}}^{\nu_{l}})^{\kappa_{l}}$ holds for $h_{\mathrm{X}_{l,N_{l}}}^{\nu_{l}} < \frac{1}{e}$ and some positive $\kappa_{l}$. 

For uniform tensor grids $X_{l,N_{l}}$, the fill distance $h_{\mathrm{X}_{l,N_{l}}}$ is of the order $N_{l}^{-1/d}$  (Proposition 14.1 in \citep{wendland2004scattered}). Then, the upper bound in (\ref{conv_wang_again_now}) is written in the following:

\begin{eqnarray*}
\label{conv_tensor}
&\hspace{1pt}& \mathcal{K}_{l} \sigma_{l} \tilde{c}_{l} h_{\mathrm{X}_{l,N_{l}}}^{\nu_{l}} \log ^{1 / 2}\left(1 /  h_{\mathrm{X}_{l,N_{l}}}^{\nu_{l}}\right)^{\kappa_{l}} +u  \nonumber \\
&=& \mathcal{K}_{l} \sigma_{l} \tilde{c}_{l} h_{\mathrm{X}_{l,N_{l}}}^{\nu_{l}} \log ^{1 / 2}\left(1 /  h_{\mathrm{X}_{l,N_{l}}}^{\nu_{l}}\right)^{\kappa_{l}} +\sigma_{l}P_{K_{l}, X_{l,N_{l}}}A_{1,l} \left(2\log{2 \over 1-p} \right)^{1/2} \nonumber \nonumber \\
&\leq& \mathcal{K}_{l} \sigma_{l} \tilde{c}_{l} h_{\mathrm{X}_{l,N_{l}}}^{\nu_{l}} \log ^{1 / 2}\left(1 /  h_{\mathrm{X}_{l,N_{l}}}^{\nu_{l}}\right)^{\kappa_{l}} +\sigma_{l} \tilde{c}_{l} h_{\mathrm{X}_{l,N_{l}}}^{\nu_{l}} A_{1,l} \left(2\log{2 \over 1-p} \right)^{1/2} \nonumber \\
&\leq& \mathcal{K}_{l} \sigma_{l} \tilde{c}_{l} \kappa_{l}^{1/2} h_{\mathrm{X}_{l,N_{l}}}^{\nu_{l}} \log ^{1 / 2}\left(1 /  h_{\mathrm{X}_{l,N_{l}}}^{\nu_{l}}\right) +\sigma_{l} A_{1,l}  \left(2\log{2 \over 1-p} \right)^{1/2} \tilde{c}_{l}  h_{\mathrm{X}_{l,N_{l}}}^{\nu_{l}} \log ^{1 / 2}\left(1 /  h_{\mathrm{X}_{l,N_{l}}}^{\nu_{l}}\right) \nonumber \\
&\simeq& \sigma_{l} h_{\mathrm{X}_{l,N_{l}}}^{\nu_{l}} \log ^{1 / 2}\left(1 /  h_{\mathrm{X}_{l,N_{l}}}\right) \nonumber \\
&\simeq& \sigma_{l} N_{l}^{-{\nu_{l} \over d}} \log ^{1 / 2} N_{l} \nonumber \\
&\lesssim&  (h_{l}-h_{l-1})^{2\alpha} N_{l}^{-{\nu_{l} \over d}} \log ^{1 / 2} N_{l} \nonumber \\
\end{eqnarray*}

In the fourth line, we used the fact that $x \leq x \log^{1/2}\frac{1}{x}$ for $x < \frac{1}{e}$. In the sixth line, the detail of the proof of $ h_{\mathrm{X}_{l,N_{l}}}^{\nu_{l}} \log ^{1 / 2}\left(1 /  h_{\mathrm{X}_{l,N_{l}}}\right) \simeq  N_{l}^{-{\nu_{l} \over d}} \log ^{1 / 2} N_{l}$ for $N_{l} \geq 2$ is in \ref{appendix_inequality}. The same result of one dimensional case (without proof) is in Section 3.3 in \citep{wang2020prediction}. We finally arrive at a conclusion by summing up the above upper bound:
\begin{eqnarray*}
\label{error_overall_2}
\sum_{l=1}^{L} | \eta_{l}(x) - m_{l,N_{l}}^{\eta_{l}}\left(x\right)| 
&\lesssim & \sum_{l=1}^{L} \Bigl( (h_{l}-h_{l-1})^{2\alpha} N_{l}^{-{\nu_{l} \over d}} \log ^{1 / 2} N_{l} \Bigr).
\end{eqnarray*}
\end{proof}

The comments on the assumptions of $p$ and $\sigma_{l}$ are made here. The probability $p$ is no longer present in the upper bound since the constants are independent of the order of convergence but $p$ does affect the actual convergence speed. Besides, we cannot know the exact or estimated value of $\sigma_{l}$ before running computer simulations though $\sigma_{l}$ represents the magnitude (hence importance) of the increment $\delta_{l}$. Therefore, we assume that $\sigma_{l}$, which regulates the magnitude of the increment $\delta_{l}$, is proportional to the degree of accuracy $h_{l}$ and $h_{l-1}$ corresponding to the increment $\delta_{l}$. 

Our final goal is to minimize the (equivalent) upper bound of (\ref{error_overall}) under the restriction of computational budget $\sum_{l=1}^{L}N_{l} t_{l}= T$ and $N_{l} \geq 2$ for every $l$. We are now ready to present the following proposition.

\begin{prop}
\label{mlmcresult}
Under the same conditions in Theorem \ref{mlmc_bound}, the optimal numbers of runs $N_{l}$ are given by the solution of the following optimization problem:

\begin{eqnarray*}
\label{MLMC_like_opt}
&\argmin_{N_{1},\cdots,N_{L}}& \ \sum_{l=1}^{L} \Bigl( (h_{l}-h_{l-1})^{2\alpha} N_{l}^{-{\nu_{l} \over d}} \log ^{1 / 2} N_{l} \Bigr)  \\
&s.t.&  \sum_{l=1}^{L} N_{l}t_{l} =T \\
&\hspace{10pt} & N_{l} \geq 2 \hspace{10pt}  l=1,\ldots, L
\end{eqnarray*}

where the optimal numbers of runs mean that the equivalent quantity of the upper bound $\sum_{l=1}^{L} |m_{l,N_{l}}^{\eta_{l}}\left(x\right) - \eta_{l}(x)|$ is minimized with respect to $N_{l}$.
\end{prop}

The contribution from $\delta_{l}$ for large $l$ is weighted with a relatively small value $\sigma_{l}^{2}$, which is proportional to the difference of the accuracy level $(h_{l}-h_{l-1})^{2\alpha}$, which would be usually expected to decline as the level $l$ increases. Hence we do not have to run expensive computer simulations many times. The optimization problem would not be solved in a closed form in general but we can apply an ordinary method of numerical optimization such as Trust-Region Constrained Algorithm. The choice of $\alpha$ might be a challenge while $h_{l}$ can be determined by the user, typically taken to be proportional to discretization level in a numerical solver. In the similar framework of MLMC \citep{giles2015multilevel}, deciding unknown parameters is also a difficult task. A feasible strategy for choosing $\alpha$ in our analysis is to select a relatively small value. For $\alpha < \alpha^{\prime}$, if $\sigma_{l} \lesssim (h_{l}-h_{l-1})^{2\alpha^{\prime}}$ holds then $\sigma_{l} \lesssim (h_{l}-h_{l-1})^{2\alpha}$ is satisfied since $h_{i} - h_{i-1}$ is smaller than 1 as defined before. Figure \ref{explicit} shows how the numbers of data points at different levels change as $\alpha$ takes different values. More emphasis is placed on the lower levels as $\alpha$ takes larger values, which is consistent with our expectation. We assume that the number of levels $L$ is $5$, $d=1$, $\nu_{l}=2.5$ for every $l$ and total budget $T$ is $800$. The degree of accuracy $h_{l}$ and computational costs $t_{l}$ are in Table \ref{explicit_info}. The rate of increase in the cost is exponential and this assumption is similar to that of the theorem of MLMC (see Theorem 2.1 in \cite{giles2015multilevel}).  

\begin{table}[htb]
 \caption{Computational costs of computer simulations}
 \begin{center}
  \begin{tabular}{|c|c|c|c|c|c|} \hline 
 & $l=1$ & $l=2$ & $l=3$  & $l=4$ & $l=5$  \\ \hline 
  $h_{l}$ & $1$ & $1/2$ & $1/4$ & $1/8$ & $1/16$  \\ \hline 
  $t_{l}$ & $0.5$ & $2$ & $8$ & $32$ & $128$  \\ \hline 
  \end{tabular}
  \label{explicit_info}
 \end{center}
\end{table}

A Larger value of $\alpha$ means more importance is placed on $\delta_{l}$ of lower accuracy hence the smaller number of samples is picked from expensive simulations. Unless the simulations of low fidelity are known to be trustworthy in advance, the small $\alpha$ is recommended.

\begin{remark}
\label{closed_form}
If we assume $\nu_{l} = \nu^{*}$ for every $l$ and $ \nu^{*} > d /2e $, a closed form of the numbers of sample points can be derived. By taking note that $\log x \leq x^{a}$ for $x \geq 1$ and $a \geq 1 / e$, the upper bound in (\ref{error_overall}) is replaced with the more loose bound $\sum_{l=1}^{L} \Bigl( (h_{l}-h_{l-1})^{2\alpha} N_{l}^{-{\nu^{*} \over d} + {1 \over 2e}}\Bigr)$. Therefore, the optimization problem in Proposition \ref{mlmcresult} can be replaced with $\argmin_{N_{1},\cdots,N_{L}} \sum_{l=1}^{L} \Bigl( (h_{l}-h_{l-1})^{2\alpha} N_{l}^{-{\nu^{*} \over d} + {1 \over 2e}} \Bigr) $ under the restriction $ \sum_{l=1}^{L} N_{l}t_{l} =T$ for positive $N_{l}$. By applying Lagrange multipliers, the optimal numbers of points (these should be no less than 2 when used) are given by 

\[N_{l}^{*} =
 \Bigl(\frac{t_{l}}{-r (h_{l}-h_{l-1})^{2 \alpha}} \Bigr)^{\frac{1}{r-1}} \frac{T}{ \sum_{j=1}^{L}t_{j}^{\frac{r}{r-1}} (-r (h_{l}-h_{l-1})^{2\alpha})^{\frac{1}{1-r}}} \] where $r= -\nu^{*} / d + 1 /2e $  $(<0)$. 
\end{remark}

\begin{figure}[h]
\begin{center}
\begin{tabular}{cc}
\subfloat[The number of runs depending on different budgets. $\alpha=1$.]{
\includegraphics[width=0.45\linewidth]{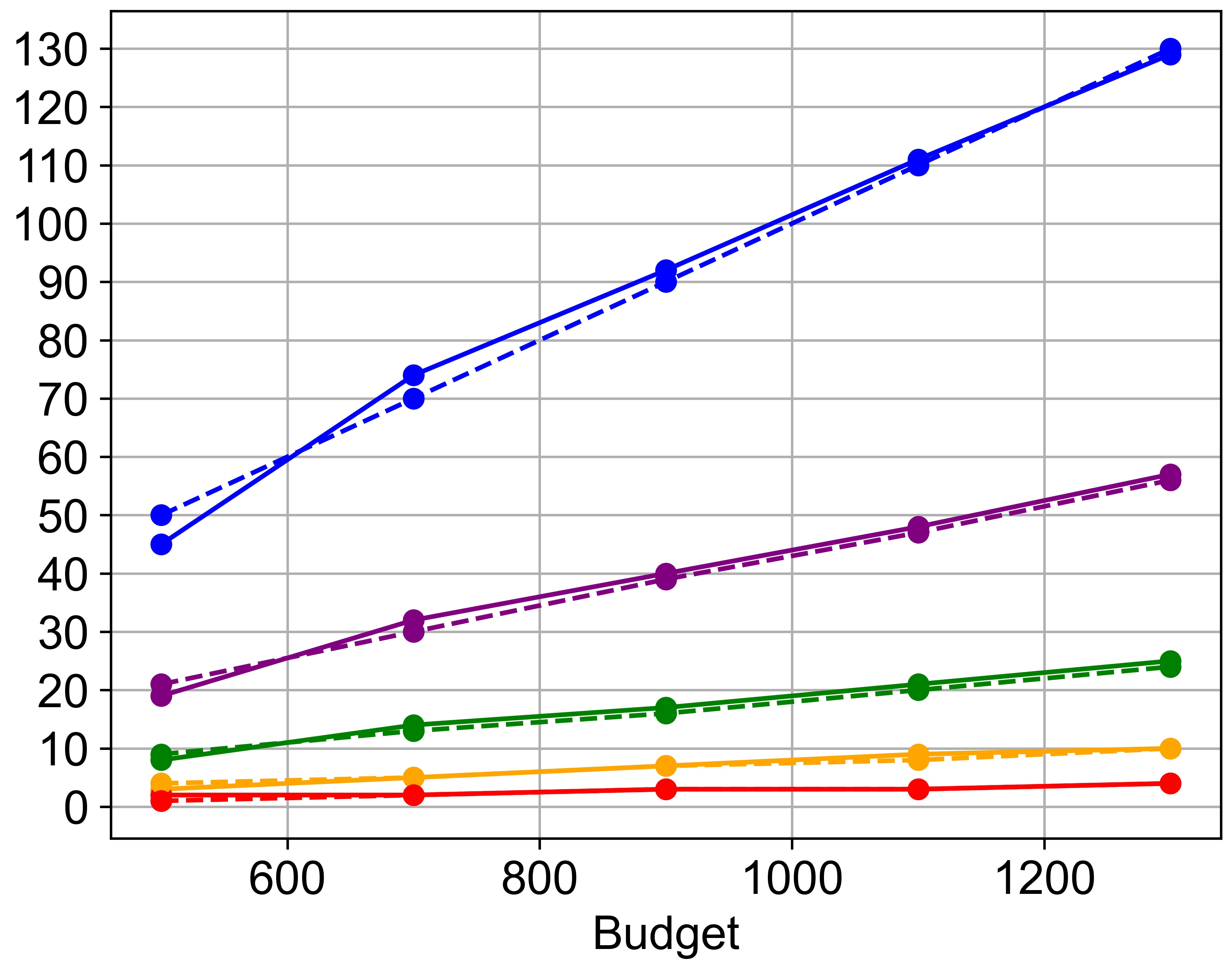}
\label{explicit_BUDGET}
} &
\subfloat[The number of runs depending on different $\alpha$. Budget is $800$.]{
\includegraphics[width=0.45\linewidth]{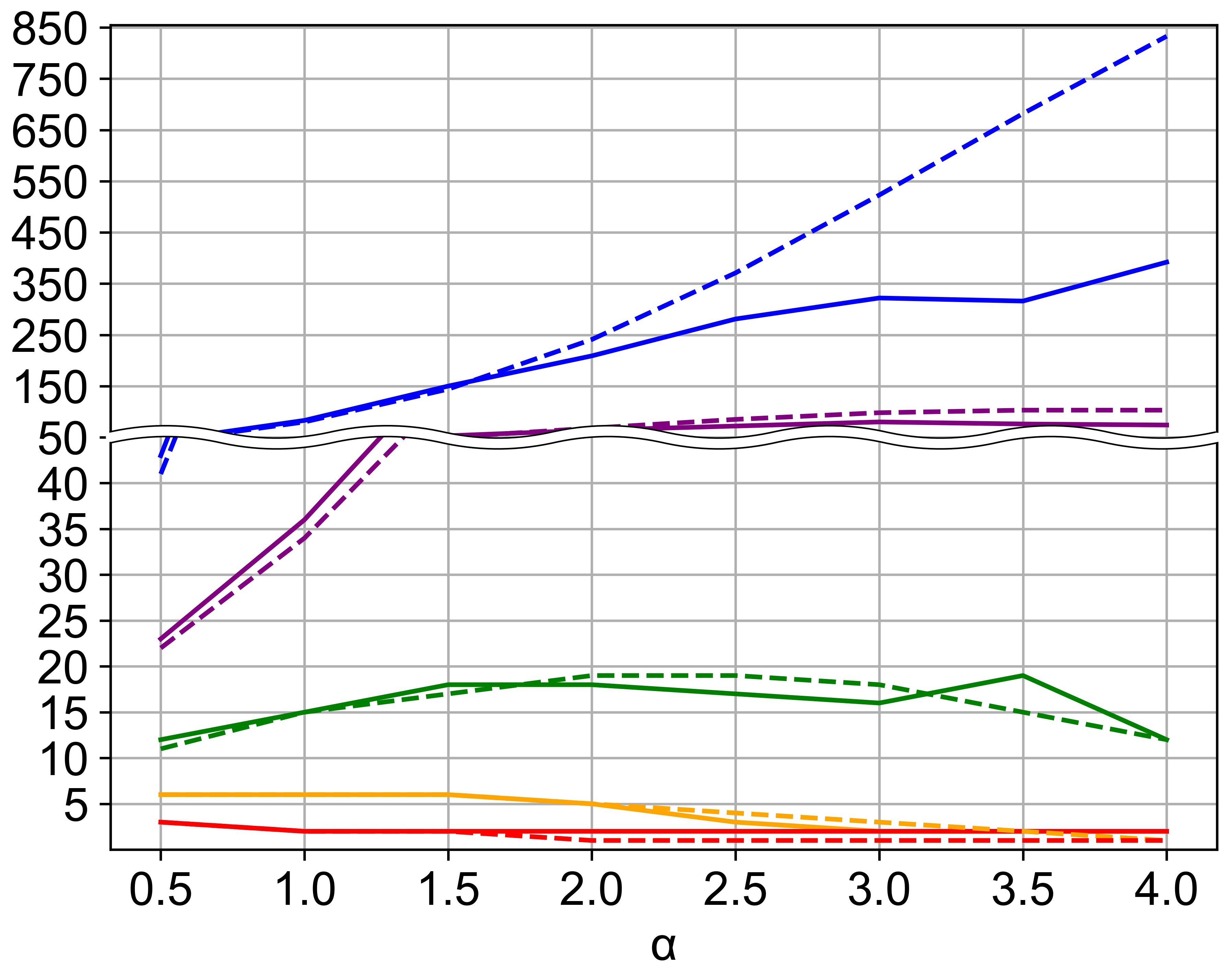}
\label{explicit}
}  \\
\end{tabular}
\caption{The number of runs of at each level ( $L = 5$). The blue, purple, green, orange, red lines denote level 1, 2, 3, 4, and 5 respectively. $t_{l}$ and $h_{l}$ are in Table \ref{explicit_info}. $\nu_{l}=2.5$ for every $l$, $d=1$. The solid lines denote the numerical solutions of optimization problem in Proposition \ref{mlmcresult} and the dashed line denote the closed form in Remark \ref{closed_form}.}
\label{fig}
\end{center}
\end{figure}

\section{Numerical examples}
\label{nume_example_1}
In this section, we present the results of our implementation of MLASCE on several numerical examples and compare with existing multi-fidelity kriging models. Emulators created by MLASCE yield more accurate results than the existing methods. In our algorithm, $a_{l}$ in (\ref{score}) is set to $t_{l}$ since the difference between successive computer simulations gets smaller and smaller. Section \ref{numerical_example_1} highlights the difference in the performance of MLASCE and other methodology in a simple example. Section \ref{numerical_example_2} shows that MLASCE is more robust in the case where different versions of a computer simulation have different degrees of smoothness, which can occur in practice as resolutions increase for instance. Section \ref{numerical_exmaple_tdac} reveals a more realistic application of MLASCE in a long wave simulation with multidimensional inputs. Compared to a non-multilevel methodology, the predictions by MLASCE are still valid in this example.

\subsection{Example 1: a toy model with fixed smoothness}
\label{numerical_example_1}
A simple toy model is considered to illustrate the performance of MLASCE and contrast it with the models in \cite{le2014recursive,le2015cokriging} as the most sophisticated formulation of multi-fidelity GPR and strategy of sequential design. Three levels of computer simulations defined in $[0,\pi]$ are used (Figure \ref{bump}). This type of nested functions is used as the numerical examples in \cite{le2013bayesian}. We define the simulations $f_{1}(x) = \mathrm{sin}x$, $f_{2}(x) = f_{1}(x) + \xi \bigl(x,\frac{\pi}{3},0.4 \bigr)$ and $f_{3}(x) = f_{2}(x) - \frac{1}{2}\xi \bigl(x,\frac{\pi}{4},0.2 \bigr) +\frac{1}{2}\xi \bigl(x,\frac{3\pi}{4},0.2\bigr)$ where $\xi(x,a,\lambda) = \bigl( 1+\frac{\sqrt{5}}{\lambda}|x-a|  +\frac{5|x-a|^{2}}{3\lambda^{2}} \bigr) \mathrm{exp}\bigl(-\frac{\sqrt{5}}{\lambda}|x-a| \bigr)$.

\begin{figure}[h]
\begin{center}
\begin{tabular}{cc}
\subfloat[The black line is $f_{1}(x)$,  green is $f_{2}(x)$ and blue is $f_{3}(x)$.]{
\includegraphics[width=0.45\linewidth]{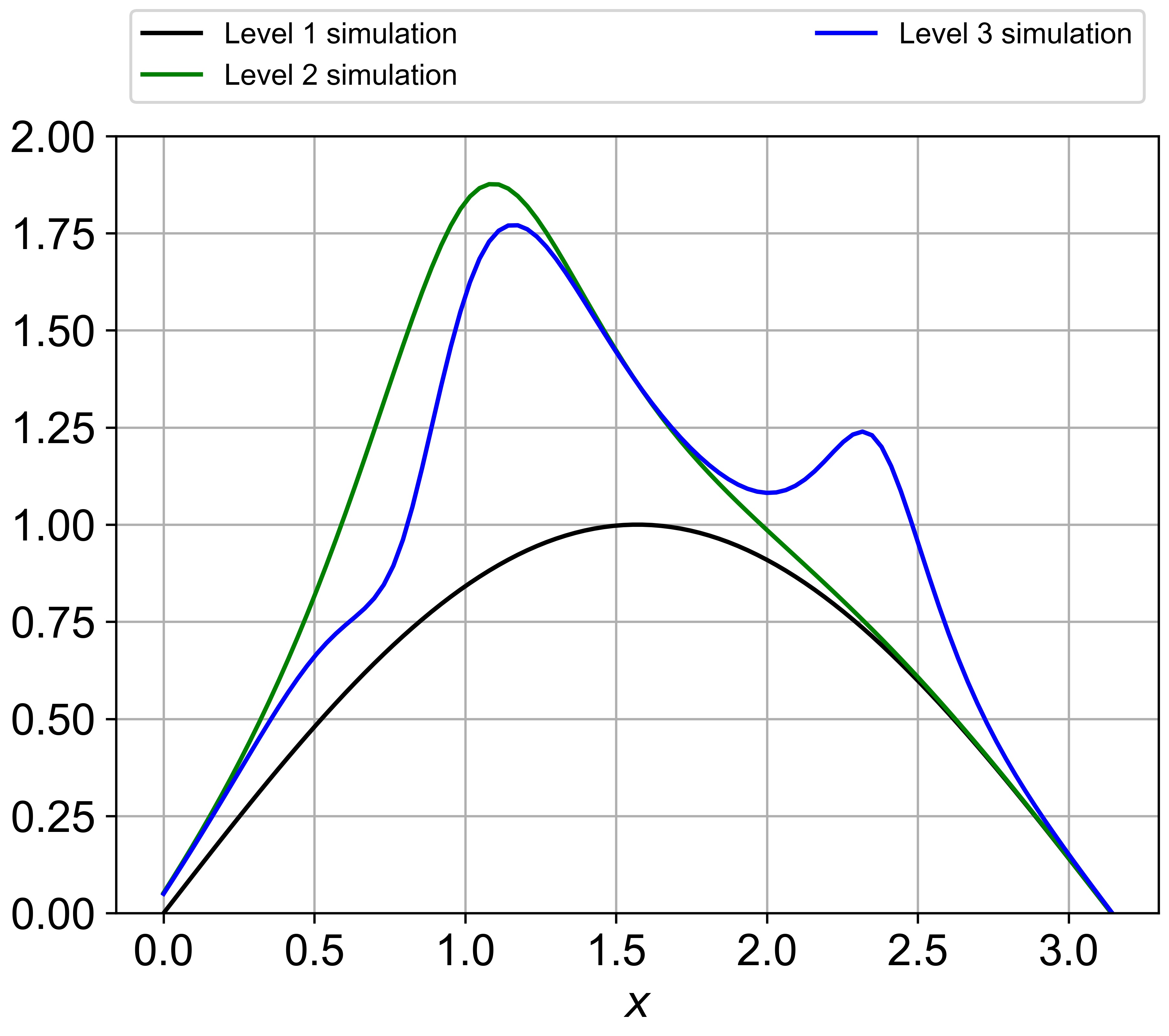}
\label{bump}
} &
\subfloat[]{
\includegraphics[width=0.45\linewidth]{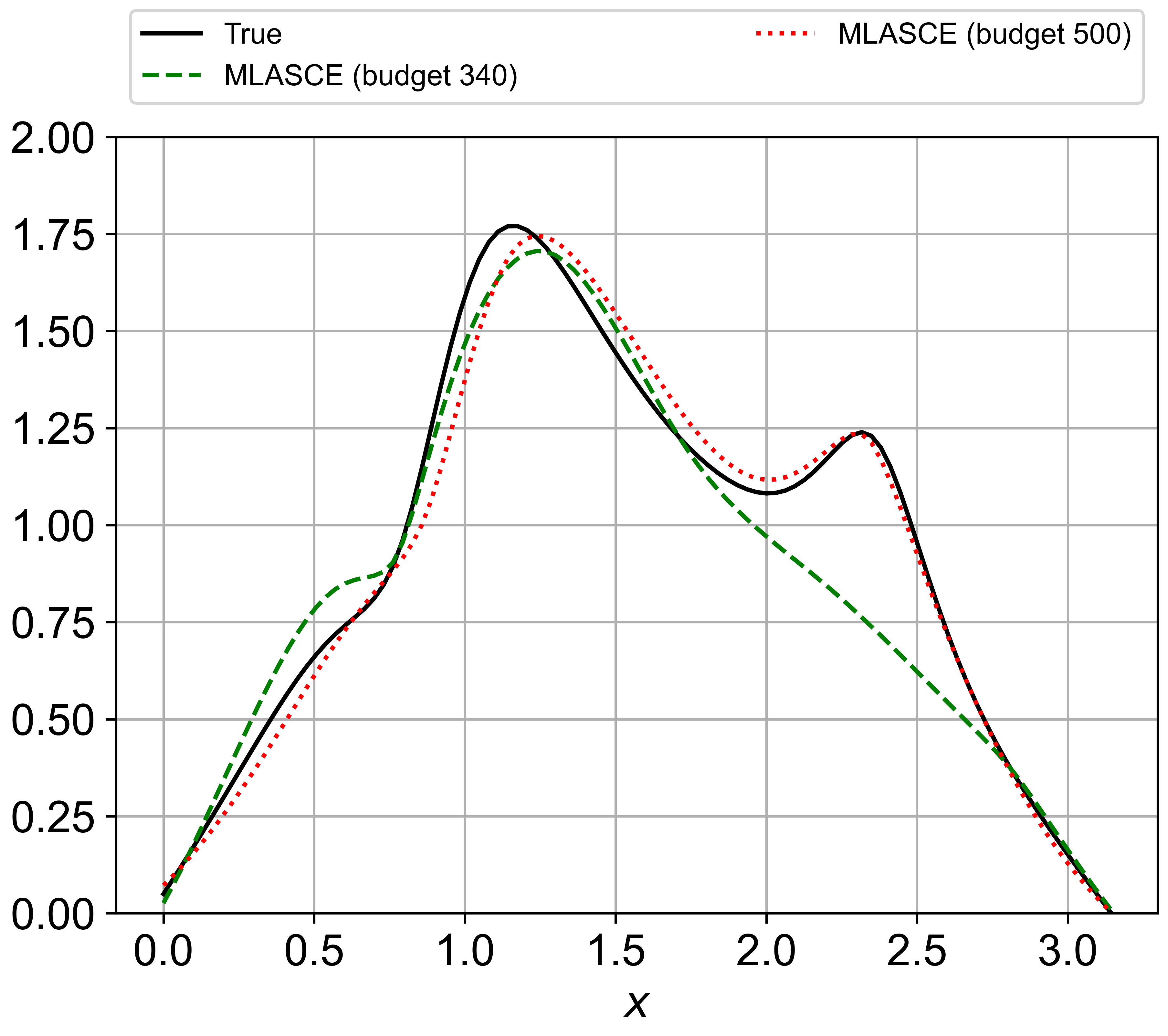}
\label{budget_500}
}  \\
\subfloat[]{
\includegraphics[width=0.45\linewidth]{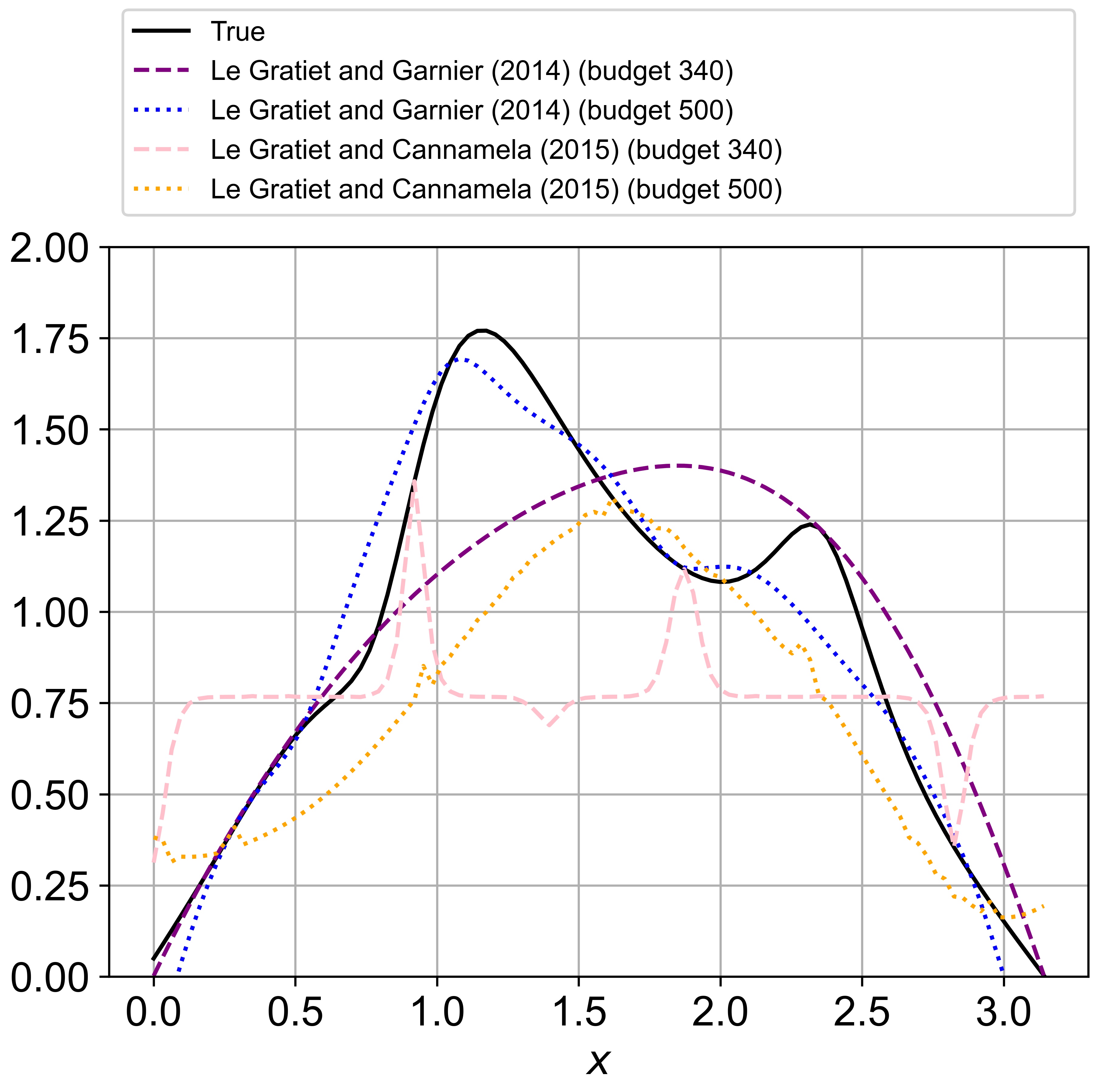}
\label{budget_500_2}
}  &
\subfloat[]{
\includegraphics[width=0.45\linewidth]{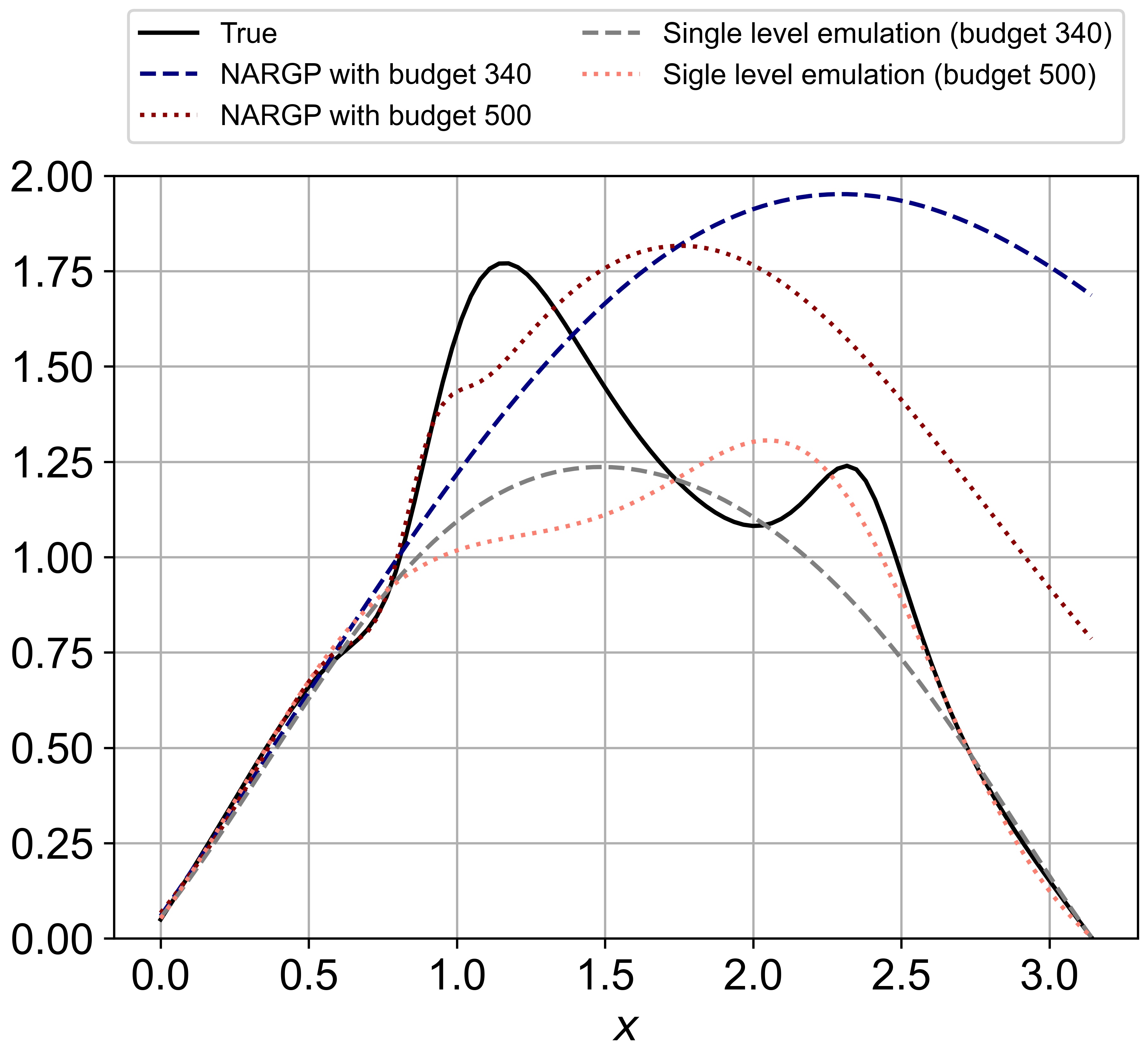}
\label{budget_500_2}
} \\
\end{tabular}
\caption{The graphs of computer simulations and emulators}
\end{center}
\end{figure}

$f_{1}(x)$, $f_{2}(x)$ and $f_{3}(x)$ share the same trend $\mathrm{sin}x$ while $ f_{2}(x)$ and $f_{3}$ have the idiosyncratic movements of $\xi(x,a,\lambda)$ at $x=\pi/3, \pi/4, 3\pi/4$ respectively. A similar phenomenon can be found in a simulation of complex tsunami dynamics as the degree of accuracy of the numerical simulator improves, assuming the conditions of numerical stability, such as CFL conditions, are guaranteed. The computational costs for running these simulations are illustrated in Table \ref{info}. These costs are artificially assumed and increase at an exponential rate of levels. We consider five cases of different total budgets $340$, $380$, $420$, $460$ and $500$. As discussed in the next part where the methods of \cite{le2014recursive,le2015cokriging} are compared with ours, their formulations require a sufficiently large amount of training data for every level hence the generous budget while MLASCE is robust in small size of training data thanks to relatively few numbers of hyperparameters. 

\begin{table}[htb]
 \caption{Computational costs of computer simulations}
 \begin{center}
  \begin{tabular}{|c|c|c|c|c|c|} \hline 
 $f_{1}(x)$ & $f_{2}(x)$ & $f_{3}(x)$ & $h_{1}(x)$ & $h_{2}(x)$ & $h_{3}(x)$  \\ \hline 
  4 & 16 & 64 & 4 & 20 & 80  \\ \hline 
  \end{tabular}
  \label{info}
 \end{center}
\end{table}
We use the incremental functions $h_{l}, l=1,2,3$ defined as $h_{1}(x)=f_{1}(x)$ and $h_{l}(x) = f_{l}(x) - f_{l-1}(x) \hspace{4pt} (l=2, 3)$. 
We now compare with other existing methods.
We prepare three Gaussian processes $GP(0,K_{5/2, \lambda _{l},\sigma^{2}_{l}}(x, x^{\prime}) )$ with a Mat\'ern kernel for $l=1, 2, 3$. MICE and fitting of GPs are implemented by python modules ``Multi-Output Gaussian Process Emulator'' and  ``scikit-learn'' and a nugget parameter is chosen to be $10^{-8}$ for every level.
The mean predictions $m^{f_{3}}_{3}\left(x\right) = \sum_{l=1}^{3} m_{l,N_{l}}^{h_{l}}\left(x\right)$ of MLASCE method with different budgets are illustrated in Figure \ref{budget_500} in comparison with \cite{le2014recursive,le2015cokriging} and the numbers of runs are also presented in Figure \ref{data_imp}. As expected, the performance of the emulators improves as the budget increases. Since MLASCE tends to run the function which yields a larger incremental value $\delta_{l}(x)$ more often than the other, the number of training data points at level $3$ gets larger if $\delta_{3}$ is affordable. Moreover, the numbers of training data points at level $1$ and $2$ stabilize after sufficient numbers of data points are obtained while MLASCE focus on low-fidelity simulator if the budget is small. The $L^{2}$ error is also presented in Figure \ref{L2} and $L^{2}$ error is defined as the integral of the squared difference between the mean prediction and $f_{3}(x)$ in $[0, \pi]$. 


\begin{figure}[h]
\begin{center}
\begin{tabular}{cc}
\subfloat[$L^{2}$ error of the different emulators (MLASCE, \cite{le2014recursive}, \cite{le2015cokriging} and \cite{perdikaris2017nonlinear}) with various budgets from 340 to 500.]{
\includegraphics[width=0.45\linewidth]{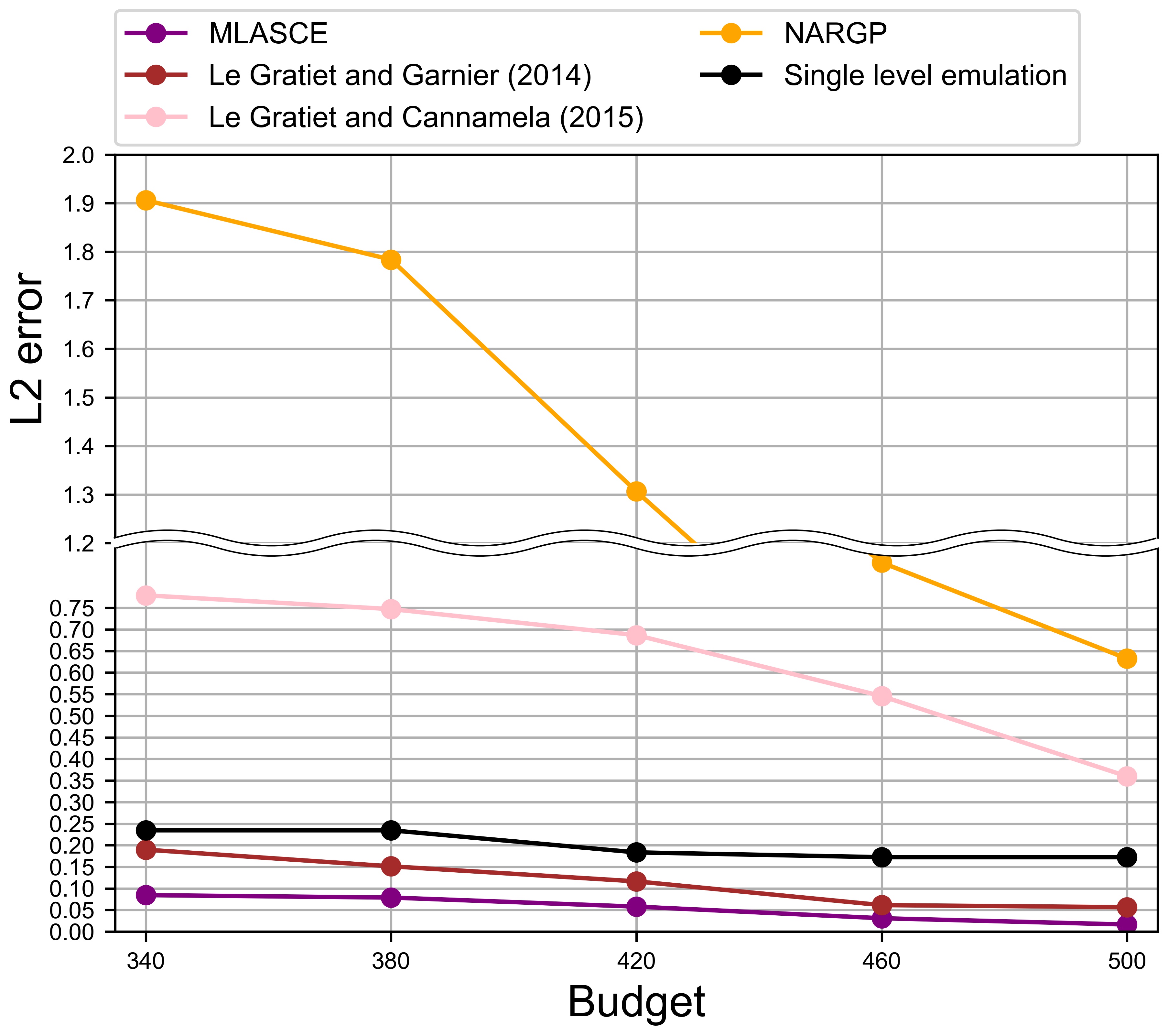}
\label{L2}
} &
\subfloat[Number of data points per level]{
\includegraphics[width=0.45\linewidth]{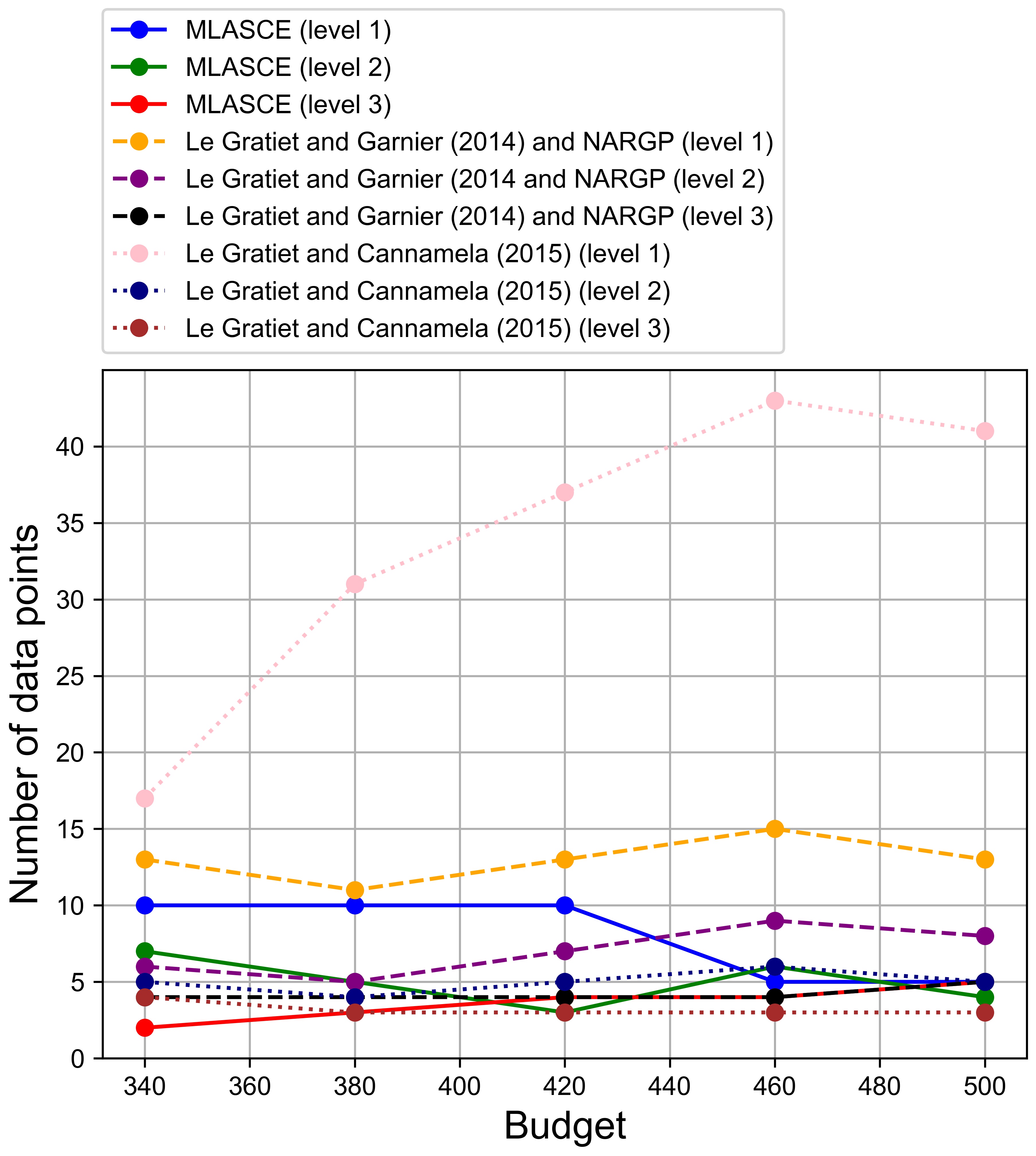}
\label{data_imp}
} \\
\end{tabular}
\caption{The $L^{2}$ error (left) and numbers of runs of different emulators (right).}
\end{center}
\end{figure}

   We compare MLASCE with the model proposed by \cite{le2014recursive,le2015cokriging} by looking at the performances of these GP emulators under the same amount of computational budget.  Following the formulation in (\ref{leg}), we prepare three Gaussian processes $GP(\beta_{0,l},K_{5/2, \lambda _{l},\sigma^{2}_{l}}(x, x^{\prime}) )$ with a Mat\'ern kernel for $\delta_{1}(x), \delta_{2}(x)$ and $\delta_{3}(x)$ in (\ref{leg}) where $\beta_{0,l}$ denotes a constant mean. We specify the adjustment effect $\rho_{l}(x)$ as $\rho_{l}(x) = \beta_{\rho_{l},0} + \beta_{\rho_{l},1}x +  \beta_{\rho_{l},2}x^{2}$ for the method by \cite{le2014recursive} where $\beta_{\rho_{l},0}$, $\beta_{\rho_{l},1}$ and $\beta_{\rho_{l},2}$ are constant coefficients. This specification is to include a nonlinear adjustment effect since the relationships among different $f_{i}(x)$ seem nonlinear. On the other hand, the method of \cite{le2015cokriging} assumes constant for the adjustment effect $\rho_{l}(x)=\rho_{l}$. $f_{l}(x)$ is used as a simulation instead of an increment $h_{l}(x)$. We assume noninformative priors for all of the hyperparameters and maximum likelihood estimation is employed. The single level emulation which only uses the outputs from $f_{3}(x)$ and relies on MICE in terms of its design strategy is also compared. The numbers of its design points are 5 (budget 340, 380), 6 (budget 420) and 7 (budget 460, 500) and the Mat\'ern 2.5 kernel is used. NARGP \cite{perdikaris2017nonlinear} is also in comparison although NARGP does not consider an efficient  strategy for the experimental design in their framework. The same numbers of training points as in \cite{le2014recursive} are prepared. 

As proposed in \cite{le2014recursive}, a strategy of nested space filling design is presented in which the designs of input at higher levels are incorporated in the ones at lower levels while the method of \cite{le2015cokriging} employs their own strategy of sequential design. Their strategy of design and construction of the emulators are implemented by R package ''MuFiCokriging'' (\url{https://CRAN.R-project.org/package=MuFiCokriging}). Although a rule for deciding the number of design points is not presented in \cite{le2014recursive}, the number of data points at each level is determined to keep a good proportion of the data size across every level (Figure \ref{data_imp}). More crucially, due to a large amount of hyperparameters in their formulation, sufficient numbers of training data points should be provided even for the most expensive simulators. Otherwise, determining the hyperparameters fails and the emulator cannot be constructed. This drawback also holds for \cite{le2015cokriging} and it requires decent numbers of initial design points at every level before implementing their strategy of sequential design even though it originally intends to allocate limited computational resources among cheap and expensive simulators. On the other hand, in MLASCE, only the two hyperparameters $\lambda_{l}, \sigma_{l}^{2}$ are estimated at each level and the incremental function $\delta_{l}$ tends to take smaller values hence fitted hyperparameters are relatively more stable. Although the method of \cite{le2014recursive} tries to overcome the nonlinear correlations among simulators of different fidelity by designing $\rho_{l}(x)$, this attempt requires a large amount of training data, which is practically impossible if the simulator is expensive. On the other hand, constant $\rho_{l}(x) = \rho_{l}$ by \cite{le2015cokriging} ends up yielding much worse emulators. As discussed previously, the budget is decided so that their method can construct an emulator successfully and more data is available for them since the cheaper simulator $f_{l}(x)$ is used instead of $h_{l}(x)$.
The results are shown in Figure \ref{budget_500}. Although the mean predictions improve as the budget increases, MLASCE performs better than the method of \cite{le2014recursive,le2015cokriging} especially in the cases of the smaller budgets. This is more evident in terms of $L^{2}$ error (Figure \ref{L2}).
The key difference is whether the error between the true response and emulator is considered or not. As already mentioned, the existing co-kriging methods do not take into account the error of the emulator and only consider the posterior variance as an indicator of accuracy. This leads to relatively poor performance as an emulator even though enough amount of training data is provided. The sequential design strategy by \cite{le2015cokriging} may not be helpful since it places too much emphasis on the cheaper simulations, which is their original intention, hence skews the emulation. This type of problem of the existing co-kriging methods becomes more evident or even worse in the next part, where the computer simulations with more levels of fidelity are dealt with.
The single level emulation cannot capture the complex behavior of the true underlining function, as we expected. NARGP predicts much worse in terms of the $L^{2}$ error and its unstable behavior, since NARGP seems intended to be applied in the case where more abundant amount of training data is available, as in the examples in the original paper \cite{perdikaris2017nonlinear} where at least 20 samples were available at the highest resolution. However, NARGP could be improved using new Deep GP analytical formula from \cite{ming2021linked}.

\subsection{Example 2: five levels and variations of smoothness}
\label{numerical_example_2}
We focus on the functions with different smoothness and compare the performance of emulators with kernels of different smoothness. As discussed, it is unlikely to know in advance the exact degree of smoothness of the simulation in practice and an estimation of the smoothness parameter $\nu$ in a Mat\'ern kernel is not reliable if the simulation is expensive to run. Thus, as mentioned before, a feasible strategy is to choose a relatively small $\nu$ so that the corresponding RKHS can incorporate a wider class of functions. The simulated function can be less smooth as the degree of fidelity in the simulator refines hence flexible choice of $\nu$ may improve the quality of the emulator. In the following examples, we see that the kernels with fixed smoothness across every level lead to a relatively inefficient surrogate. Besides, we also assume more levels of computer simulations so that our method is robust in such a case.

The five levels of functions are defined in $[0, \pi]$ and five different bump effects are incorporated:
\begin{eqnarray*}
f_{1}(x) &=& \mathrm{sin}x \\
f_{2}(x) &=& f_{1}(x) + \xi_{2} \bigl(x,\frac{\pi}{6} \bigr) + \xi_{2} \bigl(x,\frac{5\pi}{6} \bigr) \\
f_{3}(x) &=& f_{2}(x)  -\xi_{3} \bigl(x,\frac{\pi}{4} \bigr) -\xi_{3} \bigl(x,\frac{3\pi}{4} \bigr) \\
f_{4}(x) &=& f_{3}(x)  +\xi_{4} \bigl(x,\frac{\pi}{3} \bigr) +\xi_{4} \bigl(x,\frac{2\pi}{3} \bigr) \\
f_{5}(x) &=& f_{4}(x) + \xi_{5} \bigl(x,\frac{\pi}{8} \bigr)  -\xi_{5} \bigl(x,\frac{4\pi}{8} \bigr) + \xi_{5} \bigl(x,\frac{7\pi}{8} \bigr)
\end{eqnarray*}
where
\begin{eqnarray*}
\xi_{2} \bigl(x, a \bigr) &=& 
 \mathrm{exp}\Bigl(-\frac{1}{(\pi / 8)^{2} - (x-a)^{2} }\Bigr) \times \boldsymbol{1}_{(|x-a| < \frac{\pi}{8})} \\
 \xi_{3} \bigl(x,a \bigr)  &=& 0.3 \hspace{2pt}\mathrm{e}^{-8(x-a)^{2}}\left(-|x-a|^{5}+1 \right) \\
 \xi_{4} \bigl(x,a \bigr) &=& 0.3 \hspace{2pt}\mathrm{e}^{-8(x-a)^{2}}\left(-|x-a|^{3}+1 \right) \\
 \xi_{5} \bigl(x, a \bigr) &=& 0.15 \hspace{2pt} \mathrm{e}^{-12(x-a)^{2}}(x-a+1)^{2} \boldsymbol{1}_{(a-1< x \leq a)}  \\
&+&0.3 \hspace{2pt}\mathrm{e}^{-12(x-a)^{2}} \bigl(1 - {1 \over 2} (x-a-1)^{2} \bigr) \boldsymbol{1}_{(a< x \leq a+1)} \\
&+&0.3 \hspace{2pt}\mathrm{e}^{-12(x-a)^{2}}\boldsymbol{1}_{(a+1 < x )} \\
\end{eqnarray*}
and $\boldsymbol{1}_{(\cdot)}$ denotes an indicator function.

As a function of $x$ and for a fixed $a$, $\xi_{2}$ is a function of class $C^{\infty}$ (but not analytical) and $\xi_{3}$, $\xi_{4}$, $\xi_{5}$ is of class $C^{4}$, $C^{2}$ and $C^{1}$. Thus, every $f_{i}$ has the same degree of smoothness as that of $\xi_{i}$ while $f_{1}$ is analytical. An equivalent Sobolev space $H^{s}(\mathbb{R}^{d})$ is incorporated in $C^{[s- d/2]}(\mathbb{R}^{d})$ for nonnegative and non-integer $s- d/2$ where $[s- n/2]$ denotes the maximum integer which does not exceed $s- n/2$. Then, we can see that the equivalent space $H^{3}(\mathbb{R}) \subset C^{2}(\mathbb{R})$ for $\nu = 5/2$ in the Mat\'ern kernel and $H^{2}(\mathbb{R}) \subset C^{1}(\mathbb{R})$ for $\nu = 3/2$ albeit indirect relationships. These functions are illustrated in Figure \ref{cp_1}. The bump term $\xi_{2}$ is very smooth and has a small effect on the simulation while the other bump terms have more effects (0.3 for $\xi_{3}$ and $\xi_{4}$ and 0.15 for $\xi_{5}$ at location $x=a$) and become less smooth as levels go up. The computational costs are listed in Table \ref{info_bump2} and definition of $h_{l}$ is the same as in the previous part.

\begin{table}[htb]
 \caption{Computational costs of computer simulations. The costs increase at an exponential rate of the levels.}
 \begin{center}
  \begin{tabular}{|c|c|c|c|c|c|} \hline 
  & $l=1$ & $l=2$ & $l=3$& $l=4$ & $l=5$ \\ \hline 
  $f_{l}(x)$ & 0.5 & 2 & 8 & 32 & 128  \\ \hline 
  $h_{l}(x)$ & 0.5 & 2.5 & 10 & 40 & 160  \\ \hline 
  \end{tabular}
  \label{info_bump2}
 \end{center}
\end{table}

\begin{figure}[h]
\begin{center}
\subfloat[Computer simulations of different levels: $f_{1}(x)$ (blue), $f_{2}(x)$ (purple), $ f_{3}(x)$ (green), $f_{4}(x)$ (orange) and $f_{5}(x)$ (red). The objective function for emulation is $f_{5}(x)$.]{
\includegraphics[width=0.49\linewidth]{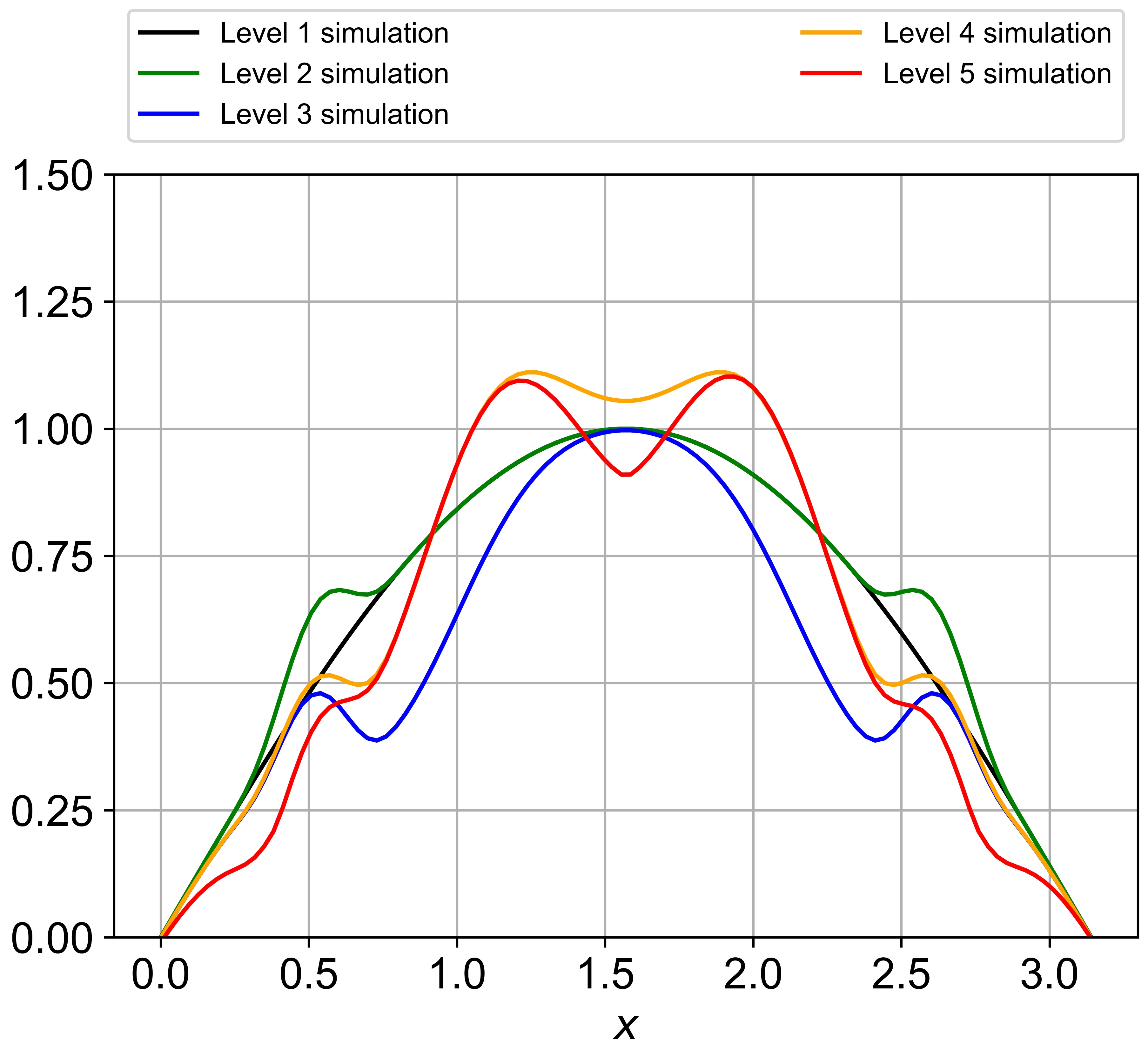}
\label{cp_1}
} 
\subfloat[MLASCE with varied and fixed smoothness.]{
\includegraphics[width=0.49\linewidth]{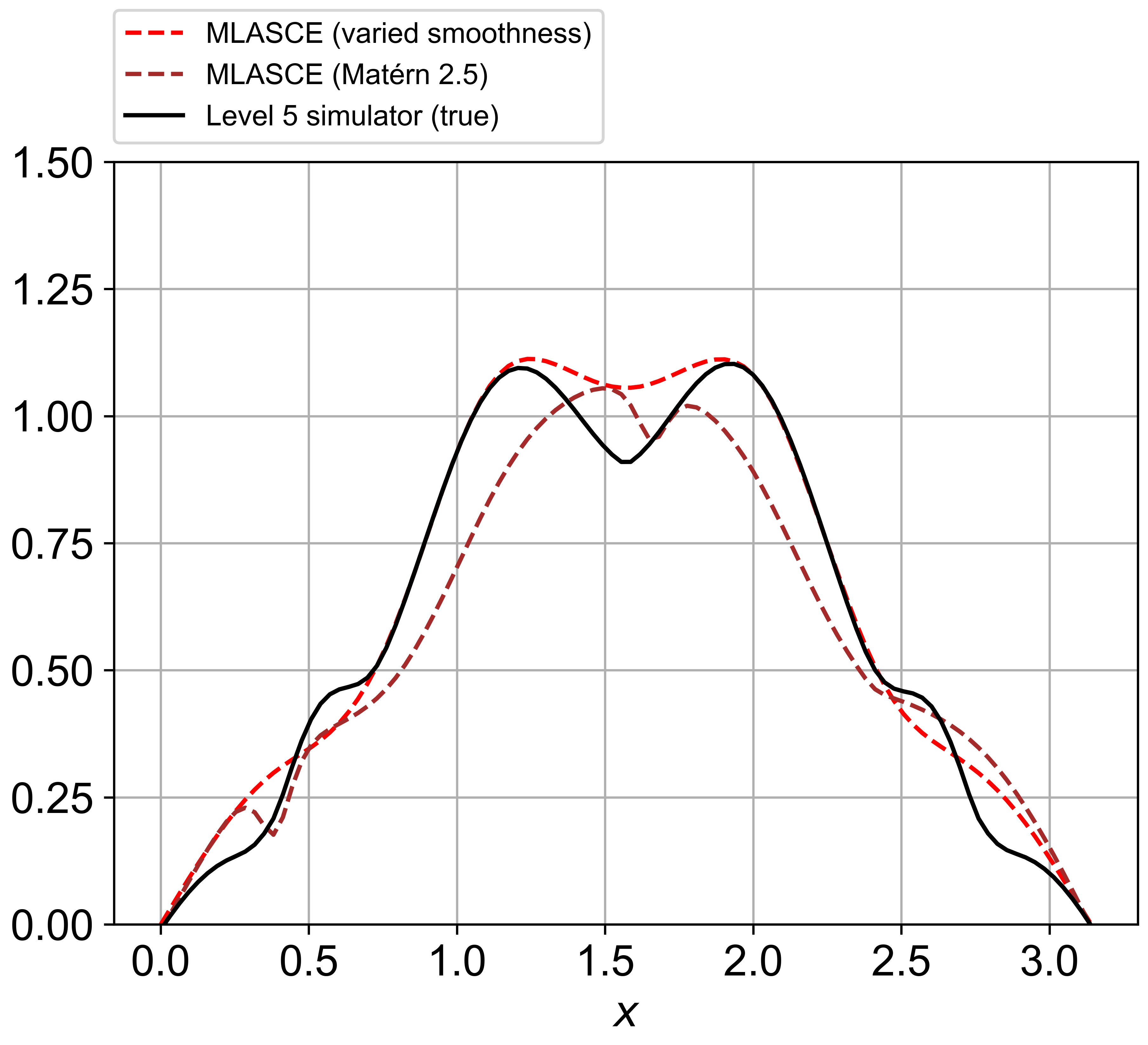}
\label{d}
} \\
\subfloat[\cite{le2014recursive} with Gaussian and Mat\'ern kernels, \cite{le2015cokriging} with Mat\'ern kernels and \cite{perdikaris2017nonlinear}.]{
\includegraphics[width=0.49\linewidth]{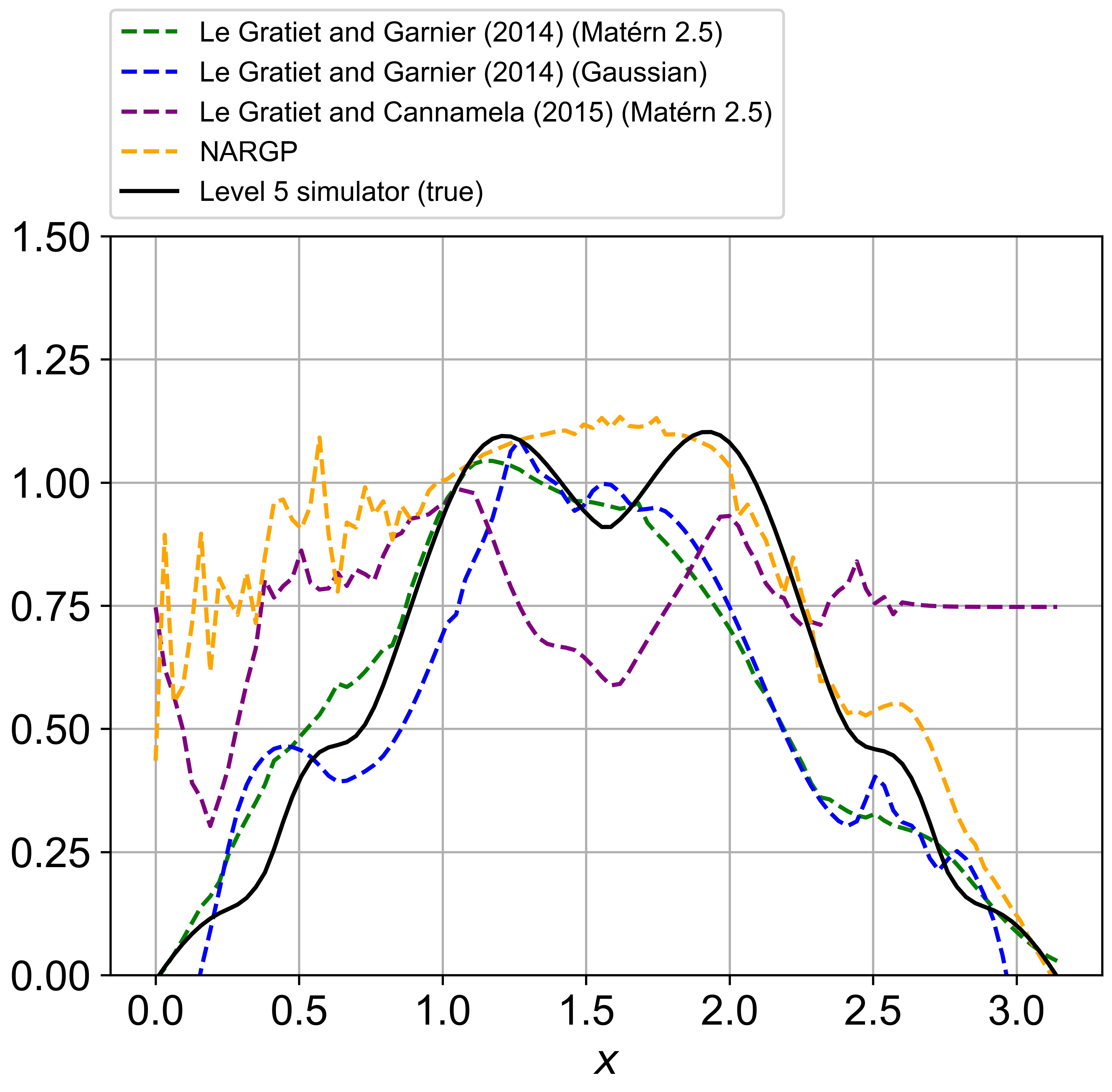}
\label{bump2_L2}
} 
\subfloat[$L^{2}$ error of  emulators as a function of budgets ranging from 1050 to 1400.]{
\includegraphics[width=0.49\linewidth]{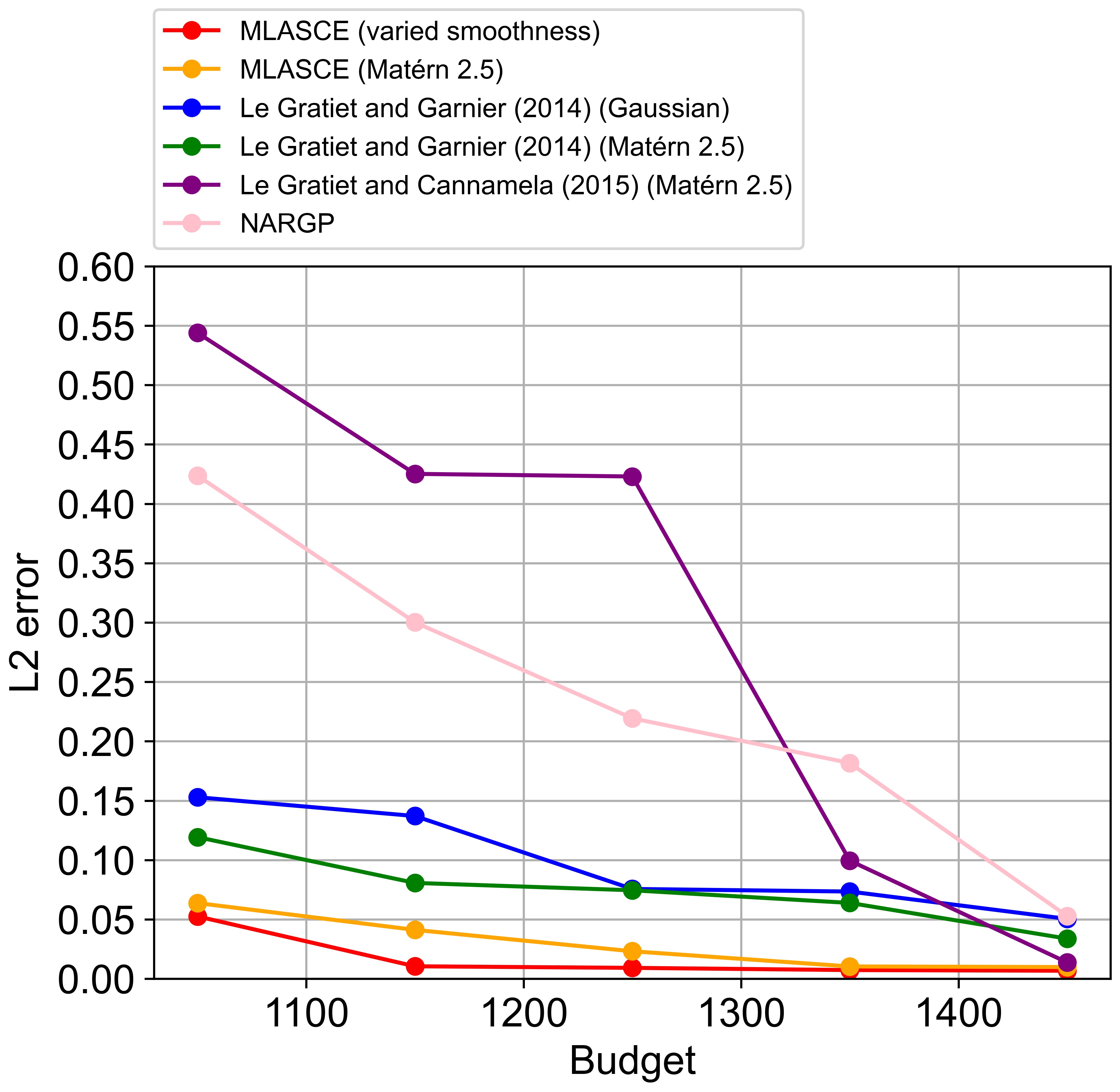}
\label{bump2_L2}
} 
\caption{Example 2: Computer simulations (a), emulations with different strategies (b), (c) and error in the approximation (d).}
\end{center}
\end{figure}

To implement MLASCE in the same way as before, we prepare five Gaussian processes $GP(0,K_{\nu_{l}, \lambda _{l},\sigma^{2}_{l}}(x, x^{\prime}) )$ with a Mat\'ern kernel for $l=1, \ldots, 5$. We set $\nu_{1}  = 3.5$, $\nu_{2}, \nu_{3} = 2.5$ and $\nu_{4},\nu_{5} = 1.5$ based on the smoothness of each computer simulation. We obtain the mean prediction by implementing MICE and fitting these GPs to the data. On the other hand, the other five Gaussian processes $GP(0,K_{2.5, \lambda _{l},\sigma^{2}_{l}}(x, x^{\prime}) )$ with a Mat\'ern kernel for $l=1, \ldots, 5$ are used to prepare another mean prediction. The smoothness parameter $\nu_{l}$ is chosen to be fixed in this case. For comparison, three other existing methods (\cite{le2014recursive} with Gaussian kernels and Mat\'ern kernels with $\nu= 2.5$ and  \cite{le2015cokriging} with Mat\'ern kernels $\nu= 2.5$) are implemented. The Gaussian kernel is too smooth compared with the computer simulations and Mat\'ern kernels with $\nu= 2.5$ may still be overly smooth for $f_{5}$. We assume five levels for this model and the specification of adjustment effect $\rho_{l}(x)$ is the same as in the previous part. The Maximum Likelihood method is used in fitting and nugget parameters are $10^{-8}$. We consider the different case of computational budgets $1050$, $1150$, $1250$, $1350$ and $1450$ thereby sufficient amount of initial design is prepared for \cite{le2014recursive,le2015cokriging}. NARGP is again compared in the same way and the same numbers of training points as in  \cite{le2014recursive} are prepared.

The numbers of samples are plotted in Figure \ref{data_imp_bump_2}. MLASCE with Mat\'ern kernel $\nu_{l} = 2.5$ still mimics the overall movements but gives the slightly less efficient result. The autoregressive co-kriging \citep{le2014recursive} has difficulty in capturing the less smooth movements of $\xi_{5}$, $\xi_{4}$ and $\xi_{3}$, although relatively a large amount of samples is prepared for every level (Figure \ref{data_imp_bump_2}). The emulator of \cite{le2014recursive} with Gaussian kernels produces some kinks around the points where $\xi_{i}(x)$ lies and their method with Mat\'ern kernel $\nu_{l} = 2.5$ also has difficulty in approximating the true function. Due to the relatively inefficient specification of the multi-fidelity modeling, the strategy by \cite{le2015cokriging} does not provide efficient emulators if the budget is limited. Moreover, even if the budget is rich enough, numerical difficulty may emerge due to too many samples from cheap simulations. This is because their strategy of sequential design emphasizes low-fidelity simulations too much as the improvements from simulators are directly weighted with corresponding computational costs. In this setting, the computational costs increase at an exponential rate (the most expensive simulation is 256 times more expensive than the cheapest) hence applying directly their strategy could be problematic. NARGP produces unreliable results again, potentially due to the lack of an efficient framework of the experimental design.  

These results show that an emulator by MLASCE with a flexible choice of small smoothness parameters performs better than the one with fixed and overly smooth kernels. More importantly, autoregressive co-kriging with Gaussian kernels may result in poor performance if the target is not very smooth. Besides, even if the kernel is not so smooth, it may still have difficulty in mimicking low-smooth behaviors. 
The comparison of these emulators in terms of $L^{2}$ error is illustrated in Figure \ref{bump2_L2}.  To show the robustness of MLASCE, we prepare the box plots of the $L^{2}$ errors of MLASCE in this numerical examples (Figure \ref{boxplot_flex} and \ref{boxplot_fix}). Although the variations of the $L^{2}$ error are relatively large when the budget is more constrained, the predictive accuracy of MLASCE seems stable.

\begin{figure}[h]
\begin{center}
\subfloat[Number of data points per level (MLASCE, varied smoothness)]{
\includegraphics[width=0.45\linewidth]{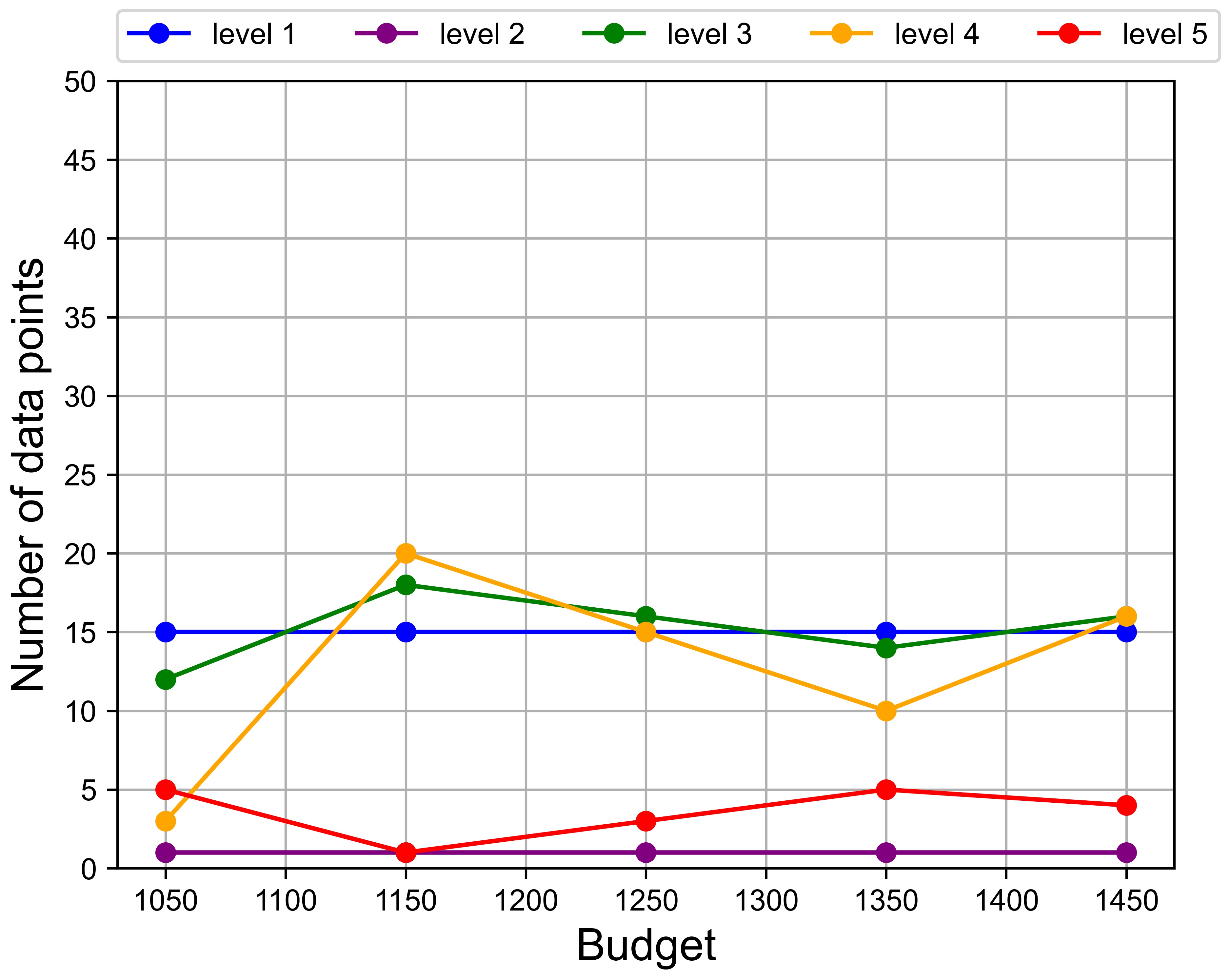}
\label{data_imp_bump_2}
} 
\subfloat[Number of data points per level (MLASCE, fixed smoothness)]{
\includegraphics[width=0.45\linewidth]{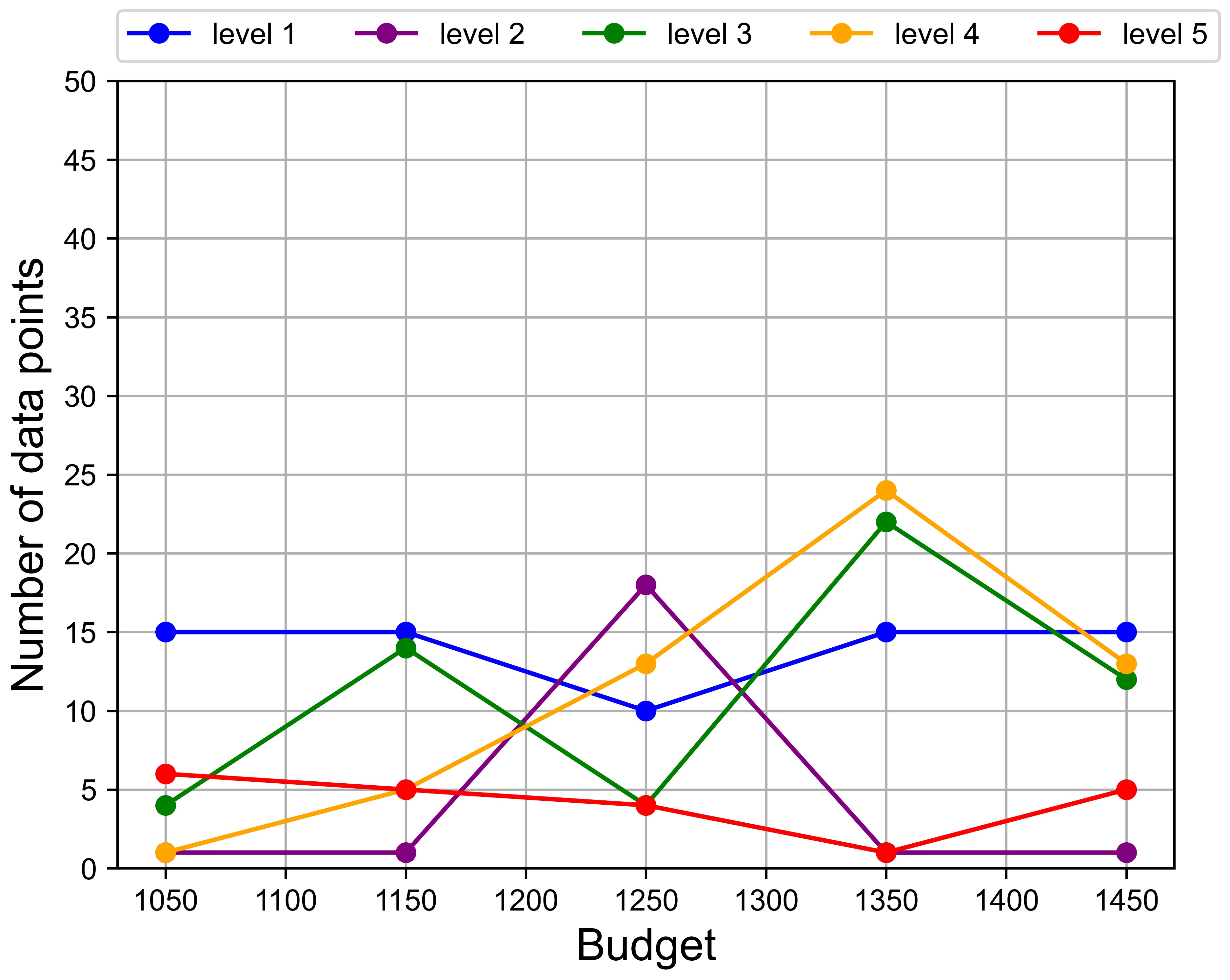}
\label{bump2_L2_1}
} 
\\
\subfloat[Number of data points per level (\cite{le2014recursive} and \cite{perdikaris2017nonlinear})]{
\includegraphics[width=0.45\linewidth]{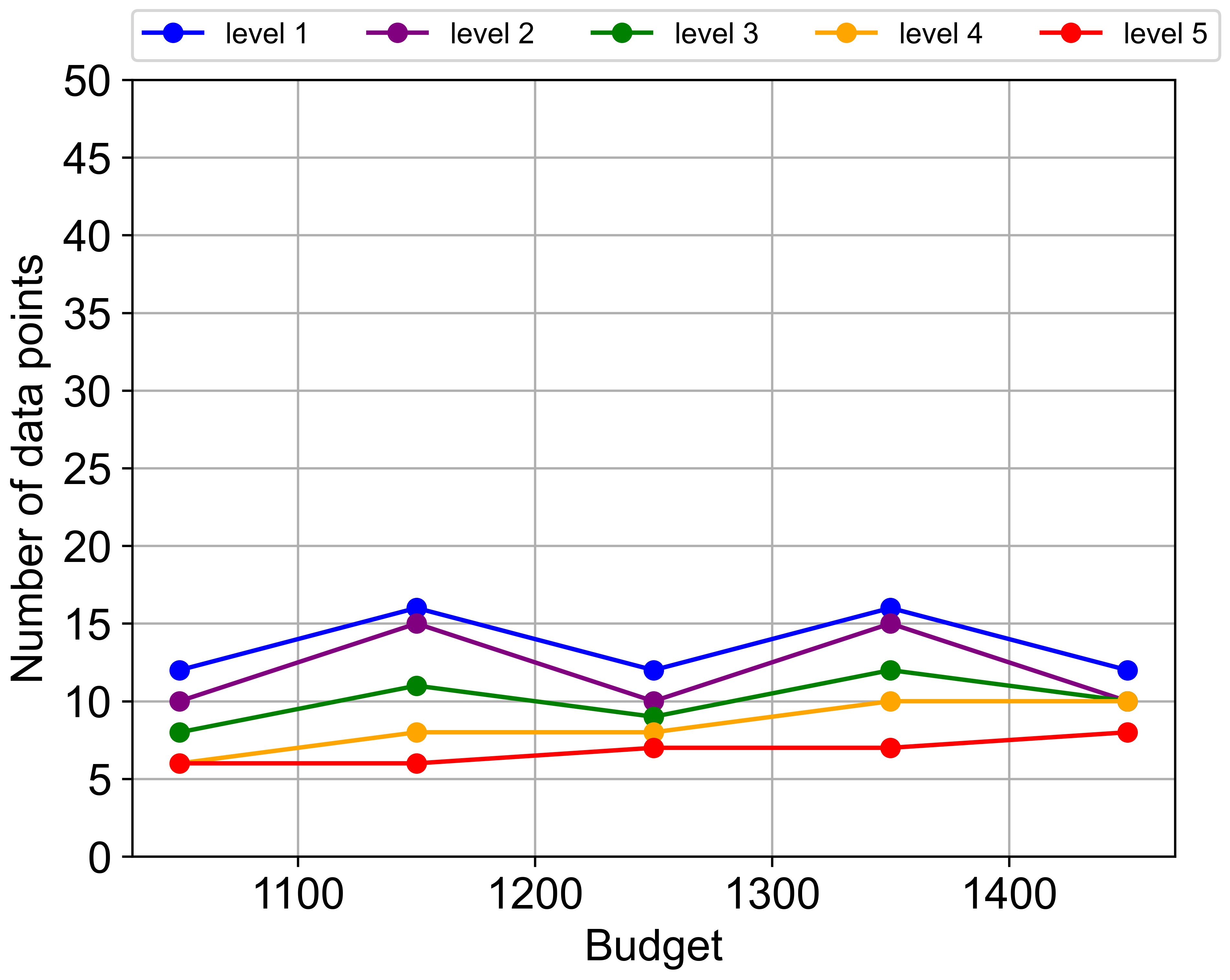}
\label{bump2_L2_2}
} 
\subfloat[Number of data points per level, level 2 and level 3 coincides. (\cite{le2015cokriging})]{
\includegraphics[width=0.45\linewidth]{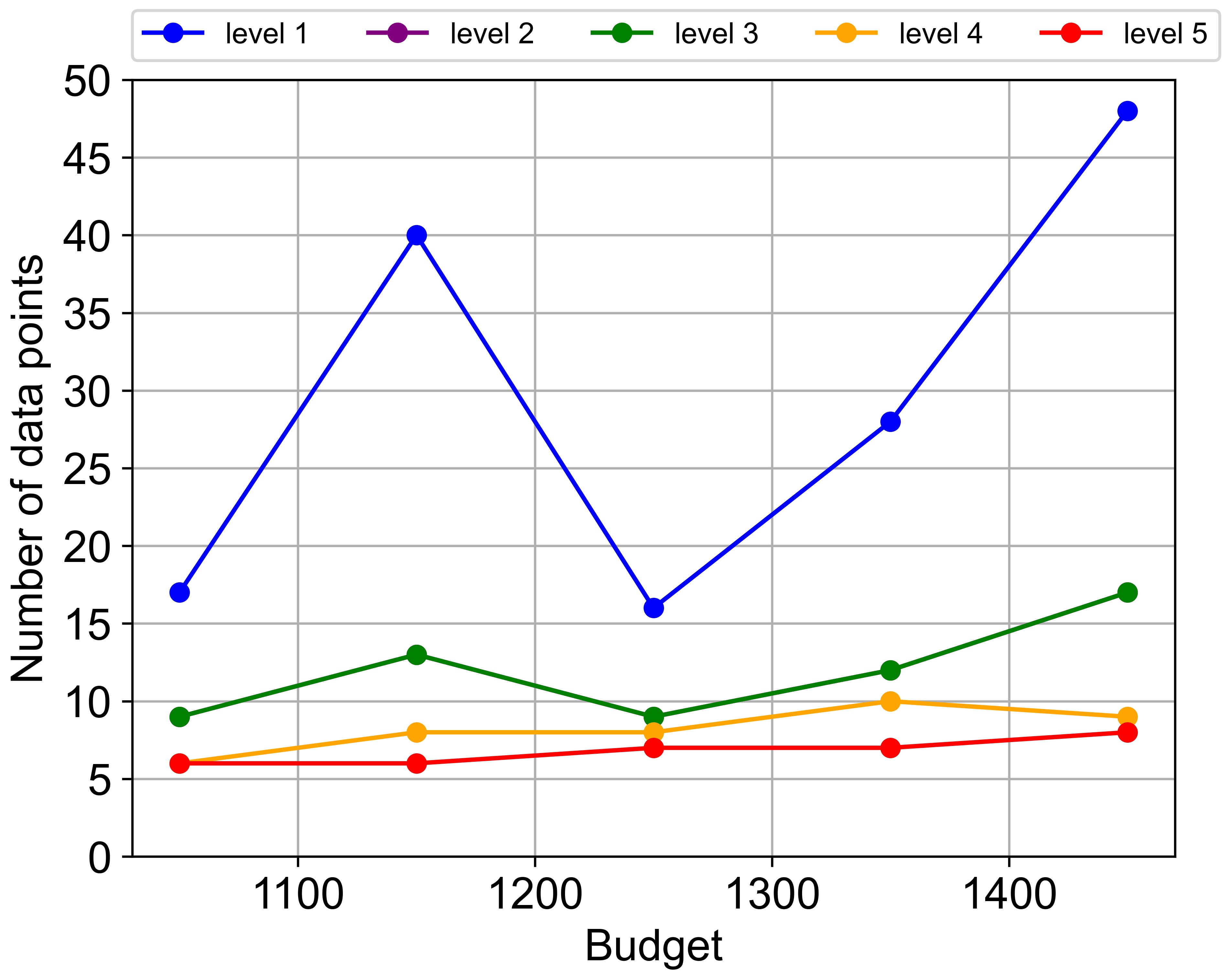}
\label{bump2_L2_3}
} 
\caption{The numbers of runs of different emulators.}
\end{center}
\end{figure}

\begin{figure}[h]
\begin{center}
\subfloat[Boxplot]{
\includegraphics[width=0.45\linewidth]{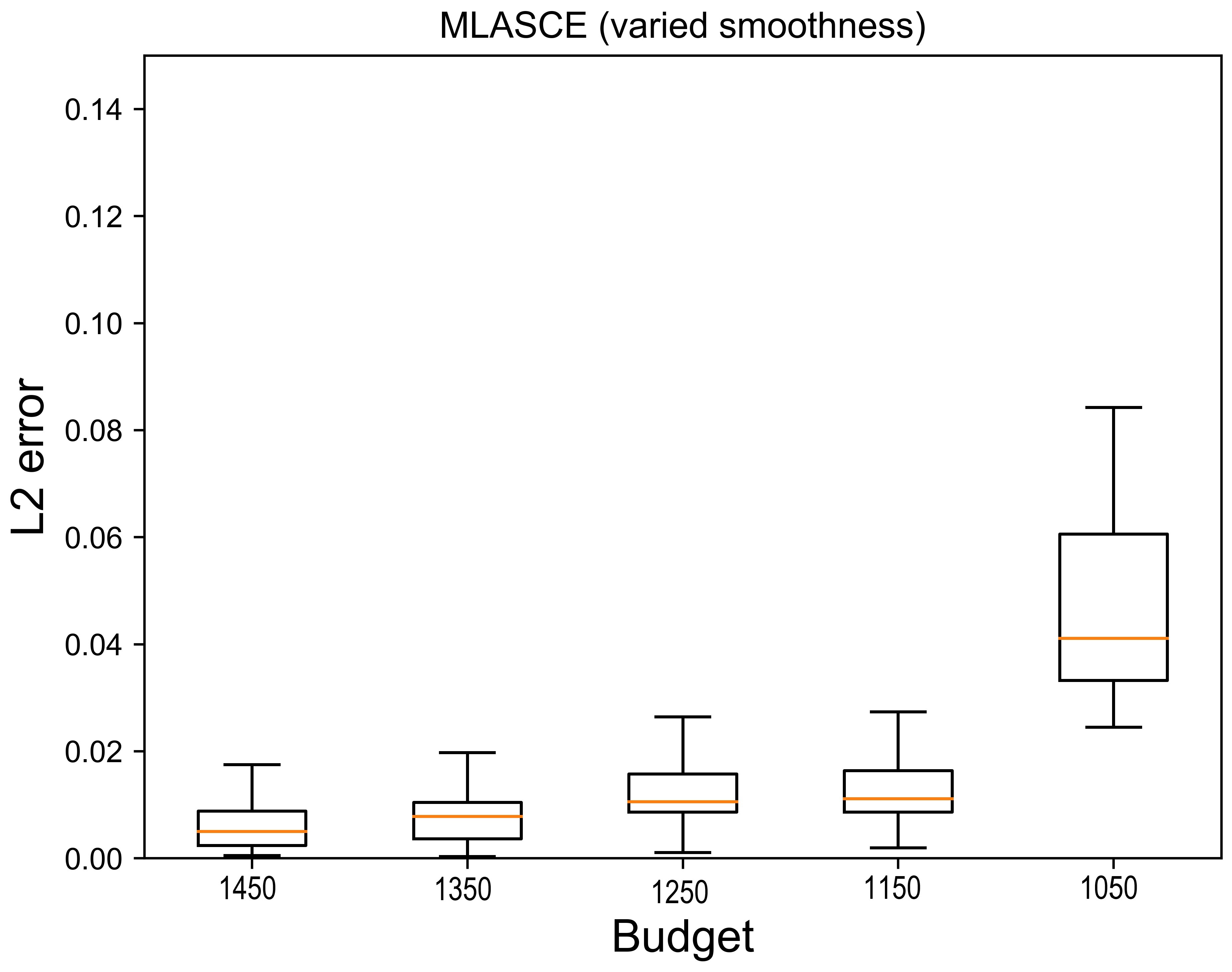}
\label{boxplot_flex}
} 
\subfloat[]{
\includegraphics[width=0.45\linewidth]{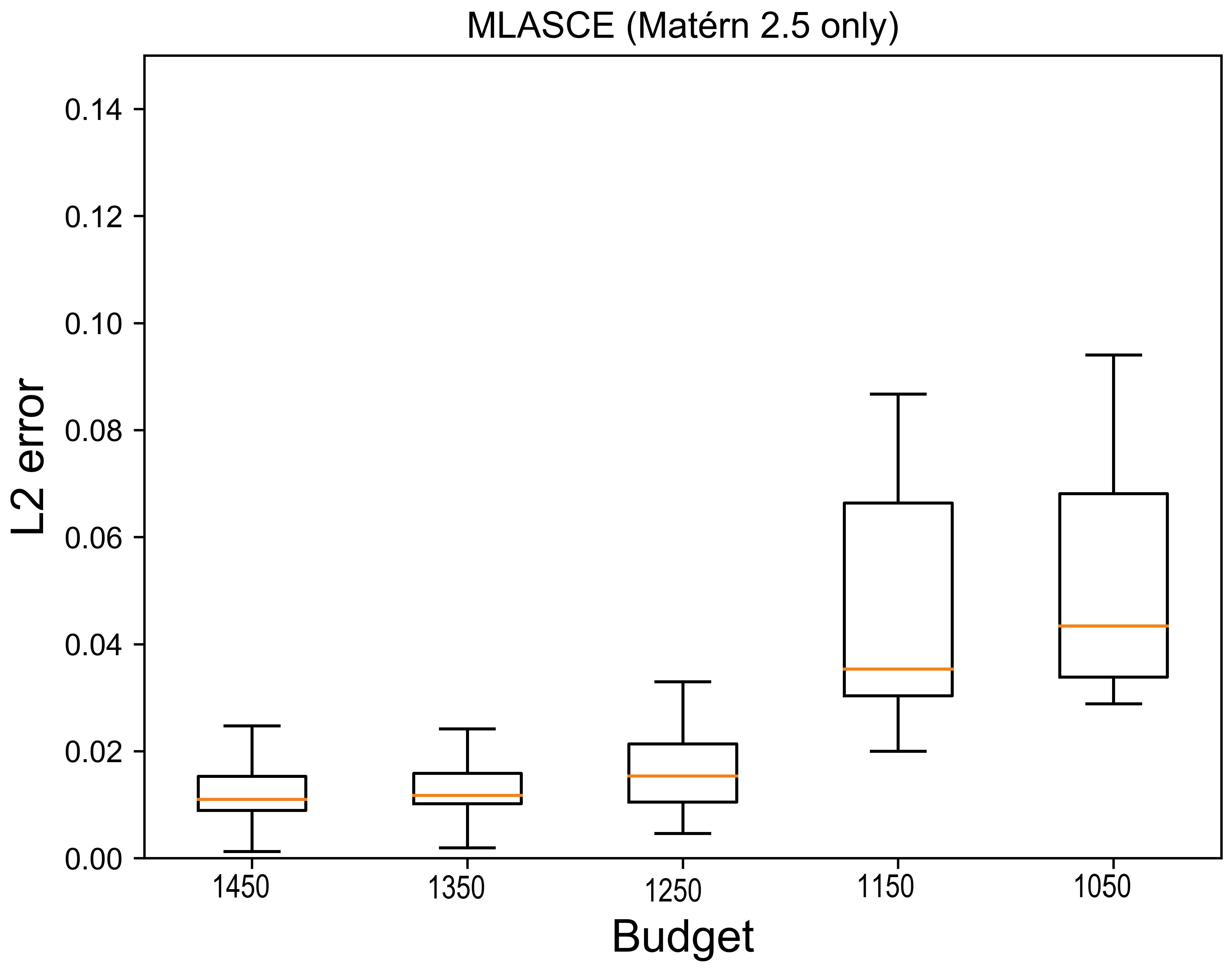}
\label{boxplot_fix}
} 
\caption{The box plots of the $L^{2}$ errors, depending on the varied locations of the initial design point.}
\end{center}
\end{figure}

\subsection{Example 3: multi-level tsunami emulation with multi-dimensional inputs}
\label{numerical_exmaple_tdac}
In this section, we consider a more realistic multi-dimensional simulation (input dimension 3) and show that MLASCE outperforms the standard framework. As our focus is on creating a design of experiment, we did not include NARGP in this example. The simulator is TDAC (Tsunami Data Assimilation Code), see \url{https://github.com/Team-RADDISH/tdac}. We do not employ the data assimilation component. 
TDAC adopts a simple 2-D linear long-wave equation as the governing equation \citep{maeda2015successive} for tsunamis observed in the deep ocean:
$$
\begin{aligned}
&\frac{\partial \gamma(x, y, t)}{\partial t}=-\frac{\partial M(x, y, t)}{\partial x}-\frac{\partial N(x, y, t)}{\partial y} \\
&\frac{\partial M(x, y, t)}{\partial t}=-g D(x, y) \frac{\partial \gamma(x, y, t)}{\partial x} \\
&\frac{\partial N(x, y, t)}{\partial t}=-g D(x, y) \frac{\partial \gamma(x, y, t)}{\partial y}
\end{aligned}
$$
where $\gamma$ is the tsunami height, $(M, N)$ are the vertically integrated horizontal velocity components of the tsunami in the $x$ and $y$ directions, $g$ is the gravitational acceleration constant $\left(9.80665 \mathrm{~m} / \mathrm{s}^{2}\right)$, and $D(x, y)$ is the sea depth (3000m over the whole domain). The numerical solution of this simulation is based on the finite difference method with the first-order accuracy. To implement TDAC, one should specify the initial conditions of the governing equations through the input variables: the initial wave height $\zeta_{1}$ and the location of the elevation in the $x$-$y$ coordinates $(\zeta_{2},\zeta_{3})$. A circle defines the boundary of the initial elevation area (in yellow in Figure \ref{tdac_coordinate}), whose radius is $15$ km.

We assume that the entire domain of the simulation is $[0,500] \times [0,500]$ (km) in the $x-y$ coordinate and fix the gauge point where the wave height is collected at $(80,80)$. We set values for the parameters that determine the initial conditions of the governing equation. We are interested in how the maximum wave height $\max_{0\leq t \leq T}\gamma(80,80,t) $ at the gauge point changes depending on the values of these parameters. Moreover, as one would usually assume nonlinear relationships among the input variables in practical situations (e.g., the Okada dislocation model \citep{okada_1985} in a tsunami simulation), we specify the following toy nonlinear equations (a reparameterization mimicking a nonlinear influence) among the input variables $\zeta_{1}$, $\zeta_{2}$, and $\zeta_{3}$ by introducing other independent variables $s_{1}$, $s_{2}$ and $s_{3}$:
\begin{eqnarray*}
\zeta_{1} &=& \sum_{k=1}^{2}\sum_{j=1}^{2}\exp\bigl(-12|s_{2}-\psi_{j}|^{3/2}-12|s_{3}-\psi_{k}|^{3/2} \bigr) \bigl(s_{1} -2 \bigr) + 2 \\
\zeta_{2} &=& s_{2} \\
\zeta_{3} &=& s_{3} 
\end{eqnarray*}
where $\psi_{1} = 31.25, \psi_{2} = 43.75$ so that the peaks of $\zeta_{1}$ are in the middle of the domain of $s_{2}$, and $s_{3}$. Note that $s_{1} \in [2,9]$ represents the energy of the wave elevation, and both $s_{2}$ and $s_{3}$ (location of the initial wave elevation) are in $[25,50]$. The initial wave height is determined by not only the energy of the wave elevation $s_{1}$ but also its location $s_{2}$ and $s_{3}$.
Our input-output structure is $\gamma_{\text{max}}(s_{1},s_{2},s_{3})$
where $\gamma_{\text{max}}$ is the maximum wave height at the gauge point $(80,80)$ over the simulation time period $[0,T]$. We set $T=1500$ (seconds).

Next, we run TDAC with three different sizes of the grids: 4 km (level 1), 3 km (level 2) and, 2km (level 3). The CFL condition is set to 0.2 for every simulation to avoid numerical instability based on the fact that the wave velocity is around 170 m/s. Figure \ref{tdac_slip} shows the typical time series of the wave height at the gauge point $(80,80)$ with the different discretization levels. We observe from Figure \ref{tdac_slip} the maximum wave height tends to be smaller as the discretization becomes coarser. Therefore, as far as we are interested in the maximum elevation at the gauge point, these simulations are valid examples for illustrating our methodology.

\begin{figure}[h]
\begin{center}
\subfloat[The gauge point (red circle) and the location of the initial wave elevation.]{
\includegraphics[height=0.23\textheight]{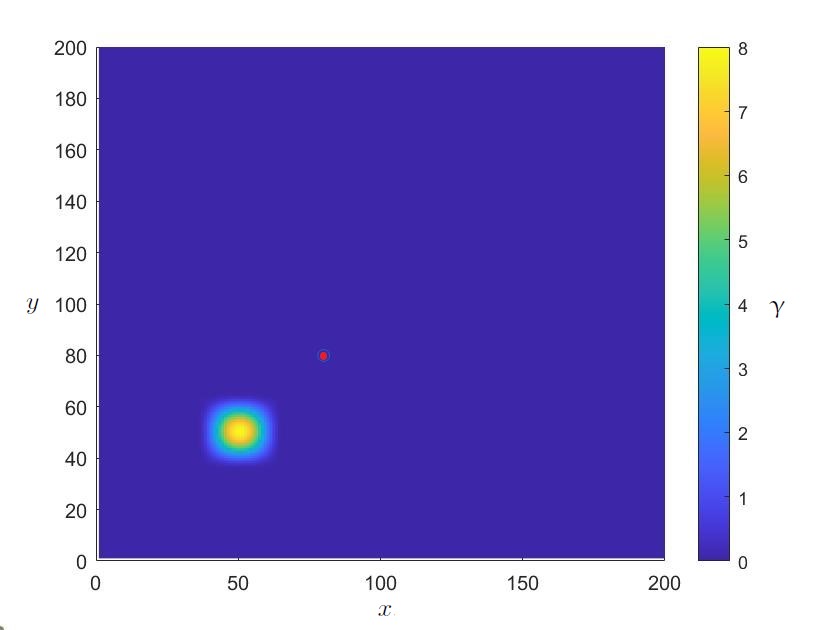}
\label{tdac_coordinate}
} 
\subfloat[The typical results of the simulations with different fidelities ($\zeta_{1} = 9$, $\zeta_{2} = \zeta_{3} =  25$).]{
\includegraphics[height=0.23\textheight]{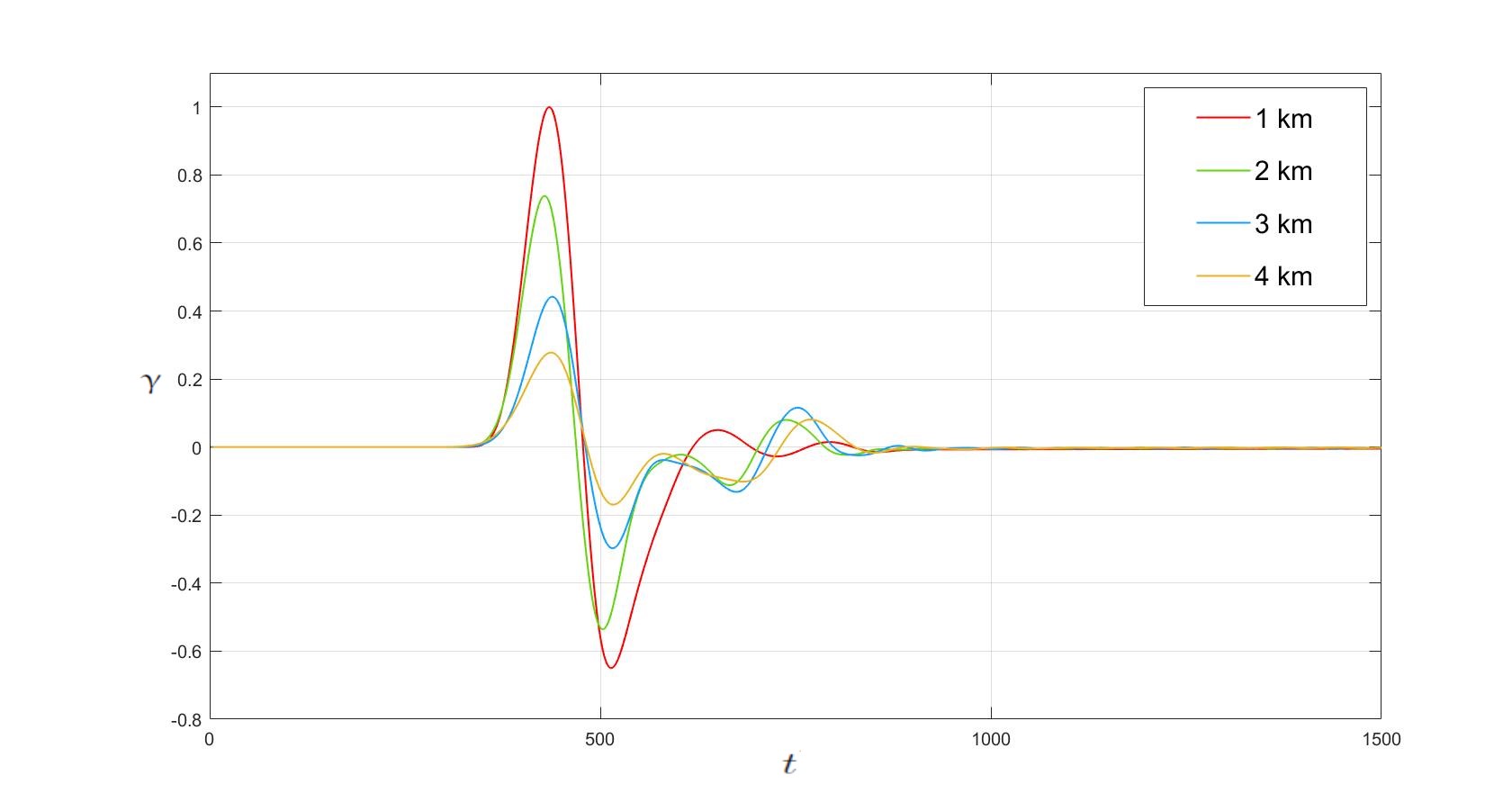}
\label{tdac_slip}
} \\
\caption{Time series of wave height at the gauge point at $(80, 80)$ according to resolutions from 1 km to 4 km.  }
\end{center}
\end{figure}

Let $\gamma_{\text{max},1}(s)$, $\gamma_{\text{max},2}(s)$ and $\gamma_{\text{max},3}(s)$ denote the maximum wave height provided by TDAC with the grid size $4$ km, $3$ km and $2$ km respectively where $s=(s_{1},s_{2},s_{3}) \in [2,9] \times [25,50] \times [25,50]$ and $T=1500$. We call $\gamma_{\text{max},l}(s)$ the level $l$ simulation. Our goal is to construct the multi-fidelity GP emulators of $\gamma_{\text{max},3}(s)$ under the restriction of the computational budget $T_{0}$. We prepare the incremental functions $\tilde{\gamma}_{1}(s)=\gamma_{\text{max},1}(s)$ and $\tilde{\gamma}_{l}(s)=\gamma_{\text{max},l}(s) - \gamma_{\text{max},l-1}(s)$ for $l=2,3$. Then we prepare three mutually independent GPs $\eta_{l}(x) \sim \text{GP}(0,K_{l}(x), x^{\prime})$ for $\tilde{\gamma}_{l}(s)$. As observed in \citep{gopinathan2021probabilistic}, the Mat\'ern 5/2 kernel is a common choice in constructing a GP emulator for a tsunami simulation. Therefore, we select the Mat\'ern 5/2 kernel for the specification of the kernels in the GP emulations. The nugget parameters and parameter estimation method are the same as in the previous subsection. The typical computational time for running a certain level of TDAC is in the Table \ref{cost_tdac} and we choose the different values, 168, 196, 224, 252, 280 seconds for $T_{0}$ as examples of a limited amount of available computational resource. 

For the purpose of validating the performance of MLASCE, We prepare 80 true values of $\gamma_{\text{max},3}(s)$,  $\bigl\{s =(s_{1},s_{2},s_{3}) \in X_{\text{validation}} \mid X_{\text{validation}} = (3, 4.25, 5.5, 6.75, 8) \times (30, 35, 40, 45) \times(30, 35 , 40, 45) \bigr\}$ and none of these data points is included in our experimental designs. These true values and our predictions are compared by computing the root mean square errors (RMSE) as a measure of the overall prediction error. Moreover, to show the advantage of MLASCE over the non-multilevel methodology, we prepare the single level emulator, which is based on the training data only from the level 3 simulation $\gamma_{\text{max},3}(s)$ using MICE for the design under the same budget. The result of validation is shown in Figure \ref{tdac_RMSE}. 
MLASCE is more reliable than the single level emulator under the whole range of the budget and the prediction error of our multi-fidelity emulator also gets reduced as the budget increases, albeit at a slow rate. The RMSE of the single level emulator declines with the bigger budget but it is well above that of MLASCE. The number of design points is in Table \ref{number_tdac} and the experimental designs under $T_{0} = 280$ are shown in Figure \ref{tdac_design_1} - \ref{tdac_design_3}. MLASCE successfully allocates the total budget among the different levels of the simulations by evaluating the magnitude of their contributions.

\begin{table}[htb]
 \centering
 \caption{The typical cost (computational time) for each level of TDAC.}
  \scalebox{0.8}{\begin{tabular}[width=0.5\textwidth]{|c|c|c|c|c|} \hline 
$\gamma_{\text{max},1}(s)$  & $\gamma_{\text{max},2}(s)$  & $\gamma_{\text{max},3}(s)$ & $\tilde{\gamma}_{2}(s)$ & $\tilde{\gamma}_{3}(s)$ \\ \hline 
$4$ (second) & $8$ (second)  & $ 20$ (second) & $ 12$ (second) & $ 28$ (second)    \\ \hline
  \end{tabular}}
  \label{cost_tdac}
\end{table}

\begin{figure}
\begin{center}
   \includegraphics[width=0.4\linewidth]{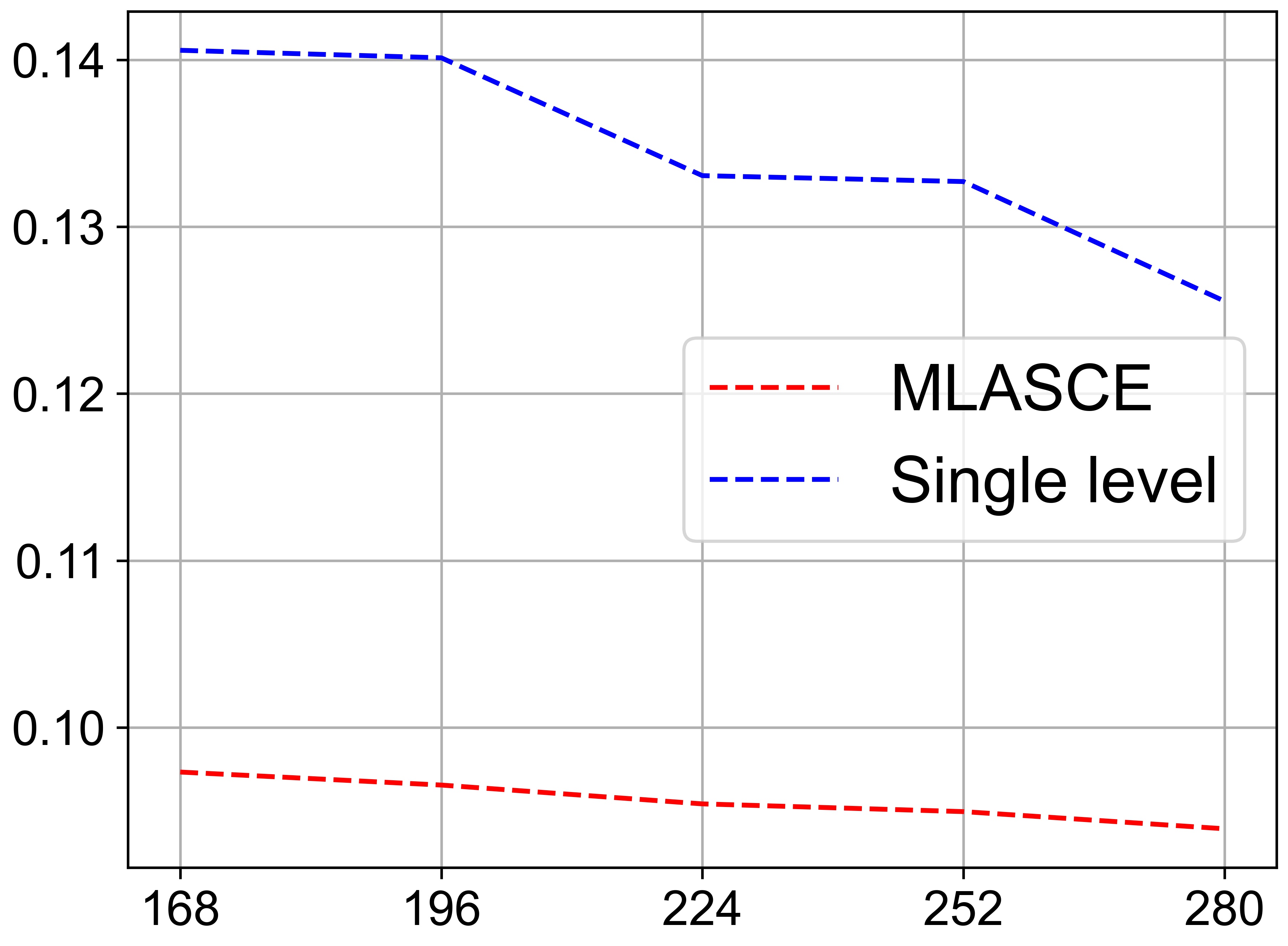} 
\caption{Tsunami wave example. RMSE with budget $T_{0} = 168$ to $280$.}
\label{tdac_RMSE}
    \end{center}
\end{figure}

\begin{figure}
\begin{center}
   \subfloat[Design points of $s_{1}$ and $s_{2}$. ]{\includegraphics[width=0.3\linewidth]{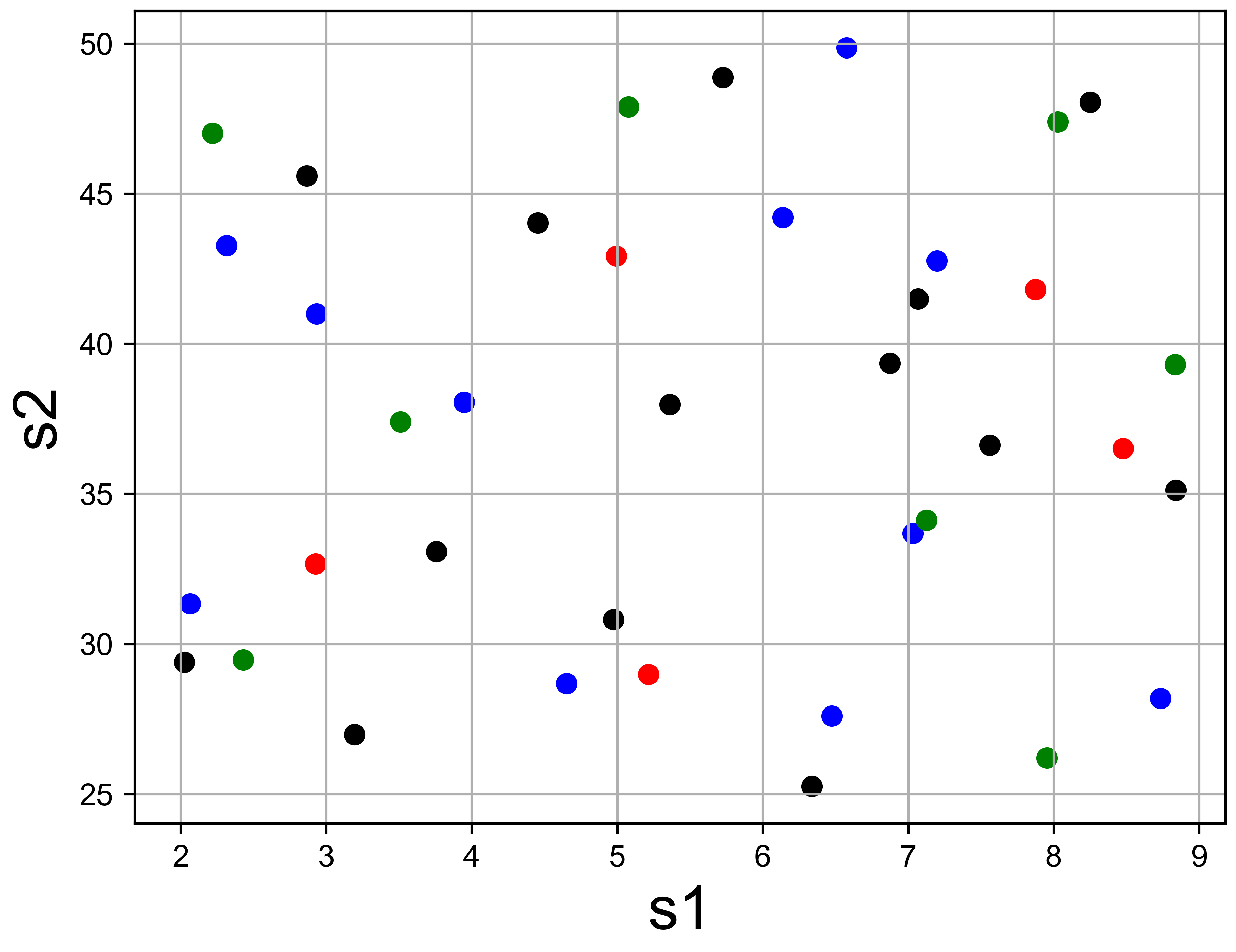} \label{tdac_design_1}}
    \hfil 
    \subfloat[Design points of $s_{1}$ and $s_{3}$.]{\includegraphics[width=0.3\linewidth]{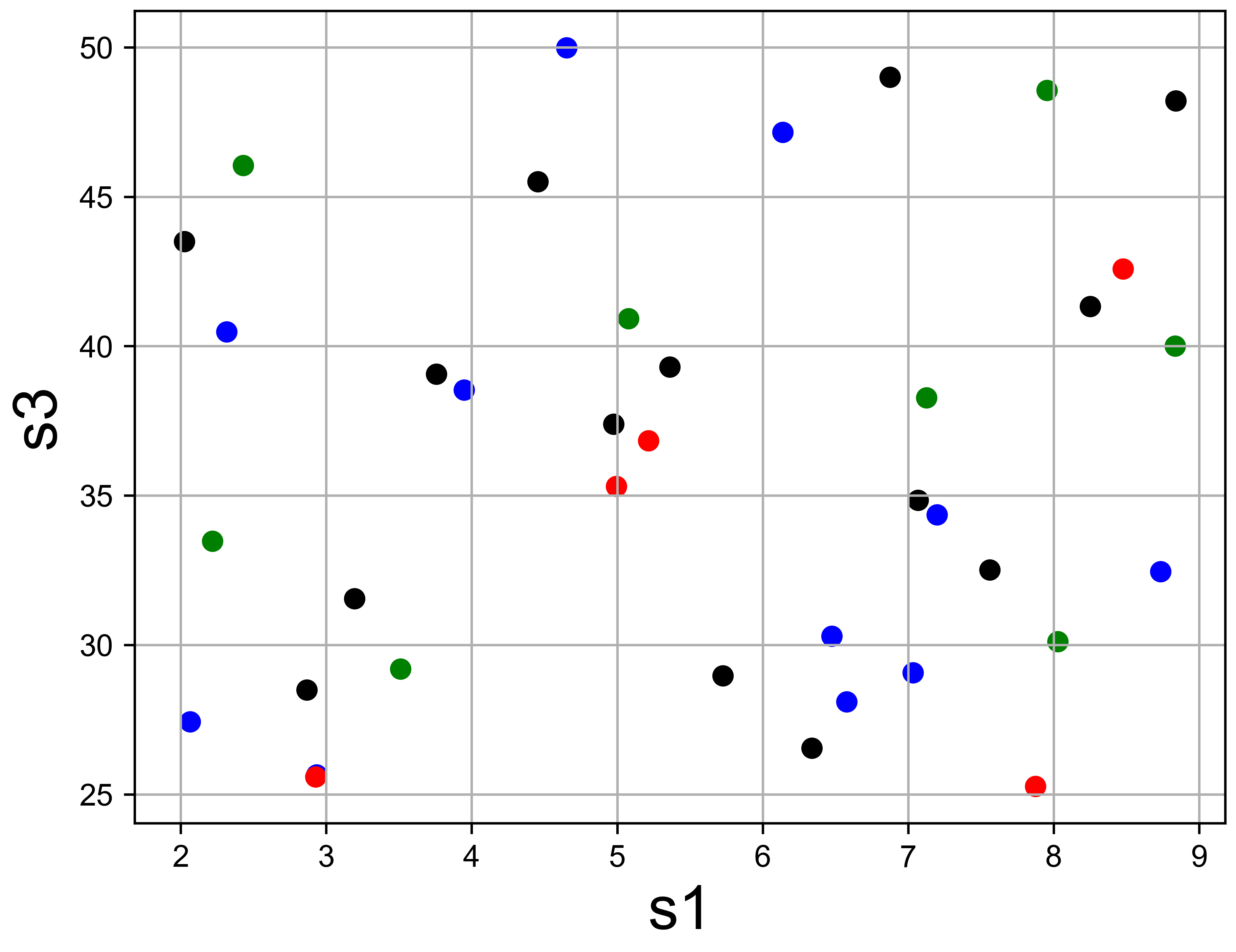} \label{tdac_design_2}}
    \subfloat[Design points of $s_{2}$ and $s_{3}$.]{\includegraphics[width=0.3\linewidth]{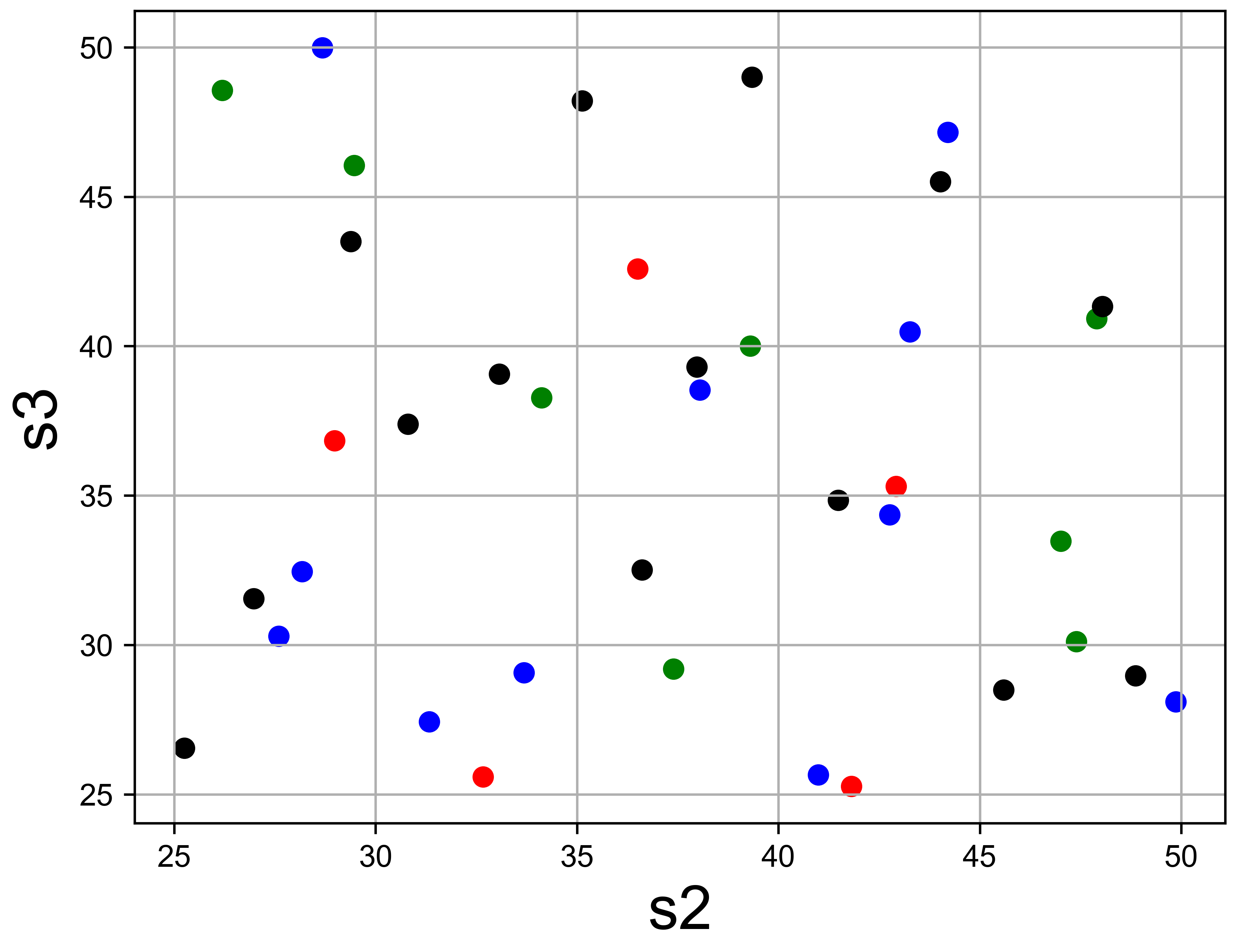} \label{tdac_design_3}}
    \caption{The experimental design of the multi-fidelity emulator and single level emulator with $T_{0}=280$. The red, green, blue and black points denote the training data points for $\tilde{\gamma}_{1}(s)$, $\tilde{\gamma}_{2}(s)$, $\tilde{\gamma}_{3}(s)$ and single level emulator.}
    \end{center}
\end{figure}

\begin{table}[htb]
 \centering
 \caption{The number of design points at each level according to budget.}
  \scalebox{0.8}{\begin{tabular}[width=0.5\textwidth]{|c|c|c|c|c|c|} \hline 
  Budget ($T_{0}$) & 168  & 196 & 224 & 252 & 280  \\ \hline 
 Level 1 ($\tilde{\gamma}_{1}(s)$)  & 9 & 9 & 10 & 11 & 11  \\ \hline
 Level 2 ($\tilde{\gamma}_{2}(s)$) & 4 &  4 & 6  & 8 & 8 \\ \hline
 Level 3 ($\tilde{\gamma}_{3}(s)$) & 3 &  4 & 4  & 4 & 5 \\ \hline
 Single level ($\gamma_{\text{max},3}(s)$) & 8 &  9 & 11  & 12 & 14 \\ \hline
  \end{tabular}}
  \label{number_tdac}
\end{table}

\subsection{Example 4: non-hierarchical robustness}
\label{minor_correction_decomposition}

The decomposition structure (\ref{decomp}) assumes that the incremental functions $\delta_{i}(x)$ are summed up in a hierarchical way. In other words, the fidelity of the simulations $y_{l}(x)$ improves with $l$. It might be of interest to observe how MLASCE behaves when the hierarchical order in the decomposition structure (\ref{decomp}) is violated. Assuming three levels of simulation, not ordered in terms of fidelity, our purpose is to predict $y_{3}$ while enforcing MLASCE to this structure of the non-hierarchical fidelity. Using $f_{1}(x) = \cos{x}$, $f_{2}(x) = 0.5\cos{2x}$, $f_{3}(x) = 0.3\sin{3x}$, $f_{4}(x) = 0.3\cos{4x}$, $f_{5}(x) = 0.3\sin{4x}$, our simulators are $y_{1}(x) = f_{1}(x)$, $y_{2}(x) = f_{1}(x) + f_{2}(x) + f_{3}(x)$, $y_{3}(x) = f_{1}(x) + f_{4}(x) + f_{5}(x)$.

We compare MLASCE and co-kriging methods, \citep{le2014recursive} and \citep{le2015cokriging}. The Mat\'ern 2.5 kernels and the nugget parameter $10^{-8}$ are assumed for every framework and $\rho(x) = \beta_{0} + \beta_{1}x + \beta_{2}x^{2}$ is determined for the model of \citep{le2014recursive} while constant $\rho$ is set for \citep{le2015cokriging}. The domain is $[0, \pi]$ and the computational costs for level 1, level 2 and level 3 simulations are $1$, $2$ and $4$ with the budgets $35, 47, 59,71,83$. The $L^{2}$ errors are compared in Figure \ref{minor_multi_error} and MLASCE still outperforms the co-kriging methods of \citep{le2014recursive,le2015cokriging} albeit with slightly better predictions than that of \cite{le2014recursive}.

\begin{figure}
\begin{center}
   \subfloat[The $L^{2}$ errors.]{\includegraphics[width=0.49\linewidth]{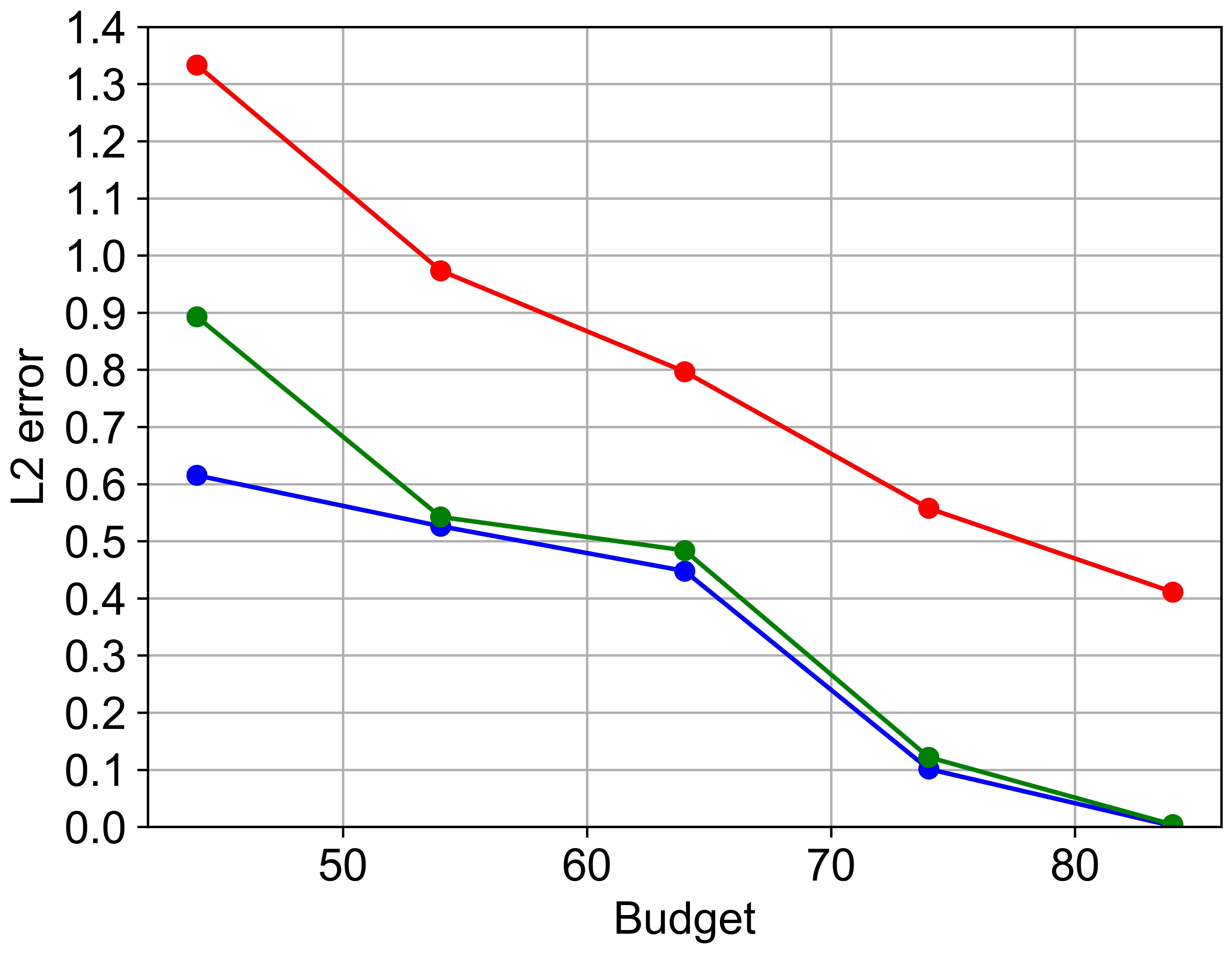} \label{minor_multi_error}}
    \subfloat[The numbers of the design points. The solid, dotted, dashed lines correspond to level 1, level 2 and level 3.]{\includegraphics[width=0.49\linewidth]{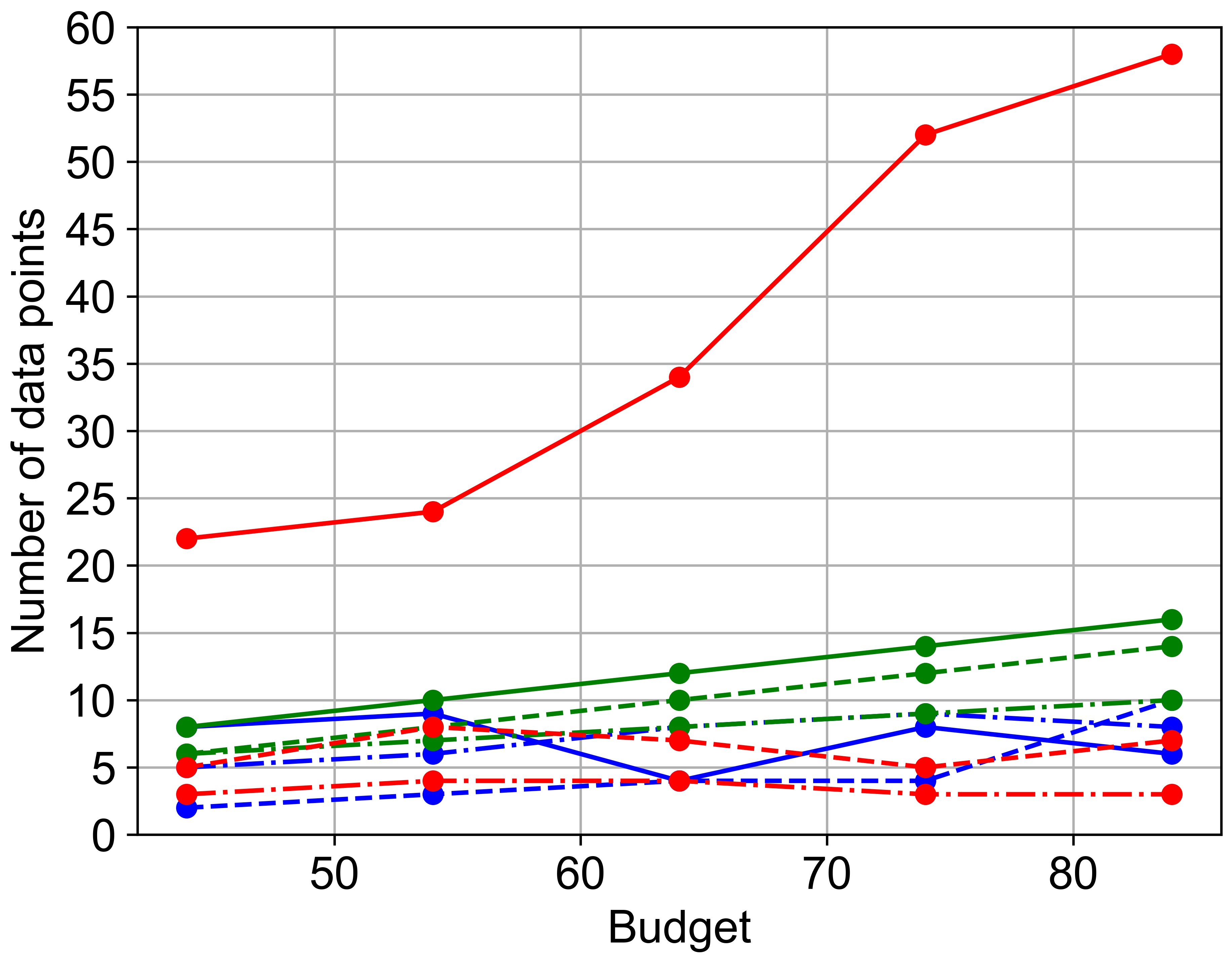} \label{minor_multi_num}}
    \hfil 
    \caption{Example 4: Errors and sampling points: blue, green and red lines correspond to MLASCE, \citep{le2014recursive} and \citep{le2015cokriging} respectively.}
    \end{center}
\end{figure}


\section{Conclusions}
We presented a novel method of multilevel adaptive sequential design of computer experiments (MLASCE) in the framework of multi-fidelity Gaussian process emulators. This strategy allows us to allocate efficiently limited computational resource over simulations of different levels of fidelity and build the GP emulator accordingly. Furthermore, the allocation of computational resource is shown to be the solution of a simple optimization problem in a special case that turns out to match well the general optimal allocation of resources. This allows ambitious computer experiments at multiple fidelities to be vetted before launching the expensive suite of simulations, thus reassuring modelers that their budget will be effective in building the multilevel emulator. The strength of the suggested approach is that it considerably improves the accuracy of the prediction by a GP emulator compared to other existing multi-fidelity models. One of the most important consequences of our proposed method is that the size of improvements from using more sophisticated, hence more expensive, simulations is evaluated while the convergence of the prediction is guaranteed. 

First, we decompose the most refined simulation into the functions of increments and scales of these functions are measured in the RKHS norms. Mat\'ern kernels are used in the covariance structures of GPs for each increment and a principle for choosing the smoothness parameters is suggested. Second, using the RKHS norms of the conditional means of increments, we introduce MLASCE, an algorithm of sequential designs for simulations of different levels of fidelity based on MICE. Third, we present the explicit form of optimal numbers of runs in the special cases where the designs are located in tensor grids.

MLASCE is applied to one-dimensional toy examples in comparisons with other existing multi-fidelity GP emulators. We further employed MLASCE in a multidimensional case comparing it with the non-multi-fidelity emulator. From these illustrations we can derive three conclusions. First, when the computational budget is small (i.e. the number of high fidelity runs is tiny), other existing multi-fidelity methods do not perform as well as MLASCE, or sometimes cannot even construct an emulator. Our method is more robust whenever the budget is small, because of the increased flexibility in adaptively choosing the level and the input locations. Secondly, the choice of smoothness parameters of the GP emulators of our method is adjustable. The two cases of series of either fixed or different smoothness parameters are considered. In the usual case where the smoothness of a simulation is unknown, MLASCE with varied levels of smoothness parameters across different levels shows more accurate results. Indeed, our method is able to follow the usual growth in roughness of differences across two consecutive levels whenever fidelities increase. Obviously, this is now an \textit{ad hoc} approach to the variation in smoothness across levels, so more research, e.g in terms of eliciting prior knowledge about smoothness across levels would improve our method further. 

MLASCE performs better than the non-multi-fidelity framework in a realistic multi-dimensional computer simulation. The numerical solution of this system is also dependent on a complex relationship among the input variables, which is motivated by a practical simulation. The predictions by the non-multi-fidelity emulator are unreliable due to the lack of training data under the limited computational budget. In contrast, MALSCE achieves a decent quality of predictions with an efficient allocation of the small budget.

Finally, our paper has only illustrations in dimensions of the input space up to 3. In higher dimensional settings we expect MLASCE to perform possibly even better as the choice of next point and fidelity is done jointly and efficiently, and the curse of dimensionality requires a more efficient algorithm to navigate a more ``empty'' space. However, this is an area of research that needs more investigation and tuning but will benefit the upcoming exascale computing in which fidelities and dimensions are naturally large. 


\bibliographystyle{abbrv}  
\bibliography{references}

%





\appendix
\section{Proof of Theorem 2}
\begin{proof}
Let $L_{N}^{\theta}$ be a closure of $\{\sum_{i=1}^{N}a_{i} K_{\nu,\theta}(\cdot, x_{i}) ; a_{i} \in \mathbb{R} \}$, that is

\[L^{\theta}_{N} = \overline{\{  \sum_{i=1}^{N} a_{i} K_{\nu,\theta}(\cdot, x_{i}) ; a_{i} \in \mathbb{R}  \}} .\]

$\{ K_{\nu,\theta}(\cdot, x_{i}) \}_{i=1}^{N}$ are nonorthogonal basis in $L^{\theta}_{N}$ and $L^{\theta}_{N}$ is a closed subspace of $\mathbb{H}_{\theta}$. We obtain the orthogonal basis functions $\{ e_{i}\}_{i=1}^{N}$ in $L^{\theta}_{N}$ which are specified by

\[(e_{1}, \ldots, e_{N})^\top{} =   K_{\nu,\theta}(\cdot, X_{N})K_{\nu,\theta}^{-1 / 2}(X_{N}, X_{N})  .\]

It is seen that $\{ e_{i}\}_{i=1}^{N}$ are orthogonal by taking the inner product
\[\langle K_{\nu,\theta}(\cdot, X_{N}) K_{\nu,\theta}^{-1/2}(X_{N}, X_{N}), K_{\nu,\theta}(\cdot, X_{N})K_{\nu,\theta}^{-1/2}(X_{N}, X_{N})   \rangle_{\mathbb{H}_{\theta}}= I_{N} \]
where $I_{N}$ denotes the $N$ dimensional identity matrix. Since $L^{\theta}_{N}$ is a closed subspace of $\mathbb{H}_{\theta}$, the decomposition $\mathbb{H}_{\theta} = L^{\theta}_{N} \oplus {L^{\theta}_{N}}^{c}$ holds. Here, $A^{c}$ denotes the complement of an set $A$. We define the projection operator $P_{N}$ onto $L^{\theta}_{N}$ for $y \in \mathbb{H}_{\theta}$ as 
\begin{eqnarray*}
\label{proj}
P_{N}y = \sum_{i=1}^{N} \langle y, e_{i} \rangle_{\mathbb{H}_{\theta}} e_{i} = K_{\nu,\theta}(\cdot, X_{N})K_{\nu,\theta}^{-1}(X_{N}, X_{N}) \langle y, K_{\nu,\theta}(\cdot, X_{N}) \rangle_{\mathbb{H}_{\theta}}
\end{eqnarray*}
 As $N \to \infty$, $L_{N}^{\theta}$ itself coincides with $\mathbb{H}_{\theta}$ because $\mathbb{H}_{\theta}$ is separable and $\text{Span} \{K_{\nu,\theta}(\cdot, x_{i}) | 1 \leq i  \leq \infty\}$ is dense in $\mathbb{H}_{\theta}$. As a result, $P_{N}$ tends to the identity map $I: \mathbb{H}_{\theta} \to \mathbb{H}_{\theta}$ in the operator norm since 
\begin{eqnarray*}
\label{ope_norm}
\sup_{\|y \|_{\mathbb{H}_{\theta}} =1} \| (I-P_{N})y \|_{\mathbb{H}_{\theta} } =  
 \sup_{\|y \|_{\mathbb{H}_{\theta}} =1} \| \sum_{i=N+1}^{\infty} \langle y, e_{i} \rangle_{\mathbb{H}_{\theta}} e_{i} \|_{\mathbb{H}_{\theta}} =  \sup_{\|y \|_{\mathbb{H}_{\theta}} =1}  \Bigl( \sum_{i=N+1}^{\infty} \langle y, e_{i} \rangle^{2}_{\mathbb{H}_{\theta}} \Bigr) ^{1/2}  \to 0
 \end{eqnarray*}
as $N \to \infty$ where $\{e_{i}\}_{i=1}^{\infty}$ denote the orthogonal basis functions in $\mathbb{H}_{\theta}$. Applying Mercer's theorem (Theorem 4.6.5 in \citep{hsing2015theoretical}), $K_{\nu,\theta}(x, u) $ is represented as $K_{\nu,\theta}(x, u) =   \sum_{i=1}^{\infty} \lambda_{i}   \phi_{i}(x)  \phi_{i}(u)$ with eigenvalues $\{\lambda_{i}\}_{i=1}^{\infty}$ and eigenfunctions $\{ \phi_{i} \}_{i=1}^{\infty}$. Then, $K_{\nu,\theta}(x, X_{N})K_{\nu,\theta}^{-1}(X_{N}, X_{N}) K_{\nu,\theta}(X_{N}, u)$ becomes
\begin{eqnarray*}
K_{\nu,\theta}(x, X_{N})K_{\nu,\theta}^{-1}(X_{N}, X_{N}) K_{\nu,\theta}(X_{N}, u) &=& K_{\nu,\theta}(x, X_{N})K_{\nu,\theta}^{-1}(X_{N}, X_{N}) \Bigl( \sum_{i=1}^{\infty} \lambda_{i} \phi_{i}(X_{N})\phi_{i}( u) \Bigr) \\
&=&  \sum_{i=1}^{\infty} \lambda_{i} \bigl(  P_{N} \phi_{i}(x) \bigr) \phi_{i}(u)
\end{eqnarray*}

where $\phi_{i}(X_{N}) =  (\phi_{i}(x_{1}), \ldots, \phi_{i}(x_{N}))^\top{}$. We can establish $P_{N} \phi_{i} \to \phi_{i}$ pointwise for every $i$ since $P_{N} \to I$ in the operator norm as $N \to \infty$. Thus, 
\begin{eqnarray*}
K_{\nu,\theta}(x, u) - K_{\nu,\theta}(x, X)K_{\nu,\theta}^{-1}(X, X) K_{\nu,\theta}(X, u) =  \sum_{i=1}^{\infty} \lambda_{i}   \phi_{i}(x)  \phi_{i}(u) -  \sum_{i=1}^{\infty} \lambda_{i} \bigl(  P_{N} \phi_{i}(x) \bigr) \phi_{i}(u) \to 0.
\end{eqnarray*}
Similarly,
\begin{eqnarray*}
m_{N,\theta}^{y}(x) = K_{\nu,\theta}(x, X_{N})K_{\nu,\theta}^{-1}(X_{N}, X_{N}) y(X_{N})&=& K_{\nu,\theta}(x, X_{N})K_{\nu,\theta}^{-1}(X_{N}, X_{N}) \langle y, K_{\nu,\theta}(\cdot, X_{N}) \rangle_{\mathbb{H}_{\theta}}  \\
&=&    P_{N} y(x) \to y(x)
\end{eqnarray*}
Regarding $\|m_{N,\theta}^{y} -y\|_{\mathbb{H}_{\theta}} \to 0$, we observe again that $P_{N}$ tends to the identity map $I$ in the operator norm. Let $\overline{y} = y/ \|y\|_{\mathbb{H}_{\theta}} $ for $y \neq 0$ and we can see that $ \|\overline{y}\|_{\mathbb{H}_{\theta}} = 1$. Then, for any $y \in \mathbb{H}_{\theta}$ $(y \neq 0)$, the following holds. 
\begin{eqnarray*}
\|m_{N,\theta}^{y} -y\|_{\mathbb{H}_{\theta}} &=& \|(P_{N} -I)y\|_{\mathbb{H}_{\theta}} = \|y\|_{\mathbb{H}_{\theta}} \cdot \|(P_{N} -I) \overline{y}\|_{\mathbb{H}_{\theta}}  \\
&\leq& \|y\|_{\mathbb{H}_{\theta}} \cdot \sup_{\|\overline{y}\|_{\mathbb{H}_{\theta}}  =1}  \| (P_{N}-I) \overline{y} \|_{\mathbb{H}_{\theta}} \to 0
\end{eqnarray*}
For $y=0$, this is trivial since $m_{N,\theta}^{y}$ becomes 0.
\end{proof}

We present Mercer's theorem in the following (Theorem 4.6.5 in \citep{hsing2015theoretical}) without proof which is used in the above proof .

\begin{thm}[\cite{hsing2015theoretical}, Theorem 4.6.5]
\label{mercer}
 Let the continuous kernel $K$ be symmetric and nonnegative definite and $\mathcal{K}$ the corresponding integral operator. If $\left(\lambda_{j}, e_{j}\right)$ are the eigenvalue and eigenfunction pairs of $\mathcal{K}$, then $K$ has the representation
$K(s, t)=\sum_{j=1}^{\infty} \lambda_{j} e_{j}(s) e_{j}(t)$
for all $s, t,$ with the sum converging absolutely and uniformly.
\end{thm}

\begin{remark}
\label{dif_theta}
Assuming the same conditions in Theorem \ref{conv}, let $\{\hat{\theta}_{N}\}_{N=1}^{\infty}$ be a convergent sequence in $\Theta$ with a fixed limit $\theta_{0} \in \Theta$ and $g_{N}(\theta) \coloneqq m_{N,\theta}^{y}\left(\boldsymbol{x}\right) $. We further assume that $\sup_{\theta \in \Theta} \sum_{i=1}^{2}|\frac{\partial g_{N}}{\partial \theta_{i}}( \theta) | < M$ for every $N$ and some finite $M$ independent of $N$ and $\{X_{N},y(X_{N})\}$ where $\theta_{1} = \lambda$ and $\theta_{2} = \sigma^{2}$. Based on these above assumptions and for any $\boldsymbol{x} \in \mathbb{D}$, we can establish $|y(\boldsymbol{x})  -m_{N,\hat{\theta}_{N}}^{y}\left(\boldsymbol{x}\right)| \to 0$. To verify this, first we see 
\begin{eqnarray*}|y(\boldsymbol{x}) -m_{N,\hat{\theta}_{N}}^{y}\left(\boldsymbol{x}\right)| \le |m_{N,\theta_{0}}^{y}\left(\boldsymbol{x}\right) - m_{N,\hat{\theta}_{N}}^{y}\left(\boldsymbol{x}\right)| + |y(\boldsymbol{x}) - m_{N,{\theta}_{0}}^{y}\left(\boldsymbol{x}\right)|
\end{eqnarray*}

and the second term converges to $0$ by applying Theorem \ref{conv}. As for the first term, 
\[|m_{N,\hat{\theta}_{N}}^{y}\left(\boldsymbol{x}\right) - m_{N,\theta_{0}}^{y}\left(\boldsymbol{x}\right) | = |g_{N}(\hat{\theta}_{N}) - g_{N}(\theta_{0})| < M \|\hat{\theta}_{N} - \theta_{0}\|_{2} \to 0\]
since $M$ is independent of $N$ and $\{X_{N},y(X_{N})\}$ and $\hat{\theta}_{N} \to \theta_{0}$. 

If $\hat{\theta}_{N} $ converges to $\theta_{0}$ in probability, the convergence of the first term is also in probability. If $\hat{\theta}_{N}$ is determined by maximum likelihood method, assuming the convergence in probability is plausible.

The condition on $g_{N}(\theta)$ is partially verified by taking note that \begin{eqnarray*}
\frac{\partial g_{N}}{\partial \theta_{i}}( \theta) =&\Bigl(&\frac{\partial k_{\nu,\theta}}{\partial \theta_{i}}(x,X_{N}) K_{\nu, \theta}(X_{N},X_{N})^{-1} \\
&-& K_{\nu, \theta}\left(x, X_{N}\right)K_{\nu, \theta}(X_{N},X_{N})^{-1} \frac{\partial k_{\nu,\theta}}{\partial \theta_{i}}(X_{N},X_{N})K_{\nu, \theta}(X_{N},X_{N})^{-1}\Bigr)y(X_{N})
\end{eqnarray*}
is continuous in $\Theta$ because of the form of Mat\'ern kernel and the assumption that $K_{\nu, \theta}(X_{N},X_{N})^{-1}$ exists hence $\sup_{\theta \in \Theta} \sum_{i=1}^{2}|\frac{\partial g_{N}}{\partial \theta_{i}}( \theta) | < M_{N,X_{N},y(X_{N})}$. Moreover, if $\sup_{\theta \in \Theta} \sum_{i=1}^{2}|\frac{\partial g_{N}}{\partial \theta_{i}}( \theta) | < M$ does not hold, it implies that $\sup_{\theta \in \Theta} |\frac{\partial g_{N}}{\partial \theta_{i}}( \theta) |$ is diverging as the number of samples grows. $\sup_{\theta \in \Theta} \sum_{i=1}^{2}|\frac{\partial g_{N}}{\partial \theta_{i}}( \theta) | < M$ may be usually implicitly assumed for the numerical stability in implementing the hyperparameter estimation typically maximum likelihood method with the gradient. 
\end{remark}
\section{The detail of the proof of Theorem 3}
\label{appendix_inequality}
\begin{proof}

Recall that we want to confirm $h_{X_{l,N_{l}}}^{\nu_{l}} \bigl( \log(1/h_{X_{l,N_{l}}}) \bigr)^{1/2} \simeq N_{l}^{\frac{-\nu_{l}}{d}} ( \log N_{l} )^{1/2} $. Since $h_{X_{l,N_{l}}} \simeq N_{l}^{\frac{-1}{d}}$, there exists positive constants $\tilde{c}_{1}, \tilde{c}_{2}$ such that 
\[\tilde{c}_{2} N_{l}^{\frac{1}{d}} \leq  \frac{1}{h_{X_{l,N_{l}}}} \leq \tilde{c}_{1} N_{l}^{\frac{1}{d}}\]. 

Hence, $\bigl(\log \tilde{c}_{2} N_{l}^{\frac{1}{d}} \bigr)^{1/2} \leq  \bigl(\log(1/h_{X_{l,N_{l}}}) \bigr)^{1/2} \leq \bigl(\log \tilde{c}_{1} N_{l}^{\frac{1}{d}} \bigr)^{1/2} $ holds. Then, $h_{X_{l,N_{l}}}^{\nu_{l}} \bigl( \log(1/h_{X_{l,N_{l}}}) \bigr)^{1/2}$ has the following upper and lower bounds. 
\begin{eqnarray*}
 h_{X_{l,N_{l}}}^{\nu_{l}} \bigl(\log \tilde{c}_{2} N_{l}^{\frac{1}{d}} \bigr)^{1/2} \leq h_{X_{l,N_{l}}}^{\nu_{l}}\bigl(\log(1/h_{X_{l,N_{l}}}) \bigr)^{1/2}\leq h_{X_{l,N_{l}}}^{\nu_{l}} \bigl(\log \tilde{c}_{1} N_{l}^{\frac{1}{d}} \bigr)^{1/2}
\end{eqnarray*}
We divide this inequalities by $N_{l}^{\frac{-\nu_{l}}{d}} ( \log N_{l} )^{\frac{1}{2}}$ and obtain the following:
\begin{eqnarray}
\label{inequality_target}
 \frac{h_{X_{l,N_{l}}}^{\nu_{l}}\bigl(\log \tilde{c}_{2} +  \frac{1}{d}\log N_{l} \bigr)^{1/2} }{N_{l}^{\frac{-\nu_{l}}{d}} (\log N_{l})^{\frac{1}{2}}  } \leq \frac{h_{X_{l,N_{l}}}^{\nu_{l}} \bigl(\log(1/h_{X_{l,N_{l}}}) \bigr)^{1/2}}{N_{l}^{\frac{-\nu_{l}}{d}} (\log N_{l})^{\frac{1}{2}} } \leq \frac{h_{X_{l,N_{l}}}^{\nu_{l}} \bigl(\log \tilde{c}_{1} + \frac{1}{d}\log N_{l} \bigr)^{1/2}}{N_{l}^{\frac{-\nu_{l}}{d}} (\log N_{l})^{\frac{1}{2}}  }
\end{eqnarray}
For the upper bound of this inequality, it has the following upper and lower bounds.
\begin{eqnarray*}
\frac{N_{l}^{-\frac{\nu_{l}}{d}}}{\tilde{c}_{1}} \frac{\bigl(\log \tilde{c}_{1} + \frac{1}{d}\log N_{l} \bigr)^{1/2}}{N_{l}^{\frac{-\nu_{l}}{d}} (\log N_{l})^{\frac{1}{2}}  }  \leq \frac{h_{X_{l,N_{l}}}^{\nu_{l}} \bigl(\log \tilde{c}_{1} + \frac{1}{d}\log N_{l} \bigr)^{1/2}}{N_{l}^{\frac{-\nu_{l}}{d}} (\log N_{l})^{\frac{1}{2}}  } \leq \frac{N_{l}^{-\frac{\nu_{l}}{d} }}{\tilde{c}_{2}} \frac{ \bigl(\log \tilde{c}_{1} + \frac{1}{d}\log N_{l} \bigr)^{1/2}}{N_{l}^{\frac{-\nu_{l}}{d}} (\log N_{l})^{\frac{1}{2}}  }   \\
\frac{1}{\tilde{c}_{1}} \frac{\bigl(\log \tilde{c}_{1} + \frac{1}{d}\log N_{l} \bigr)^{1/2}}{ (\log N_{l})^{\frac{1}{2}}  }  \leq \frac{h_{X_{l,N_{l}}}^{\nu_{l}} \bigl(\log \tilde{c}_{1} + \frac{1}{d}\log N_{l} \bigr)^{1/2}}{ N_{l}^{\frac{-\nu_{l}}{d}} (\log N_{l})^{\frac{1}{2}}  } \leq \frac{1}{\tilde{c}_{2}} \frac{ \bigl(\log \tilde{c}_{1} + \frac{1}{d}\log N_{l} \bigr)^{1/2}}{(\log N_{l})^{\frac{1}{2}}  }  
\end{eqnarray*}
By taking note that for $N_{l} \geq 2$, the following inequality holds.
\[\Bigl( \frac{1}{d} \Bigr) ^{1/2} \leq \Bigl( \frac{\log \tilde{c}_{1} + \frac{1}{d} \log N_{1}}{\log N_{l}} \Bigr) ^{1/2} \leq \Bigl( \frac{\log \tilde{c}_{1}}{\log 2}  + \frac{1}{d} \Bigr) ^{1/2}\]
Then, the upper bound of \ref{inequality_target}  is bounded in the following inequality.
\begin{eqnarray}
\label{kekka_ineq_1}
 \frac{h_{X_{l,N_{l}}}^{\nu_{l}} \bigl(\log \tilde{c}_{1} + \frac{1}{d}\log N_{l} \bigr)^{1/2}}{N_{l}^{\frac{-\nu_{l}}{d}} (\log N_{l})^{\frac{1}{2}}  } \leq  \frac{1}{\tilde{c}_{2}} \Bigl( \frac{\log \tilde{c}_{1}}{\log 2}  + \frac{1}{d} \Bigr) ^{1/2}
\end{eqnarray}
Similarly, the lower bound of \ref{inequality_target} has the following inequality.
\begin{eqnarray}
\label{kekka_ineq_2}
\frac{1}{\tilde{c}_{1}} \Bigl(\frac{1}{d} \Bigr) ^{1/2}  \leq \frac{h_{X_{l,N_{l}}}^{\nu_{l}}\bigl(\log \tilde{c}_{2} +  \frac{1}{d}\log N_{l} \bigr)^{1/2} }{N_{l}^{\frac{-\nu_{l}}{d}} (\log N_{l})^{\frac{1}{2}}  }
\end{eqnarray}
Combining \ref{kekka_ineq_1} and \ref{kekka_ineq_2}, we confirm the conclusion:
\begin{eqnarray*}
 \frac{1}{\tilde{c}_{1}} \Bigl(\frac{1}{d} \Bigr) ^{1/2} \leq \frac{h_{X_{l,N_{l}}}^{\nu_{l}} \log(1/h_{X_{l,N_{l}}})}{N_{l}^{\frac{-\nu_{l}}{d}} (\log N_{l})^{\frac{1}{2}} } \leq \frac{1}{\tilde{c}_{2}} \Bigl( \frac{\log \tilde{c}_{1}}{\log 2}  + \frac{1}{d} \Bigr) ^{1/2}.
\end{eqnarray*}
\end{proof}

\end{document}